%% file: main.tex
\documentclass[10pt,journal,compsoc]{IEEEtran}

\usepackage{multirow}
\usepackage{booktabs}
\usepackage{subfigure}
\usepackage{amsmath}
\usepackage{amssymb}
\usepackage{amsthm}
\usepackage{graphicx}
\usepackage{enumitem}
\usepackage{algpseudocode}
\usepackage[ruled]{algorithm2e}
\newtheorem{thm}{Theorem}
\usepackage{bm} 
\usepackage{xcolor}
\usepackage{makecell}
\usepackage{colortbl}

\newtheorem{definition}{Definition}

\SetAlFnt{\small}
\SetAlCapFnt{\small}
\SetAlCapNameFnt{\small}
\SetAlCapHSkip{0pt}
\IncMargin{-\parindent}

\ifCLASSOPTIONcompsoc

  \usepackage[nocompress]{cite}
\else
  \usepackage{cite}
\fi

\usepackage{url}

\begin{document}

\title{\textsf{VeriDIP}: Verifying Ownership of Deep Neural Networks through Privacy Leakage Fingerprints}

\author{Aoting Hu, Zhigang Lu, Renjie Xie, Minhui Xue
\IEEEcompsocitemizethanks{\IEEEcompsocthanksitem Aoting Hu is with the School of Electrical and Information Engineering, Anhui University of Technology. Email: aotinghu@ahut.edu.cn.
\IEEEcompsocthanksitem Zhigang Lu is with the College of Science and Engineering, James Cook University, Australia. Major work was done when he was a postdoctoral research fellow at Macquarie University. Email: zhigang.lu@jcu.edu.au.
\IEEEcompsocthanksitem Renjie Xie is with the National Mobile Communications Research Laboratory, Southeast University. Email: renjie\_xie@seu.edu.cn.
\IEEEcompsocthanksitem Minhui Xue is with CSIRO's Data61, Australia. Email: jason.xue@data61.csiro.au.

}
}

\markboth{Journal of \LaTeX\ Class Files,~Vol.~14, No.~8, August~2015}%
{Shell \MakeLowercase{\textit{et al.}}: Machine Learning Model Ownership Proof Leveraging Membership Inference Advantadges}

\IEEEtitleabstractindextext{%
\begin{abstract}

Deploying Machine Learning as a Service gives rise to model plagiarism, leading to copyright infringement. Ownership testing techniques are designed to identify model fingerprints for verifying plagiarism. However, previous works often rely on overfitting or robustness features as fingerprints, lacking theoretical guarantees and exhibiting under-performance on generalized models. In this paper, we propose a novel ownership testing method called VeriDIP, which \underline{veri}fies a \underline{D}NN model's \underline{i}ntellectual \underline{p}roperty. VeriDIP makes two major contributions. (1) It utilizes membership inference attacks to estimate the lower bound of privacy leakage, which reflects the fingerprint of a given model. The privacy leakage fingerprints highlight the unique patterns through which the models memorize sensitive training datasets. (2) We introduce a novel approach using less private samples to enhance the performance of ownership testing.

Extensive experimental results confirm that VeriDIP is effective and efficient in validating the ownership of deep learning models trained on both image and tabular datasets. VeriDIP achieves comparable performance to state-of-the-art methods on image datasets while significantly reducing computation and communication costs. Enhanced VeriDIP demonstrates superior verification performance on generalized deep learning models, particularly on table-trained models. Additionally, VeriDIP exhibits similar effectiveness on utility-preserving differentially private models compared to non-differentially private baselines.

\end{abstract}

}

\maketitle

\IEEEdisplaynontitleabstractindextext

\IEEEpeerreviewmaketitle

\ifCLASSOPTIONcompsoc
\IEEEraisesectionheading{\section{Introduction}\label{sec:introduction}}
\else
\section{Manuscript}
\label{sec:introduction}
\fi

\input{intro}

\input{related}

\input{probformu}

\input{VeriDIP}

\input{eval}

\input{conclusion}

\ifCLASSOPTIONcaptionsoff
  \newpage
\fi

\bibliography{references.bib}
\bibliographystyle{./IEEEtran.bst}

\begin{IEEEbiography}
[{\includegraphics[width=1in,height=1.25in,clip,keepaspectratio]{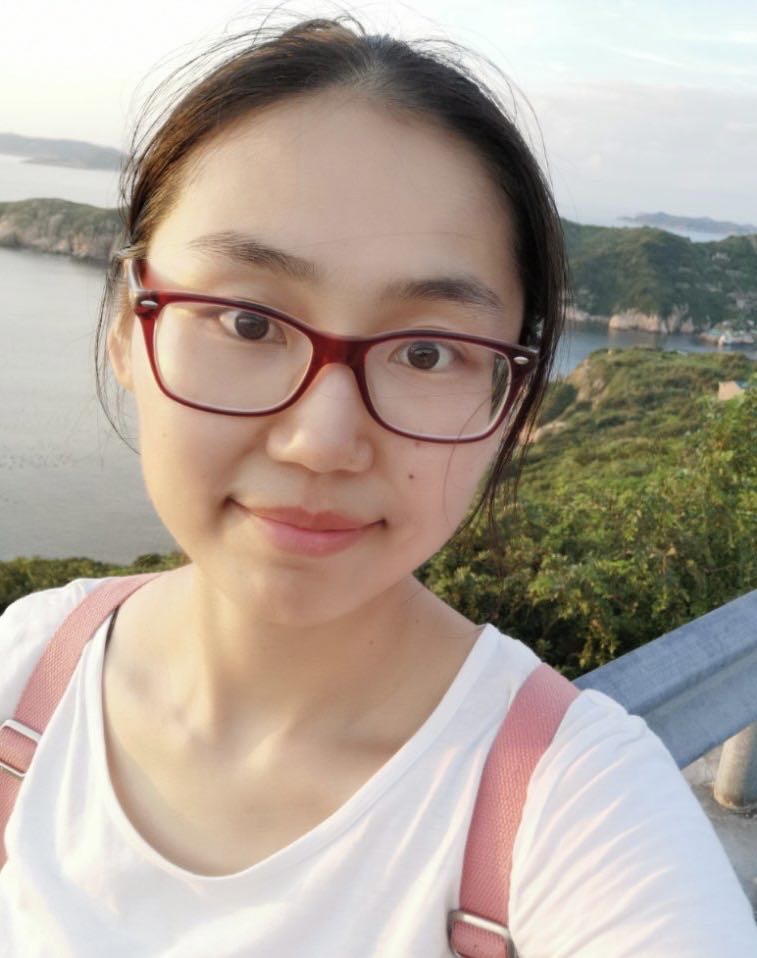}}]{Aoting Hu}  is a Lecturer with Anhui University of Technology, China. She received her B.Sc. degree in communication engineering from Anhui University of Technology in 2014. She obtained her M.Sc. degree in communication and information engineering and her Ph.D. degree in cyberscience and engineering from Southeast University, Nanjing, China. Her recent research interests include machine learning security and privacy.
\end{IEEEbiography}

\begin{IEEEbiography}
[{\includegraphics[width=1in,height=1.25in,clip,keepaspectratio]{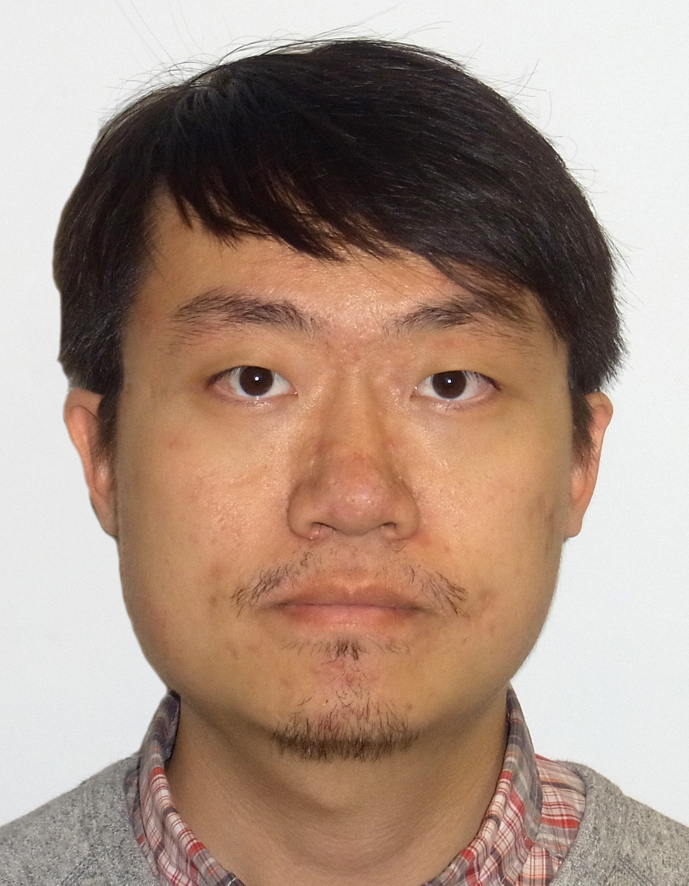}}]{Zhigang Lu} is a Lecturer with James Cook University, Australia. Prior to that, he was a Postdoctoral Research Fellow at the Macquarie University Cyber Security Hub. He received his BEng degree from Xidian University and his MPhil and PhD degrees from the University of Adelaide, all in computer science. With research interests in differential privacy and machine learning, he has published over ten papers in international journals/conferences, including IEEE TDSC, IEEE TIFS, and ACM CCS.
\end{IEEEbiography}

\begin{IEEEbiography}
[{\includegraphics[width=1in,height=1.25in,keepaspectratio]{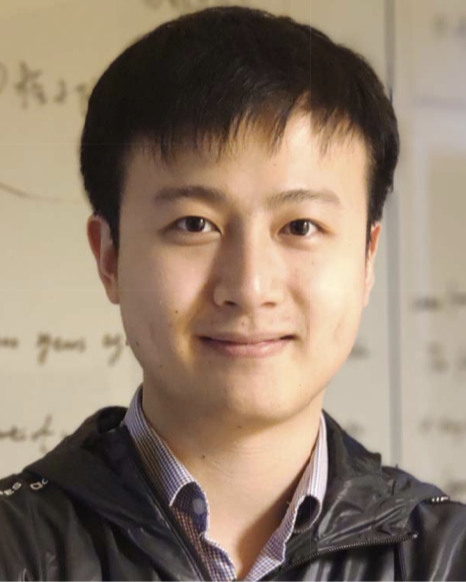}}]{Renjie Xie}~(Graduate Student Member, IEEE) is a Ph.D. candidate majoring in communication and information engineering at Southeast University, Nanjing, China. He received his B.Sc. degree in mathematics and applied mathematics from South China Agricultural University in 2015, and his M.Sc. degree in computer science from the South China University of Technology in 2018. His recent research interests include computer vision, representation learning, physical layer security, and machine learning for wireless communications.
\end{IEEEbiography} 

\begin{IEEEbiography}
[{\includegraphics[width=1in,height=1.25in,keepaspectratio]{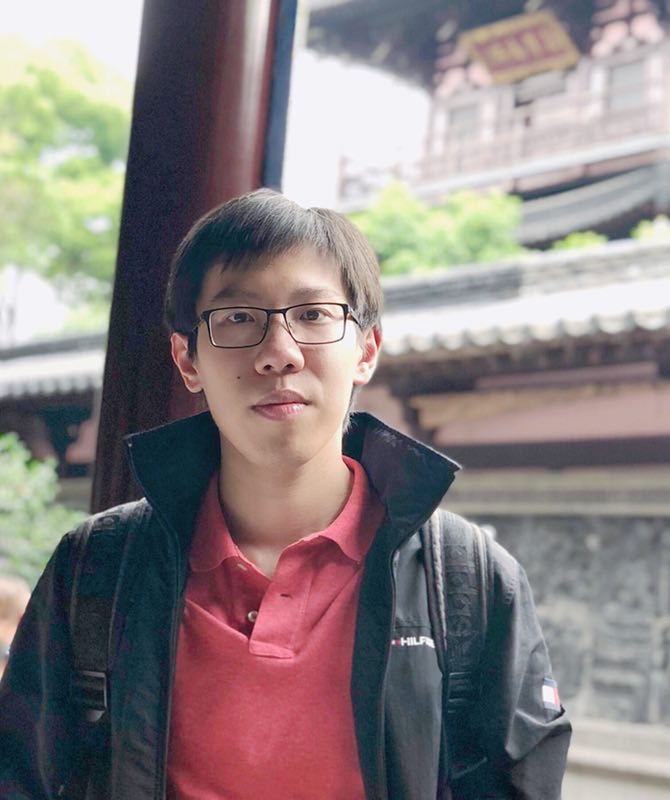}}]{Minhui Xue} is a Senior Research Scientist (lead of AI Security sub-team) at CSIRO's Data61, Australia. His current research interests are machine learning security and privacy, system and software security, and Internet measurement. He is the recipient of the ACM CCS Best Paper Award Runner-Up, ACM SIGSOFT distinguished paper award, Best Student Paper Award, and the IEEE best paper award, and his work has been featured in the mainstream press, including The New York Times, Science Daily, PR Newswire, Yahoo, The Australian Financial Review, and The Courier. He currently serves on the Program Committees of IEEE Symposium on Security and Privacy (Oakland) 2023, ACM CCS 2023, USENIX Security 2023, NDSS 2023, ACM/IEEE ICSE 2023, and ACM/IEEE FSE 2023. He is a member of both ACM and IEEE.
\end{IEEEbiography}

\end{document}

%% file: intro.tex
\IEEEPARstart{D}{eep} learning plays an important role in various tasks such as image recognition~\cite{he2016deep,krizhevsky2012imagenet,simonyan2014very}, natural language processing~\cite{goldberg2016primer}, and speech recognition~\cite{graves2013speech} tasks. Building a sophisticated deep neural network (DNN) requires a significant amount of annotated training data, which often contains user privacy, demands powerful computing resources, and necessitates machine learning expertise. These unique DNN models represent valuable intellectual property (IP) and require copyright protection. However, deploying DNN models' APIs for user queries introduces the risk of model extraction attacks, leading to copyright infringement~\cite{tramer2016stealing,papernot_practical_2017,orekondy_knockoff_2019}. Model extraction attack efficiently transfers the functionality of a \emph{victim model} to a \emph{stolen copy} using limited query answers. Additionally, attackers, who may be insiders with full access to the victim models, employ techniques such as distillation~\cite{yang_effectiveness_2019}, fine-tuning~\cite{chen_refit_2021}, or pruning~\cite{bailey_finepruning_2018,liu_rethinking_2019} in an attempt to reverse-engineer the tracking. 

Proof-of-ownership serves as an adequate protection mechanism against model stealing attacks, ensuring accountability for any theft of copyright-protected models. However, proving ownership of a neural network poses challenges due to the stochastic nature of the training and stealing process~\cite{KingmaB14}. Many stealing mechanisms have minimal side effects on the model's functionality but disable the proof-of-ownership mechanism~\cite{zhang_protecting_2018,yang_effectiveness_2019, chen2021copy}. Methods for proving ownership of DNN models can be broadly classified into two categories: \textbf{watermark embedding (WE)}~\cite{Yusuke2017Embedding,rouhani2018deepsigns,le_merrer_adversarial_2020,li_spread_2021,Enzo_Delving_2021,liu_Watermarking_2021,ChenRFZK19} and \textbf{ownership testing (OT)} ~\cite{maini2021dataset,chen2021copy,cao_ipguard_2021,lukas_deep_2021}. 

\begin{figure}[!t]
    \centering  
    \includegraphics[width=0.95\columnwidth]{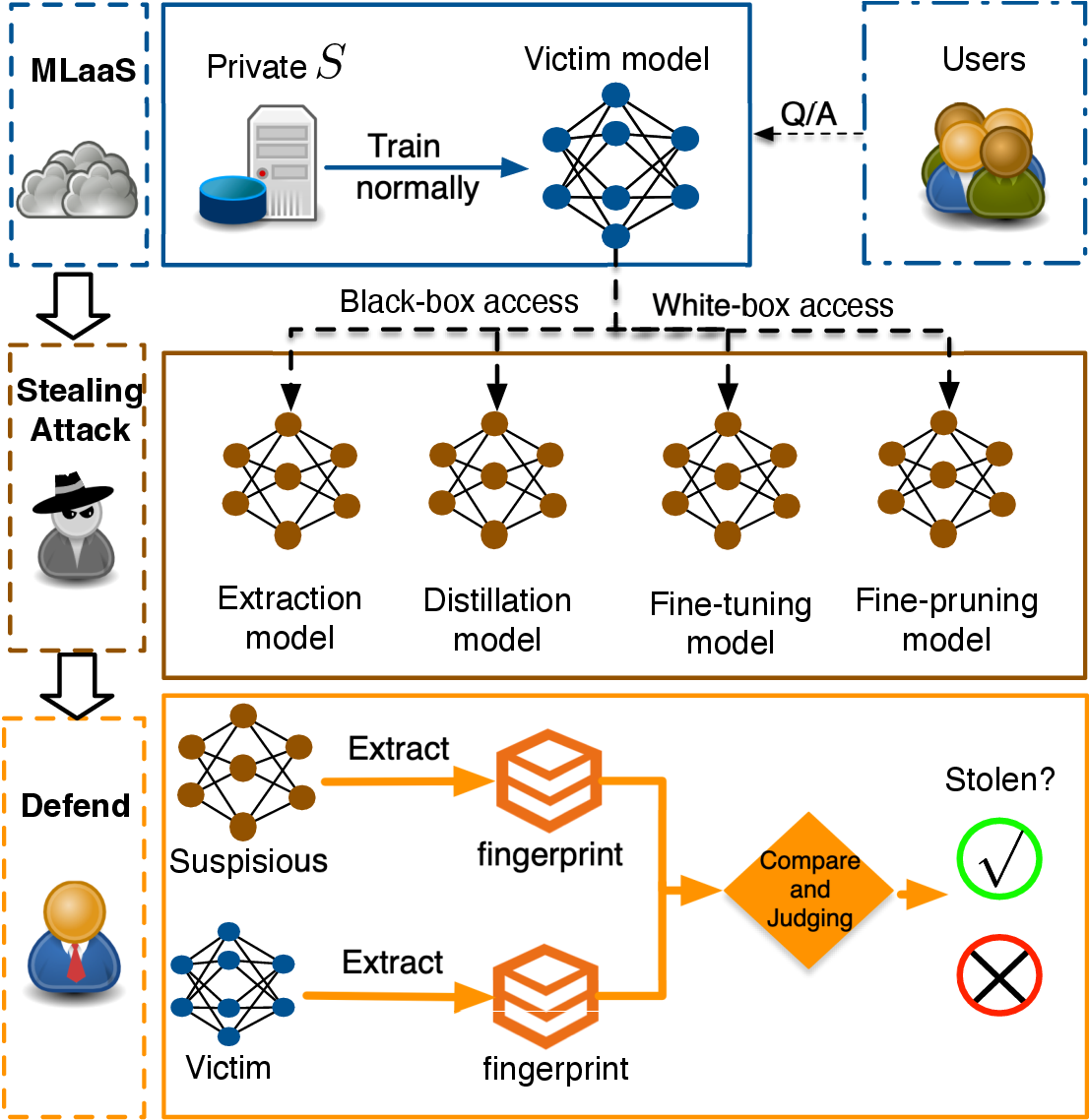}
     \vspace{-2mm}
    \caption{Ownership testing framework for DNN models.}
    \label{fig:ot}
\end{figure}

The WE methods embed customized watermarks into DNN models during the training stage then verify the ownership by confirming the presence of the respective watermarks from given suspect models. However, WE techniques have certain limitations, including tampering with the training process, potential side effects on model functionality, and vulnerability to watermark erasure attacks~\cite{yang_effectiveness_2019,jia_entangled_2020,chen_refit_2021}. In contrast, the OT methods extract the intrinsic characteristics (fingerprints) of DNN models, making them non-invasive and more resilient to adaptive attacks~\cite{maini2021dataset,chen2021copy}. In this paper, our focus is on the OT technique to verify the copyright of DNN models.

To the best of our knowledge, existing ownership testing solutions rely on two types of DNN fingerprints 
--- \emph{model robustness} and \emph{model overfitting}, which capture the uniqueness of DNN models. 
Robustness-based solutions utilize adversarial examples to delineate the decision boundary of 
both the victim model and its stolen copies, and then compare the percentage of matched answers~\cite{cao_ipguard_2021,lukas_deep_2021,chen2021copy}. 
However, techniques that enhance a DNN model's robustness against adversarial attacks, such as adversarial training~\cite{madry_towards_2019}, undermine the performance of ownership testing. 
On the other side, overfitting-based OT solutions, such as dataset inference~\cite{maini2021dataset}, leverage the observation that the stolen copies exhibit a higher level of overfitting to the training set of the victim models, thereby extracting the overfitting level as fingerprints. 
While these approaches are innovative and effective, they have certain limitations. The verification process is communicational and computationally expensive requiring thousands of queries to the stolen copy to obtain dozens of minimal adversarial noise as fingerprints~\cite{maini2021dataset} 
Continuous inquiries may raise suspicions of model theft and result in rejection of the inquiries~\cite{hitaj_have_2018}. 
Furthermore, the performance of overfitting-based solutions is negatively affected by the model's generalization ability.

To address these problems, we propose a novel ownership testing approach to \underline{Veri}fy a \underline{D}NN model's \underline{I}ntelligence \underline{P}roperty (\textsf{VeriDIP}). The key feature of VeriDIP is its utilization of \textbf{privacy leakage} fingerprints, instead of relying on overfitting~\cite{maini2021dataset} or robustness~\cite{cao_ipguard_2021,lukas_deep_2021,chen2021copy} metrics to indicate model uniqueness. Drawing on the concept of membership inference (MI) attacks from previous works~\cite{shokri2017membership,sablayrolles2019white,carlini_membership_2022}, the privacy leakage of a model against MI attacks reflects the extent to which the model has memorized its private or secret training data. Hence, considering the secrecy of the training data, a stolen model would not exhibit the same level of privacy leakage on the victim's private training data as the victim model under the same MI attacks. In other words, the privacy leakage fingerprint of a model captures the distinctive and confidential patterns learned by the model, fulfilling the criteria of a reliable fingerprint: uniqueness and irremovability. As a result,
any unauthorized DNN models that result in a certain degree of privacy leakage of a private training set can be identified as plagiarized.

Using privacy leakage fingerprints, VeriDIP consists of four components for verifying a DNN model's intelligence property. First, motivated by the aforementioned properties of the privacy leakage of a given model, we utilize MI attacks to estimate the lower bound of privacy leakage, which serves as the extracted fingerprint of a given model. Then we employ hypothesis testing on the extracted fingerprint to determine the likelihood of a suspect model being a stolen copy of the victim model. However, we may encounter the issue of ``fingerprint fading" when dealing with well-generalized models that exhibit minimal privacy leakage against MI attacks. To tackle this problem, 
we introduce an enhanced version of VeriDIP where MI attacks query the suspect models using less private samples to extract the worst-case privacy leakage fingerprints of the suspect models. These less private samples face higher privacy leaking risks against MI attacks, enabling the enhanced VeriDIP to extract stronger privacy leakage fingerprints.
To identify the less private data in advance, we train numerous shadow models to investigate the impact of each training sample on the decision boundary of DNN models. The data that significantly influences the models will be considered as the less private data.

We extensively evaluate \textsf{VeriDIP} on two image datasets (FMNIST and CIFAR-10) and two tabular datasets (Adult and Health) against three types of model stealing attacks: model extraction attack, model distillation attack, and fine-tune attack. The evaluation results for FMNIST and CIFAR demonstrate that VeriDIP can publicly authenticate all stolen copies while exposing less than 5 training samples, with a significantly reduced number of queries to the suspect models compared to~\cite{maini2021dataset}. Despite the models trained on tabular datasets having minimal overfitting, VeriDIP is still capable of publicly authenticating all stolen copies, at the cost of exposing dozens of training samples, whereas previous works~\cite{maini2021dataset,cao_ipguard_2021,lukas_deep_2021,chen2021copy} are unable to do so.

In this work, we also address an open question raised in~\cite{maini2021dataset} regarding the effectiveness of VeriDIP on differentially private DNN models. We demonstrate that VeriDIP's success rate is constrained by a stringent privacy budget, such as $\varepsilon = 0.1$. However, we find that VeriDIP remains effective even for utility-preserving differentially private models, such as those with a higher privacy budget, e.g., $\varepsilon = 0.5$.

To summarize, our contributions are as follows:
\begin{itemize}[leftmargin=*]
	\item We propose VeriDIP, a model ownership testing (OT) approach for DNN models. VeriDIP utilizes the membership inference (MI) attack to estimate the privacy leakage of DNN models, which serves as the fingerprint of a given (victim/target) model. 
 
	\item We further enhance VeriDIP by utilizing less private samples to estimate the worst-case privacy leakage, thereby strengthening the extracted fingerprints of DNN models.

    \item We perform extensive evaluations on VeriDIP using various DNN models trained with tabular or image benchmarks, against three types of model stealing attacks. The results show that VeriDIP can publicly authenticate all stolen copies with minimal verification costs.
    
	\item We theoretically and experimentally analyze the connection between the effectiveness of VeriDIP and differential privacy (DP) privacy protection. The results demonstrate that as long as a DP model is utility-preserving, VeriDIP can effectively protect its copyright. 
\end{itemize}

%% file: related.tex
\section{Related Work}
In this section, we review model stealing attacks, well-known ownership testing techniques and membership inference attacks. We list the comparison of different copyright protection methods for DNN models in Table~\ref{tab:comparison}.

\begin{table*}[t]
    \centering
    \caption{Comparison of different DNN model copyright protection methods. ME: model extraction attack; KD: knowledge distillation attack; FT: fine-tuning attack; ADV: adversarial; DP: differential privacy. ME, KD, and FT are model stealing attacks. Adaptive attacks aim to weaken the effect of ownership test approaches.}
    \label{tab:comparison}
    \vspace{-3mm}
    \begin{tabular}{l|cccccccc}
        \toprule
        \multirow{2}{*}{Approaches}  &\multirow{2}{*}{Type} &\multirow{2}{*}{Method} &\multirow{2}{*}{\makecell{Non-\\invasive}}& \multirow{2}{*}{\makecell{DP \\connection}}& \multicolumn{3}{c}{Model stealing attacks} & \multirow{2}{*}{Adaptive attacks}\\ \cmidrule(r){6-8} &&&&& ME & KD & FT \\
        \midrule
        Adi et al.\cite{adi_turning_nodate} & watermarking & backdoor & $\times$&N/A &$\times$~\cite{shafieinejad_robustness_2021}&$\times$\cite{yang_effectiveness_2019}&$\times$\cite{chen_refit_2021}&N/A\\ 
        Zhang el al.~\cite{zhang_protecting_2018} & watermarking & backdoor & $\times$&N/A &$\times$~\cite{shafieinejad_robustness_2021}&$\times$~\cite{yang_effectiveness_2019}&$\times$~\cite{chen_refit_2021}&N/A\\
        Chen et al.~\cite{chen2021copy}& fingerprinting & robustness & $\surd$&N/A  &$\times$~\cite{chen2021copy} &$\times$~\cite{chen2021copy}& $\surd$ & ADV training~\cite{madry_towards_2019}\\
        Cao et al.~\cite{cao_ipguard_2021}& fingerprinting & robustness & $\surd$&N/A  &$\times$~\cite{lukas_deep_2021} &$\times$~\cite{lukas_deep_2021}&$\surd$ & ADV training~\cite{madry_towards_2019}\\
        Lukas et al.~\cite{lukas_deep_2021}& fingerprinting & robustness & $\surd$ &N/A &$\surd$&$\surd$ &$\surd$& ADV training~\cite{madry_towards_2019}\\
        Maini et al.~\cite{maini2021dataset}& fingerprinting & over-fitting & $\surd$ &N/A&$\surd$ &$\surd$&$\surd$ &detector attacks~\cite {hitaj_have_2018} \\
        \rowcolor{gray!20}  VeriDIP (This work) & fingerprinting & privacy leakage & $\surd$& $\surd$& $\surd$& $\surd$ &$\surd$ & secure for now \\
        \bottomrule
    \end{tabular}  
\end{table*}

\subsection{Model stealing attacks}
\textbf{Black-box attacks.} Tramer et al.~\cite{tramer2016stealing} proposed the first model extraction attack that trains a stolen copy using the predictions of victim models. It requires black-box access to the victim model and some unlabeled datasets from the same distribution. According to Shafieinejad et al.~\cite{shafieinejad_robustness_2021}, existing watermark embedding techniques~\cite{adi_turning_nodate,zhang_protecting_2018} and some fingerprinting solutions~\cite{chen_refit_2021,cao_ipguard_2021} cannot withstand model extraction attacks. Distillation~\cite{hinton_distilling_2015} was first proposed to distill the knowledge of teacher models into student models and later extended as an attack against methods that protect model copyrights~\cite{yang_effectiveness_2019}. Distilled models are often able to evade copyright tracking, as demonstrated in works such as Cao et al.~\cite{cao_ipguard_2021} and Lukas et al.~\cite{lukas_deep_2021}.

\textbf{White-box attacks.} White-box attackers have full access to victim model's parameters, and their goal is modify these parameters in order to disable copyright protection mechanisms. 
For instance, fine-pruning~\cite{bailey_finepruning_2018} is a defensive method against DNN model backdooring. It prune backdoored neurons and then fine-tuning the models. Consequently, fine-pruning could potentially be an attack against backdoor-based model watermarking techniques, such as those proposed in works like Adi et al.~\cite{adi_turning_nodate,zhang_protecting_2018}.
More recently, Chen et al.~\cite{chen_refit_2021} proposed an advanced fine-tuning technique that aims to erase model watermarks. They initially increase the learning rate to make the victim model forget unnecessary details about watermarks and then gradually restore the utility of the model by reducing the learning rate step by step. While these attacks are effective in disabling watermark embedding techniques, it remains unclear whether they pose a threat to the copyright protection provided by ownership testing methods.

\subsection{Ownership testing}
Ownership testing (OT) techniques, also referred to as DNN fingerprinting techniques, are an emerging area of research that focuses on extracting the intrinsic characteristics of DNN models to track stolen copies. Currently, the research on OT is limited, with the majority of studies relying on two fingerprint characteristics: robustness~\cite{cao_ipguard_2021,lukas_deep_2021,chen2021copy} and overfitting~\cite{maini2021dataset}.

IPGuard~\cite{cao_ipguard_2021} proposes using model robustness as fingerprints. The authors observe that stolen copies exhibit similar predictions to the victim model for most adversarial data points. While IPGuard can successfully identify white-box derivation attacks, such as fine-tuning, it is not effective against black-box extraction attacks, such as model extraction attack~\cite{shafieinejad_robustness_2021}, where the attacker retrains the model from scratch, resulting in a larger disparity in the decision surface compared to the victim model. To address this limitation, Lukas et al.~\cite{lukas_deep_2021} propose the use of transferable adversarial samples to extract DNN fingerprints. This approach successfully defends against white-box derivation attacks and most black-box extraction attacks, but it is vulnerable to transfer learning and adversarial training. More recently, Chen et al.~\cite{chen2021copy} propose a testing framework for verifying ownership. Instead of relying on single metrics, they utilize multiple dimensions and combine the results to determine ownership. Their black-box metrics also use robustness as fingerprints, similar to IPGuard~\cite{cao_ipguard_2021}, making them susceptible to black-box extraction attacks. Their white-box metrics utilize the robustness of inner neuron outputs, requiring the defender to have knowledge of all parameters of stolen copies.

Dataset inference (DI)~\cite{maini2021dataset} exploits the overfitting of DNN models to their training data as a means to demonstrate that stolen copies exhibit similar overfitting fingerprints to the victim models. They employ minimal adversarial noise that leads to model misclassification~\cite{szegedy_intriguing_2014} as fingerprints. DI is capable of identifying all white-box and black-box model variations~\cite{maini2021dataset}. However, this approach has some limitations. Firstly, it cannot be directly applied to DNN models trained on tabular data since some of the features are categorical, making it challenging to perform most adversarial example attacks~\cite{ballet_imperceptible_2019}. Secondly, DI requires querying the suspect model thousands of times, which significantly increases the risk of detector attacks~\cite{hitaj_have_2018}. Thirdly, the effectiveness of DI on differentially private (DP)~\cite{dwork2006calibrating} DNN models remains unanswered. Hence, this paper aims to propose a novel ownership testing approach that addresses these limitations by achieving high verification efficiency and protecting the intellectual property of DP models.

\subsection{Membership inference attacks}
Shokri et al. proposed the first membership inference (MI) attack in 2017~\cite{shokri2017membership}, which successfully guesses the membership of the training data with black-box access to the target DNN models. Since then, researchers have made efforts to enhance the attack performance and reduce the background information required by MI attackers. More recently, some researchers have utilized MI attacks as an empirical measurement for estimating the privacy leakage of DNN models~\cite{yeom2018privacy,sablayrolles2019white,carlini_membership_2022}. This approach has inspired us to leverage the MI advantage as a lower bound for estimating model privacy leakage and consider privacy leakage characteristics as the model fingerprint. Additionally, other studies have revealed the varying exposure risks of training data against MI attacks~\cite{feldman_does_2020,carlini_membership_2022}, which have also motivated us to extract stronger fingerprints.

%% file: probformu.tex
\begin{table}[!t]
    \centering
	\caption{Summary of Notations}
	\vspace{-3mm}
	\label{tab:variable_mimop}
	\begin{tabular}{l|p{6cm}}
		\toprule
		Notations & Description  \\
		\midrule
		$\bm{x}$ & feature vector \\ 
        $y$ & the label corresponding to $\bm{x}$\\
        $\bm{z}$ & a data point $\bm{z} = (\bm{x},y)$ \\ 
        $\alpha$ & significance level for hypothesis testing \\
        $\mathcal{D}$ & data distribution \\ 
        $S$ & private training dataset\\
        $f$ & DNN models \\
        $n_{S}$ & \makecell[l]{number of exposed samples during\\ public copyright verification}\\
        $(\epsilon,\delta)$ & DP parameters (privacy budget, failure probability)\\
       ($C$, $\mathrm{z}$) & DP hyper-parameters (clipping threshold, noise multiplier)\\
        $P$ & probability of not being a stolen model\\
        $Y$ & ownership testing outcome - Stolen or Not stolen\\
        \midrule
        $\ell (f,\bm{z})$ & Loss function, output the prediction loss of model $f$ on sample $\bm{z}$\\
        $\mathcal{V} (f,\mathcal{P}_{S},\mathcal{B})$ & OT algorithm, output whether a suspect model $f$ is trained on $S$, where $\mathcal{P}_{S}$ is an auxiliary dataset to $S$ and $\mathcal{B}$ is background knowledge about $f$ or $S$\\
		$\mathcal{A} (\bm{z},f,\mathcal{D})$ & MI attack algorithm, output whether a sample $\bm{z}$ is used to train model $f$, $\mathcal{D}$ is auxiliary information\\
        $\operatorname{Adv}^{\mathrm{M}} (\mathcal{A},f,\mathcal{D})$& Membership advantage algorithm, output membership advantages of algorithm $\mathcal{A}$ on model $f$, $\mathcal{D}$ is auxiliary information\\
		\bottomrule
	\end{tabular}%
\end{table}%

\section{Ownership Testing Problem}
In this section, we first formulate the ownership testing (OT) problem, then discuss the capabilities of adversaries and defenders, followed by the backgrounds of differential privacy and membership inference.

\subsection{Notations}
Let $\bm{z}= (\bm{x},y)$ be a data point, where $\bm{x}$ is the feature vector and $y$ is the corresponding label. $\mathcal{D}$ represents the data distribution from which $\bm{z}$ is drawn. We assume that the victim model is trained on the training set $S (\sim \mathcal{D}^{n})$ consisting of $n$ data points. The loss function $\ell (f,\bm{z})$ measures the difference between the model predictions $f (\bm{x})$ and the ground-truth label $y$. We provide a summary of the notations used in this work in Table~\ref{tab:variable_mimop}.

\subsection{Problem Formulation}
Figure~\ref{fig:ot} depicts a general framework of ownership testing (OT) for DNN models, where we have three components - machine learning as a service (MLaaS), model stealing attacks and defences. 

Particularly, MLaaS provides users with access to pre-built machine learning (DNN) models through APIs, allowing the users to integrate machine learning capabilities into their applications and perform complex tasks through simple queries.
However, to fully utilize the potential of the pre-build models, attackers might attempt to steal the models by mimicking the behaviors of regular users (querying the models through the open APIs) to infer/extract the model details.
To protect the copyright of (the victim) DNN models, an OT approach extracts and compares the fingerprints of a suspect model and the victim model to produce a test outcome, indicating whether the suspect model is a stolen copy of the victim model.

In this paper, we aim to design a model OT algorithm $\mathcal{V}$, defined as follows
\begin{equation}
\label{exp_di}
	\mathcal{V} (f,\mathcal{P}_{S},\mathcal{B}) \rightarrow \{0,1\},
\end{equation}
where $\mathcal{P}_{S}$ is an auxiliary dataset, containing carefully chosen adversarial examples~\cite{lukas_deep_2021,cao_ipguard_2021,chen2021copy} or a subset of training examples~\cite{maini2021dataset} and $\mathcal{B}$ represents the publicly available knowledge about the model~\cite{lukas_deep_2021,cao_ipguard_2021,chen2021copy} or about the private training data~\cite{maini2021dataset}.
In the algorithm $\mathcal{V} (f,\mathcal{P}_{S},\mathcal{B})$, the verifier first extracts the inherent fingerprint of the suspect model $f$ using $\mathcal{P}_{S}$ and $\mathcal{B}$, and then determines the ownership based on whether it aligns with the owner's expectations. The algorithm $\mathcal{V} (f,\mathcal{P}_{S},\mathcal{B})$ outputs $1$ when the verifier believes the suspect model $f$ is a stolen copy of the victim model $f_{S}$, and vice versa. The algorithm $\mathcal{V} (f,\mathcal{P}_{S},\mathcal{B})$ should be highly accurate, computationally and communicationally efficient, and privacy-preserving (safe to audit in public).

\subsection{Threat Model}
\label{sec:threat_model}
We specify the capabilities of the attacker and verifier (defender) shown in Figure~\ref{fig:ot}.
 
 \textbf{Attacker.} We consider a wide variety of model stealing attacks, including both black-box access and white-box access capabilities. However, the adversary does not have access to the entire (private) training set of the victim model.
\begin{itemize}[leftmargin=*]
	\item {Black-box attacker.} Attackers, who are external entities attempting to exploit the functionality of the victim model, employ various attacks such as model extraction attacks~\cite{tramer2016stealing} and model distillation attacks~\cite{yang_effectiveness_2019}. 
	\item {White-box attacker.} Attackers, who are insiders with full access rights to the victim model, aim to evade tracking and detection. They employ various attacks such as model fine-tuning~\cite{chen_refit_2021} and model fine-pruning~\cite{bailey_finepruning_2018,liu_rethinking_2019}.
\end{itemize}

\textbf{Verifier.} As for defense, our focus is on black-box verifiers who have limited query access to the suspect model. There are two main reasons for this choice.  
First, when the verifier is a third-party agency, sharing excessive information such as training data or model parameters can pose risks to the model owner or data contributors. Second, allowing an unlimited number of verification queries can potentially trigger detector attacks~\cite{hitaj_have_2018}. In a detector attack, the unauthorized model API may refuse to respond or provide random responses upon detecting an attempt to verify copyright. For example, in the work by Maini et al.~\cite{maini2021dataset}, the victim model is queried 1500 times for a single data point to collect minimal adversarial noise vectors for ownership determination, which significantly increases the likelihood of triggering a detector attack (refer to Table~\ref{tab:comparison}).

\subsection{Membership Advantage}
As we know, Yeom et al.~\cite{yeom2018privacy} show that the privacy budget of a differentially private DNN model is a lower bound of the model's privacy leakage against MI attacks. Furthermore,
as demonstrated by Yeom et al.~\cite{yeom2018privacy}, the privacy budget of a differentially private DNN model serves as a lower bound for estimating the model's privacy leakage against membership inference (MI) attacks. Additionally, recent research by Hyland et al.~\cite{Stephanie_Intrinsic_2019} highlights that not only intentionally noisy DNN models provide privacy protection, but ordinary DNN models also possess a certain level of privacy protection due to the inherent randomness introduced by stochastic gradient descent (SGD). Consequently, it becomes possible to assess the potential privacy leakage of a non-differentially private DNN model by estimating the corresponding privacy budget associated with the non-DP model.

\subsubsection{Differential Privacy}
\label{sec:DP}
Recall the definition of differential privacy~\cite{dwork2006calibrating}, A learning algorithm $f:\mathcal{D} \mapsto \mathcal{R}$ satisfies ($\epsilon,\delta$)-DP if, for all adjacent databases $D$ and $D^{\prime}$ that differs in one record, and all possible outputs $\mathcal{O}\subseteq \mathcal{R}$, the following inequality holds.
\begin{equation}
\label{eq:dp}
    \Pr[f(D)\in \mathcal{O}]\leq \exp(\epsilon) \Pr[f(D^{\prime})\in \mathcal{O} ] + \delta,
\end{equation}
where the probabilities are taken only over the randomness of the learning algorithm $f$. A greater $\epsilon$ indicates a lesser degree of privacy protection for the training data, meaning that the machine learning algorithm $f$ may potentially compromise more privacy of the sensitive database $D$.

If the verifier is able to quantify the privacy risks associated with a particular learning algorithm on a specific private training set, this value can be used as a fingerprint for identifying plagiarism. This is because the target model and its pirated version are likely to exhibit higher privacy leakage of their training data compared to independently trained models. By analyzing and comparing the privacy risks of different models, the verifier can detect potential instances of plagiarism or unauthorized use of the training data. 
However, it is noteworthy that directly estimating the value of $\epsilon$ for deployed non-DP DNN models on given datasets is intractable. This is because it would require traversing all possible adjacent datasets and evaluating all possible outputs to compute the maximum divergence. This process becomes computationally expensive and impractical, especially for large-scale datasets and complex models.

\subsubsection{Membership Inference}

Membership inference (MI) attacks~\cite{shokri2017membership} aim to predict whether a particular example is part of a training dataset. Recently, some researchers~\cite{murakonda_ml_2020,song2020introducing} have proposed utilizing MI attacks as a means to measure privacy leakage. Other works~\cite{yeom2018privacy,sablayrolles2019white} have theoretically established that the privacy leakage measured by MI attacks serves as a lower bound for $\epsilon$. In this work, we leverage the concept of membership advantage~\cite{yeom2018privacy} and utilize it as a fingerprint for our model. We provide a review of the related definition below.

Before getting into membership advantage, we first define the MI attack following  ~\cite{shokri2017membership,yeom2018privacy}. 

\begin{definition}[Membership inference experiment $\operatorname{Exp}^{\mathrm{M}} (\mathcal{A}, f_S, \mathcal{D})$] \label{ME}
Let $\mathcal{A}$ be a membership inference attack algorithm, $f_{S}$ is a machine learning model trained on $S \sim \mathcal{D}^{n}$. The procedure of the membership inference experiment is as follows:
\begin{enumerate}[leftmargin=*]
 	\item Toss a coin at random $b \leftarrow\{0,1\}$;
 	\item If $b=1$, then the sample $\bm{z}$ draws from $S$, denoted as $\bm{z} \sim S$. If $b=0$, then the sample $\bm{z}$ comes from $\mathcal{D}$, denoted as $\bm{z} \sim \mathcal{D}$;
 	\item $\{0,1\}\leftarrow \operatorname{Exp}^{\mathrm{M}} (\mathcal{A}, f_S, \mathcal{D})$. The experiment $\operatorname{Exp}^{\mathrm{M}} (\mathcal{A}, f, \mathcal{D})$ returns $1$ to represent the attacker correctly guessing the answer of $b$, denoted as $\mathcal{A}\left (\bm{z}, f_{S}, \mathcal{D}\right)=b$ and vice versa.
\end{enumerate}
\end{definition}

In Definition~\ref{ME}, the attack algorithm $\mathcal{A}\left (\bm{z}, f_{S}, \mathcal{D}\right)$ inputs arbitrary sample $\bm{z}$, model $f_{S}$, public data distribution $\mathcal{D}$, and outputs the judgment about whether the sample $\bm{z}$ is used to train model $f_{S}$.
 
Membership advantage~\cite{yeom2018privacy} represents the advantage of an MI attacker's ability to guess the decision boundary of training samples and other samples over random guess.
 
\begin{definition}[Membership Advantage]\label{def:advM}
The advantage of the MI attack algorithm $\mathcal{A}$ is defined as 
\begin{equation}
\label{eq:miadv}
    \operatorname{Adv}^{\mathrm{M}} (\mathcal{A}, f, \mathcal{D})=2 \operatorname{Pr}\left[\operatorname{Exp}^{\mathrm{M}} (\mathcal{A}, f, \mathcal{D})=1\right]-1.
\end{equation}
\end{definition}
 
Membership advantage ranges from $0$ to $1$, where $0$ indicates no advantage (equivalent to random guessing), and $1$ represents a full advantage. The right-hand side of Equation~\eqref{eq:miadv} can be empirically determined by computing the difference between the true positive rate (TPR) and the false positive rate (FPR) of the attack algorithm $\mathcal{A}$. That is, 
\begin{equation}
\label{eq:MI_adv}
    \begin{aligned}
        \operatorname{Adv}^{\mathrm{M}} (\mathcal{A}, f, \mathcal{D})  &=\operatorname{Pr}[\mathcal{A}=1 \mid b=1]-\operatorname{Pr}[\mathcal{A}=1 \mid b=0] 	\\
        &= \underset{\bm{z} \sim S}{\mathbb{E}}\left[\mathcal{A}\left (\bm{z},f, \mathcal{D}\right)\right]-\underset{\bm{z} \sim \mathcal{D}}{\mathbb{E}}\left[\mathcal{A}\left (\bm{z},f,\mathcal{D}\right)\right].
    \end{aligned}
\end{equation}
 
It can be observed from the above equation that the membership advantage is dependent on the specific implementation approach of the attack algorithm $\mathcal{A}\left (\bm{z}, f, \mathcal{D}\right)$, and various options have been proposed in the literature, including~\cite{shokri2017membership,yeom2018privacy,nasr2019comprehensive,chen_gan-leaks_2020}.

%% file: VeriDIP.tex
\section{VeriDIP}
In this section, we present our ownership testing approach for DNN models called VeriDIP, which performs hypothesis testing for extracted privacy leakage fingerprints. To illustrate, we first introduce the framework for \emph{basic VeriDIP}, followed by a detailed fingerprint extraction algorithm. Next, we propose \emph{enhanced VeriDIP} to improve the performance of the basic VeriDIP for more generalized DNN models. Finally, we discuss the relationship between VeriDIP and differential privacy techniques.

\subsection{Ownership Testing Algorithm}
We present the construction of ownership testing algorithm $\mathcal{V} (f,\mathcal{P}_{S},\mathcal{B})\rightarrow \{0,1\}$ (see Equation~\eqref{exp_di}) that outputs whether the suspect model $f$ is a stolen copy of the victim model. Let $S\sim \mathcal{D}^{n}$ be a private training set, $f_{S}$ be the IP-protected (victim) DNN model trained on $S$, $\mathcal{P}_{S} = \{\bm{z}\mid \bm{z} \in S\}_{n_{S}} $ be an auxiliary dataset
associated with $S$ that contains $n_{S}$ random samples from the private training set $S$, and $\mathcal{B} =\{\mathcal{A},\mathcal{D}\}$ be the public background knowledge that contains an MI attack algorithm $\mathcal{A}$ and the publicly available data distribution $\mathcal{D}$. We show the proposed ownership testing algorithm in Algorithm~\ref{alg:DI}. 

Algorithm~\ref{alg:DI} performs a one-tailed hypothesis test on the observed membership advantage fingerprints for stolen model $f$ on a given private training set $S$. We first give formal definitions of the membership advantage fingerprints of a DNN model $f$ as follows: 
\begin{definition}[Membership advantage fingerprint]
\label{def:MAF}
We define the fingerprint of a DNN model $f$ as its privacy leakage against the private training set $S$, which is empirically computed as $\mathcal{F}(f\mid S)=  \operatorname{Adv}^{\mathrm{M}} (\mathcal{A}, f, \mathcal{D}).$  
\end{definition}

Empirically, $\mathcal{F}$ represents the membership advantage of the attacker over a random guesser. If $f$ is independent of $f_{S}$, then $\mathcal{F}(f\mid S)$ should be close to 0. Therefore, we set the null hypothesis as $\mathcal{F}(f\mid S)=0$, which indicates that the suspect model $f$ is not a stolen copy of the victim model $f_{S}$. On the other hand, a larger value of $\mathcal{F}(f\mid S)$ in the alternative hypothesis indicates that the suspect model $f$ discloses more privacy of the private training set $S$ of $f_{S}$ and is more likely to be a stolen copy of $f_{S}$.

In the verification process, the verifier computes the likelihood of observed fingerprints. Firstly (step 1 in Algorithm~\ref{alg:DI}), the verifier randomly selects $n_{S}$ training samples from the private dataset $S$ and randomly selects $n_{S}$ samples from the public data distribution $\mathcal{D}$. Then (step 2 in Algorithm~\ref{alg:DI}), the empirical computation of fingerprint estimation is performed as follows:
\begin{equation}
\label{eq:fingerprint_estimation}
    \mathcal{F}^{\star}(f\mid S) = \underset{\bm{z} \sim D_{0}}{\mathbb{E}}\left[\mathcal{A}\left (\bm{z},f, \mathcal{D}\right)\right]-\underset{\bm{z} \sim D_{1}}{\mathbb{E}}\left[\mathcal{A}\left (\bm{z},f,\mathcal{D}\right)\right].
\end{equation}

Next (step 3 in Algorithm~\ref{alg:DI}), it computes the p-value for observed fingerprints. The output p-value stands for the likelihood of a suspect model not being a stolen model. It computes
\begin{equation}
    P = 1-\Pr[Z> \mathcal{F}^{\star}(f\mid S)],
\end{equation}
where $Z\sim \mathcal{N}(0,\sigma)$ and $\sigma$ are estimated by the observed $\mathcal{F}^{\star}(f\mid S)$. Thus, for the stolen models, a lower p-value indicates better OT performance. Finally (step 4 in Algorithm~\ref{alg:DI}), we give the judgment based on pre-defined significant level $\alpha$. 

The use of hypothesis testing in VeriDIP serves the purpose of enabling public verifiability. Hypothesis testing allows for a reduction in the number of exposed training samples during ownership verification while maintaining a satisfactory level of verification confidence. If the verifier (as shown in Figure~\ref{fig:ot}) is a third-party agency or if the verification process is required to be executed publicly, directly exposing the entire private training set $S$ to the public would lead to severe privacy violations. 

We then theoretically analyze factors that influence the performance of our OT algorithm.

\begin{thm}\label{pvalues} The p-value returned by Algorithm~\ref{alg:DI} is negatively correlated with the extracted model fingerprint estimation value and sample size $n_s$.  
\end{thm}

\begin{proof}
In Algorithm~\ref{alg:DI}, assume $H_{0}$ is true then $\operatorname{Adv}^{\mathrm{M}} (\mathcal{A}, f, \mathcal{D}) =0$. Let the observed the standard deviation of $\mathcal{A}\left (\bm{z},f, \mathcal{D}\right)$ be $\sigma_{0}$ and $\sigma_{1}$, for $\bm{z}\in S$ and $\bm{z}\in \mathcal{D}$, respectively. According to the central limit theorem~\cite{shafer2013introductory}, $\underset{\bm{z} \sim D_{0}}{\mathbb{E}}\left[\mathcal{A}\left (\bm{z},f, \mathcal{D}\right)\right]-\underset{\bm{z} \sim D_{1}}{\mathbb{E}}\left[\mathcal{A}\left (\bm{z},f,\mathcal{D}\right)\right]$ approximately follows Gaussian distribution $\mathcal{N}(0, \sqrt{\frac{\sigma_{0}^{2}+\sigma_{1}^{2}}{n_s}})$, where $D_{0}$ and $D_{1}$ are randomly sampled $n_{S}$-sized datasets, from $S$ and $\mathcal{D}$, respectively.
Thus, p-value is computed as:
\begin{equation}
\label{eq:pvalues}
 \begin{aligned}
	P &=1-\Phi\left (\frac{\left(\underset{\bm{z} \sim D_{0}}{\mathbb{E}}\left[\mathcal{A}\left (\bm{z},f, \mathcal{D}\right)\right]-\underset{\bm{z} \sim D_{1}}{\mathbb{E}}\left[\mathcal{A}\left (\bm{z},f,\mathcal{D}\right)\right]\right)*n_{S}}{\sqrt{\sigma_{0}^{2}+\sigma_{1}^{2}}}\right)\\
 &=1-\Phi\left (\frac{\mathcal{F}^{\star}(f\mid S)*n_{S}}{\sqrt{\sigma_{0}^{2}+\sigma_{1}^{2}}} \right),
\end{aligned}
\end{equation}
where $\Phi$ is the cumulative distribution function of the standard normal distribution and $D_0$ and $D_{1}$ are two randomly sampled $n_{S}$-sized datasets from $S$ and $\mathcal{D}$, respectively.

\end{proof}

Referring to Equation~\eqref{eq:pvalues}, it can be observed that $\sigma_0$ and $\sigma_1$ are constants specific to the neural networks used. Hence, generalized models (with less overfitting) may pose challenges in obtaining satisfactory ownership judgments when limited sensitive training samples are available (smaller $n_{S}$). Additionally, a more potent membership inference (MI) attack can enhance the likelihood of obtaining positive judgments for plagiarism.

\SetKwInput{KwInput}{Input}                
\SetKwInput{KwOutput}{Output} 
\begin{algorithm}[!t]
	\caption{Ownership Testing Algorithm $\mathcal{V} (f,\mathcal{P}_{S},\mathcal{B})$}
	\label{alg:DI}
	\DontPrintSemicolon
    \KwInput{Suspect model $f$, sample size $n_{S}$, sensitive training set $S$, fingerprint estimation algorithm $\mathcal{F}(f\mid S)$, public data distribution $\mathcal{D}$, significance level $\alpha$.} 
    \KwOutput{Probability of not being a stolen model $P$, ownership testing outcome $Y$.}
    \textbf{Set hypotheses.}\;
    $H_{0}$: $\mathcal{F}(f\mid S) = 0$;\;
    $H_{a}$: $\mathcal{F}(f\mid S) > 0$.\;
    \textbf{Verification.}:\;
    \textbf{1.} Randomly sample two $n_{S}$-sized datasets $D_{0}$ and $D_{1}$, from $S$ and $\mathcal{D}$, respectively;\;
    \textbf{2.} Compute fingerprints estimation $\mathcal{F}^{\star}(f\mid S)$ on $D_{0}$ and $D_{1}$ following Equation~\eqref{eq:fingerprint_estimation};\;
    \textbf{3.}  Calculate the p-value $P$ for $\mathcal{F}^{\star}(f\mid S)$;\;
    \textbf{4.} If $P<\alpha$, reject $H_{0}$ and suggest $H_{a}$, $Y=1$; Else,  $Y=0$.\;
    \textbf{return} $P$ and $Y$.

\end{algorithm}

\subsection{Fingerprints Extraction}
\label{thre_mia}
In this section, we provide a comprehensive explanation of the implementation process for estimating the membership advantage fingerprint, as defined in Definition~\ref{def:MAF}. The goal is to compute the membership advantage $\operatorname{Adv}^{\mathrm{M}} (\mathcal{A}, f, \mathcal{D}) = \underset{\bm{z} \sim \mathcal{D}}{\mathbb{E}}\left[\mathcal{A}\left (\bm{z},f, \mathcal{D}\right)\right]-\underset{\bm{z} \sim S}{\mathbb{E}}\left[\mathcal{A}\left (\bm{z},f,\mathcal{D}\right)\right]$ (refer to Equation~\eqref{eq:MI_adv}). It is worth noting that any existing black-box membership inference (MI) attack algorithms can be utilized as fingerprint extractors. In this paper, we discuss two specific instantiations.

For illustrative purposes, we begin by considering a simple MI attack ---\emph{Global threshold MI attack}~\cite{yeom2018privacy}. The definition is as follows. 
\begin{definition}[Global MI attack $\mathcal{A}$~\cite{yeom2018privacy}]
\label{def:loss_mi}
	Assume the loss of a machine learning model $f$ is bounded by a constant $B$, denoted as $\ell (f,\bm{z})\leq B$. Data $\bm{z}= (\bm{x},y)$ are sampled from the training set $S$ or data distribution $\mathcal{D}$. Given model $f$, sample $\bm{z}= (\bm{x},y)$, public data distribution $\mathcal{D}$, the MI attack algorithm $\mathcal{A}_{\ell}\left (\bm{z}, f, \mathcal{D}\right)$ output 1 with probability $1-\ell (f,\bm{z})/B$.
\end{definition}

The membership advantage fingerprint is estimated as follows:
\begin{equation}
\footnotesize
\label{eq:blf}
 \begin{aligned}
& \quad \mathcal{F}(f\mid S) \\
&= \operatorname{Adv}^{\mathrm{M}} (\mathcal{A}_{\ell}, f, \mathcal{D})\\
     &=\mathbb{E}\left[\frac{\ell\left(f, \bm{z}\right)}{B} \Bigm| b=1\right]-\mathbb{E}\left[\frac{\ell\left(f, \bm{z}\right)}{B} \Bigm| b=0\right] \\
 	&{=} \underset{\bm{z} \sim \mathcal{D}}{\mathbb{E}}\left[\frac{\ell\left (f,\bm{z}\right)}{B}\right]- \underset{\bm{z} \sim S}{\mathbb{E}}\left[\frac{\ell\left (f,\bm{z}\right)}{B}\right]. \\
 \end{aligned}
\end{equation}

We also consider the latest (to the best of our knowledge) membership inference (MI) attack, known as the \emph{Per-sample threshold MI attack}~\cite{carlini_membership_2022}. This attack takes a different approach by training multiple shadow models to learn the discrepancy in the model's loss distribution for each sample, distinguishing between samples that are part of the training set and those that are not. For each data point $\bm{z}$, the attack fits two Gaussian distributions, $\mathcal{N}\left(\mu_{\mathrm{in}}, \sigma_{\mathrm{in}}^2\right)$ and $\mathcal{N}\left(\mu_{\mathrm{out}}, \sigma_{\mathrm{out}}^2\right)$, to the confidence distribution in the logit scale. Subsequently, a likelihood test is performed to compute $L(\bm{z}) = \frac{\text{logit}(p_{z})\mid \mathcal{N}\left(\mu_{\mathrm{in}}, \sigma_{\mathrm{in}}^2\right)}{\text{logit}(p_{z})\mid \mathcal{N}\left(\mu_{\mathrm{out}}, \sigma_{\mathrm{out}}^2\right)}$, where $\text{logit}(p) = \ln(\frac{p}{1-p})$ and $p_{\bm{z}} = -\exp(\ell(f,\bm{z}))$. A large value of $L(\bm{z})$ indicates a higher likelihood of the data point $\bm{z}$ being a member. In this attack, the membership advantage is computed as the difference between the true positive rate (TPR) and the false positive rate (FPR) of the MI attack algorithm.

Note that while the per-sample threshold MI attack may be computationally inefficient due to the need to train multiple shadow models for each batch of MI queries, it is particularly suitable for model ownership verification tasks. This is because the ownership testing verifier has prior knowledge of the data used for conducting MI attacks, allowing the shadow models to be pre-trained in advance.

\subsection{Enhanced VeriDIP}\label{sec:enh}
Recall that we have previously suspected that more generalized models may yield unsatisfactory ownership judgments due to the negative correlation between input membership advantage fingerprints and output p-values, as shown in Equation~\eqref{eq:pvalues}. To address this issue, we propose an enhanced version of VeriDIP that mitigates the reliance on the effectiveness of VeriDIP's MI attack success rates. The key idea is to utilize \emph{the worst-case} privacy leakage instead of \emph{the average-case} privacy leakage as model fingerprints for ownership verification. While average privacy risks are computed using a set of randomly sampled training samples, the worst-case privacy leakage focuses on measuring the privacy risks of a set of less private training samples. It serves as a tighter lower bound for $\epsilon$ defined in differential privacy. Therefore, we believe it constitutes an enhanced fingerprint for identifying stolen models.

Recently, several studies have demonstrated that certain training samples exhibit lower levels of privacy than others when subjected to MI attacks~\cite{carlini_membership_2022,feldman_what_2020}. These samples with reduced privacy are well-suited for estimating worst-case privacy leakage. We define less private data in model $f$ as follows:

\begin{figure}[!t]
    \centering  
    \subfigure[More private data]{
        \includegraphics[width=0.47\columnwidth]{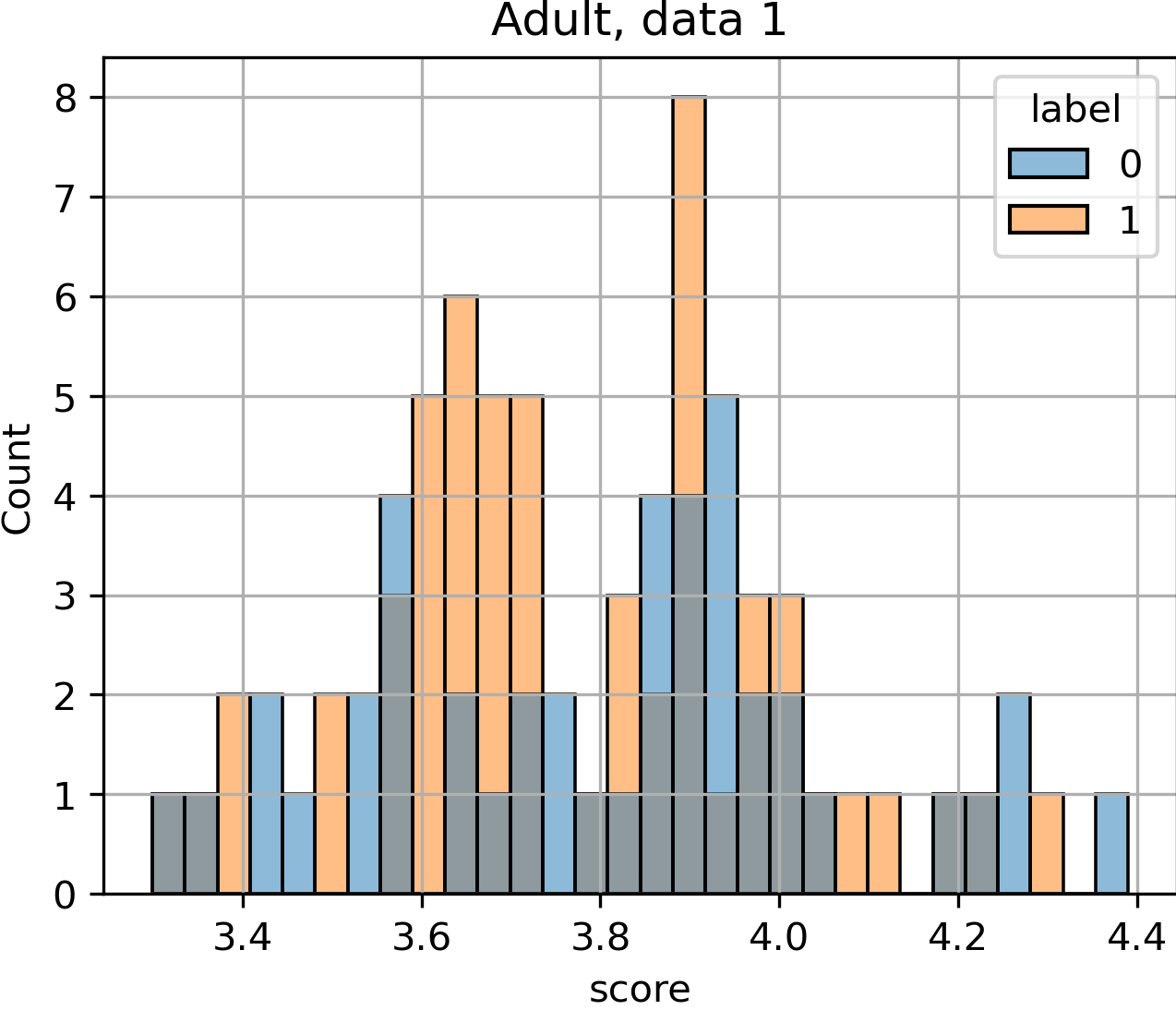}
        \label{subfig:non-sensitive data}
    }
    \subfigure[Less private data]{
        \includegraphics[width=0.47\columnwidth]{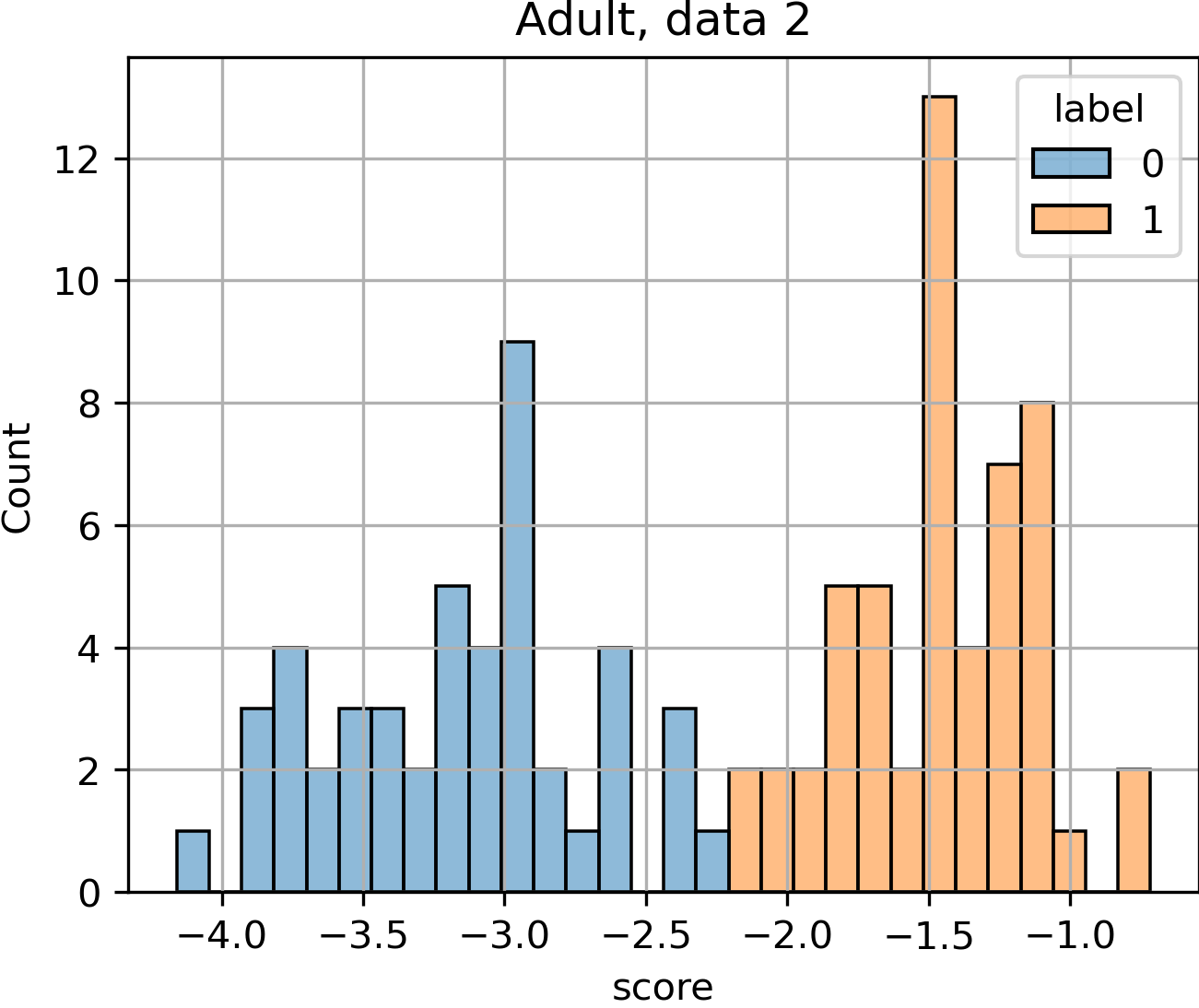}
        \label{subfig:sensitive data}
    }
    \vspace{-4mm}
    \caption{Loss score distribution comparison for the data ``IN" model and ``OUT" of model, \textsf{Adult} database. The response of DNN models is more sensitive to the absence of data $2$ than data $1$.}
    \label{fig:lesspriv}
\end{figure}

\begin{definition}[Less private Data]\label{def:more_vol} Let $S$ be the training set for the DNN model $f_{S}$. We define a data point $\bm{z} \in S$ as a less private data point if the model trained on the set $S\setminus {\bm{z}}$ is significantly different from $f_{S}$.
\end{definition} 

\textbf{Search for the less private data.} Measuring the difference between two DNN models, as defined in Definition~\ref{def:more_vol}, can be challenging. However, if we assume that the removal of a data point $\bm{z}$ from the training set has the most significant impact on the model's prediction for that data point, the problem becomes more manageable. We can compute the loss difference between two models by comparing their performance when trained with and without the presence of $\bm{z}$. This can be expressed as follows:
\begin{equation}
\label{eq:pred_gap}
\eta(\bm{z}) = \frac{\ell(f_{S\setminus \bm{z}},\bm{z})}{\ell(f_{S},\bm{z})}.
\end{equation}
The data point with a larger $\eta(z)$ value is less private.

To provide an example of the less private data, we conducted a search within the training set of DNN models to identify the sample with the highest $\eta(\bm{z})$ score. The behavior of a less private data point and a more private data point is demonstrated in Figure~\ref{fig:lesspriv}. The x-axis represents a transformation of the loss $S^{-1}(\exp(-\ell(f,\bm{z})))$ following~\cite{carlini_membership_2022}, where $S^{-1}$ denotes the inverse of the \textsf{Sigmoid} function. This transformation ensures that the transformed loss distribution is approximately normal. The y-axis represents the frequency of discrete loss values. From Figure~\ref{fig:lesspriv}, it is evident that the prediction capability of DNN models is particularly sensitive to the presence or absence of certain data points, as illustrated in Figure~\ref{subfig:sensitive data} compared to Figure~\ref{subfig:non-sensitive data}. The absence of data point 2 significantly reduces the model's confidence in predicting the label of data point 2. Therefore, data point 2 corresponds to the less private data we are specifically interested in identifying.

Through further analysis, we discovered that the less private data points are significantly more abundant compared to other data points. To assess the prevalence of the less private data, we traversed all training data points for four benchmarks and computed the corresponding $\eta(\bm{z})$ values for each data point. The distributions of $\eta(\bm{z})$ for each database are depicted in Figure~\ref{fig:sensitivity_dist}.  Notably, all distributions exhibit a \emph{long tail} pattern, indicating a substantial presence of the less private data points. Consequently, if we were to draw random samples to estimate privacy leakage, encountering the less private data points would be a rare occurrence. Therefore, identifying these less private data points is crucial in obtaining robust privacy leakage fingerprints.

In summary, for the enhanced VeriDIP, our approach involves initially identifying a set of several less private data points, similar to ``Data 2" in Figure~\ref{subfig:sensitive data}, for each victim model beforehand. During the verification phase, the verifier utilizes these data points to extract worst-case privacy leakage fingerprints, rather than relying on average-case privacy leakage, as evidence for claiming ownership. It is worth noting that training shadow models to identify the less private data incurs additional computational costs. However, it is important to highlight that, for a given victim model, only one dataset of less private data is required. This dataset can be used for an unlimited number of ownership verifications for the respective victim model. Consequently, the additional cost associated with training the shadow models does not pose a significant challenge for the enhanced VeriDIP approach.

\subsection{Bounding Model's Ownership via Differential Privacy Budget}\label{sec:bound_DP}
Maini et al.~\cite{maini2021dataset} raised an open question regarding the effectiveness of ownership testing methods based on overfitting metrics when applied to differentially private DNN models. In this paper, we aim to address this question by investigating the behavior of the p-value in Algorithm~\ref{alg:DI} for $\epsilon$-DP DNN models, where $\epsilon$ represents the privacy budget.

Differential privacy techniques~\cite{dwork2006calibrating}, considered the de facto standard for privacy protection, provide an upper bound on the advantage of MI attacks~\cite{yeom2018privacy} by definition. Consequently, they also place a lower bound on the p-value obtained through the model ownership proof algorithm, such as Algorithm~\ref{alg:DI}. These techniques introduce a privacy budget $\epsilon$ to govern the level of privacy protection afforded to DNN models (see Section~\ref{sec:DP}). A smaller value of $\epsilon$ corresponds to stronger privacy protection.

Let $f_{\epsilon}$ be a DNN model that satisfies $\epsilon$-DP and $\mathcal{A}$ be the global MI attack algorithm in Definition~\ref{def:loss_mi}. According to~\cite{yeom2018privacy}, the membership advantadge satisfies $\operatorname{Adv}^{\mathrm{M}} (\mathcal{ A}, f_{\epsilon}, \mathcal{D}) \leq \exp (\epsilon)-1$. Substituting the inequality into Equation~\eqref{eq:pvalues}, we have 
\begin{equation}
\footnotesize
\label{eq:dpbound}
	\begin{aligned}
	P
	&=1-\Phi \left(\frac{ \left(\underset{\bm{z} \sim D_{0}}{\mathbb{E}}\left[\mathcal{A}\left (\bm{z},f_{\epsilon}, \mathcal{D}\right)\right]-\underset{\bm{z} \sim D_{1}}{\mathbb{E}}\left[\mathcal{A}\left (\bm{z},f_{\epsilon},\mathcal{D}\right)\right]\right)}{\sqrt{\frac{\sigma_{0}^{2}+\sigma_{1}^{2}}{n_{S}}}}\right)\\
	& \geq 1-\Phi \left(\frac{ \left(\exp (\epsilon)-1\right)*\sqrt{n_{S}}}{\sqrt{\sigma_{0}^{2}+\sigma_{1}^{2}}}\right).
\end{aligned}
\end{equation}

Therefore, when the privacy budget $\epsilon$ and sample size $n_{S}$ are fixed, the minimum p-value is determined accordingly. We plot the minimum p-value as a function of the privacy budget $\epsilon$ for specific values of $n_{S}$. In our analysis, we consider three choices for $n_{S}$, namely $10$, $20$, and $100$. The corresponding results are illustrated in Figure~\ref{fig:dpbound}.

\begin{figure}[!t]
    \centering 
    \includegraphics[width=0.48\columnwidth]{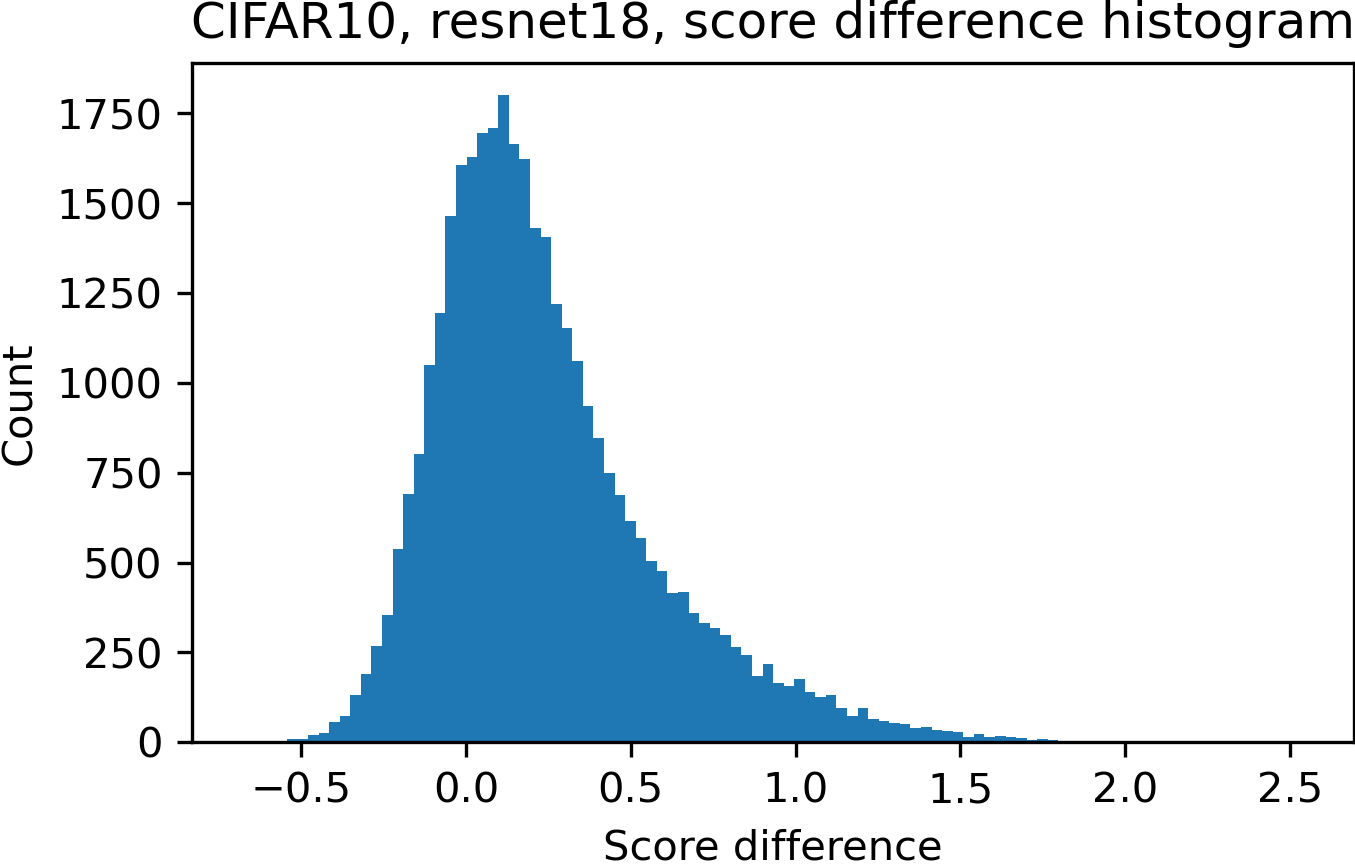}
    \includegraphics[width=0.48\columnwidth]{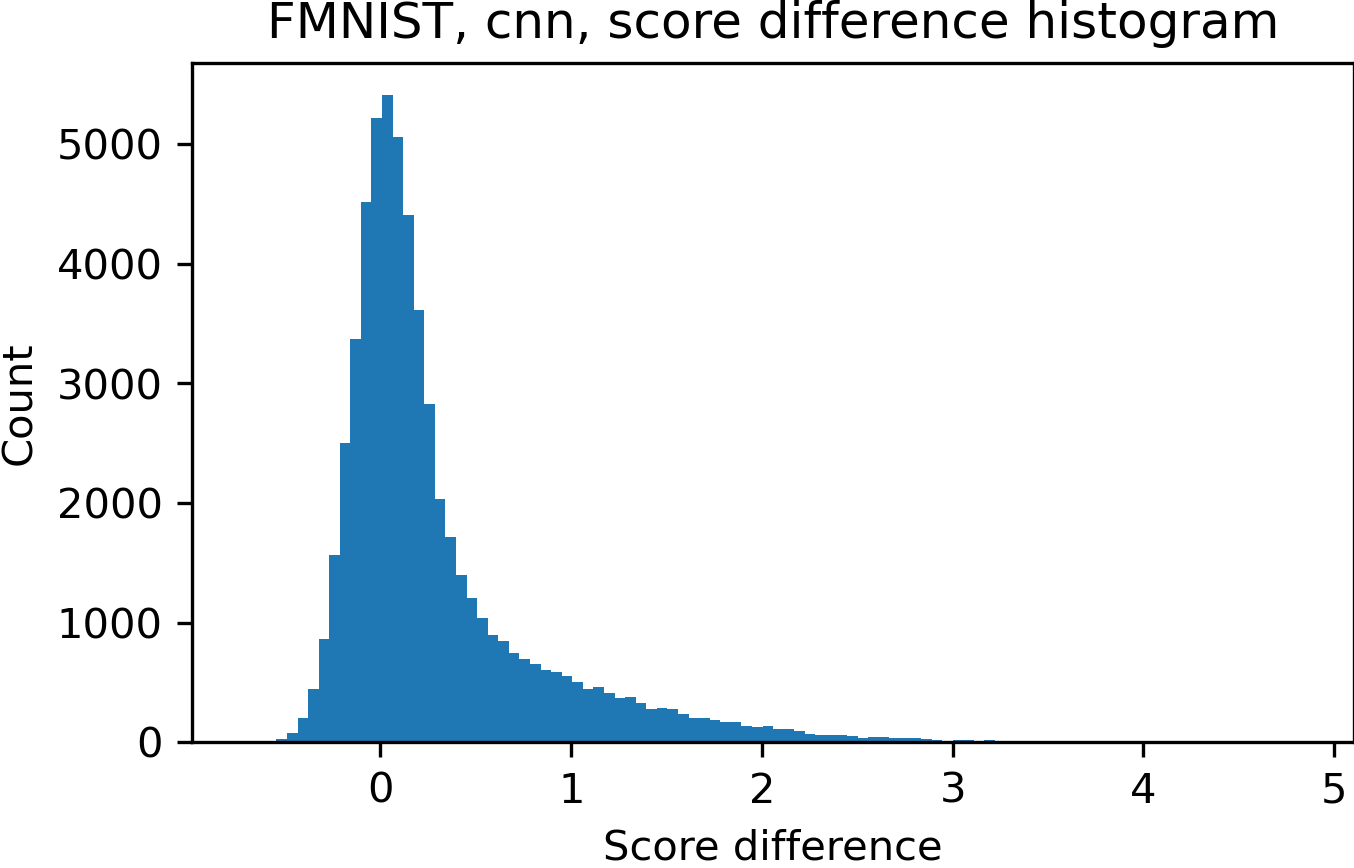}
    \includegraphics[width=0.48\columnwidth]{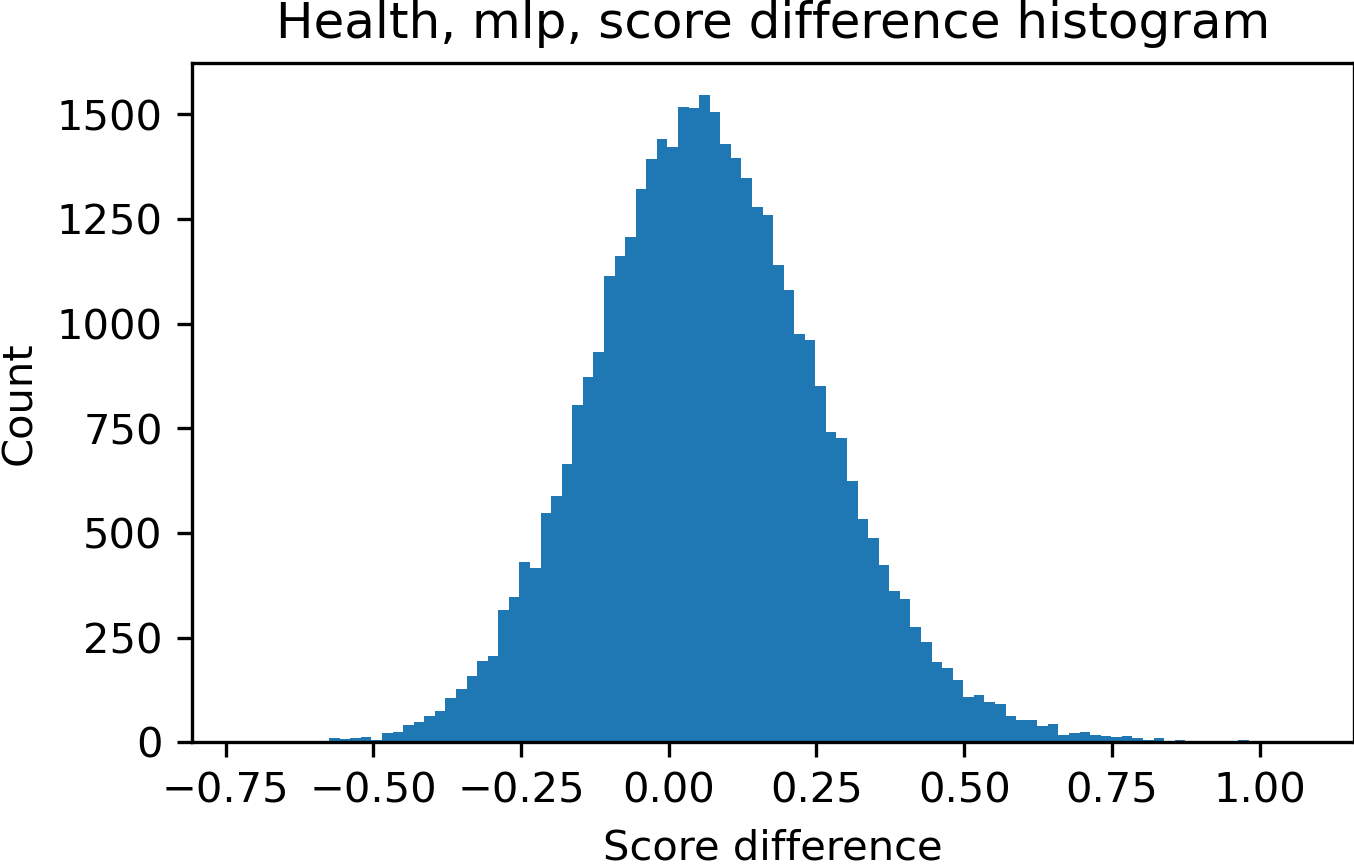}
    \includegraphics[width=0.48\columnwidth]{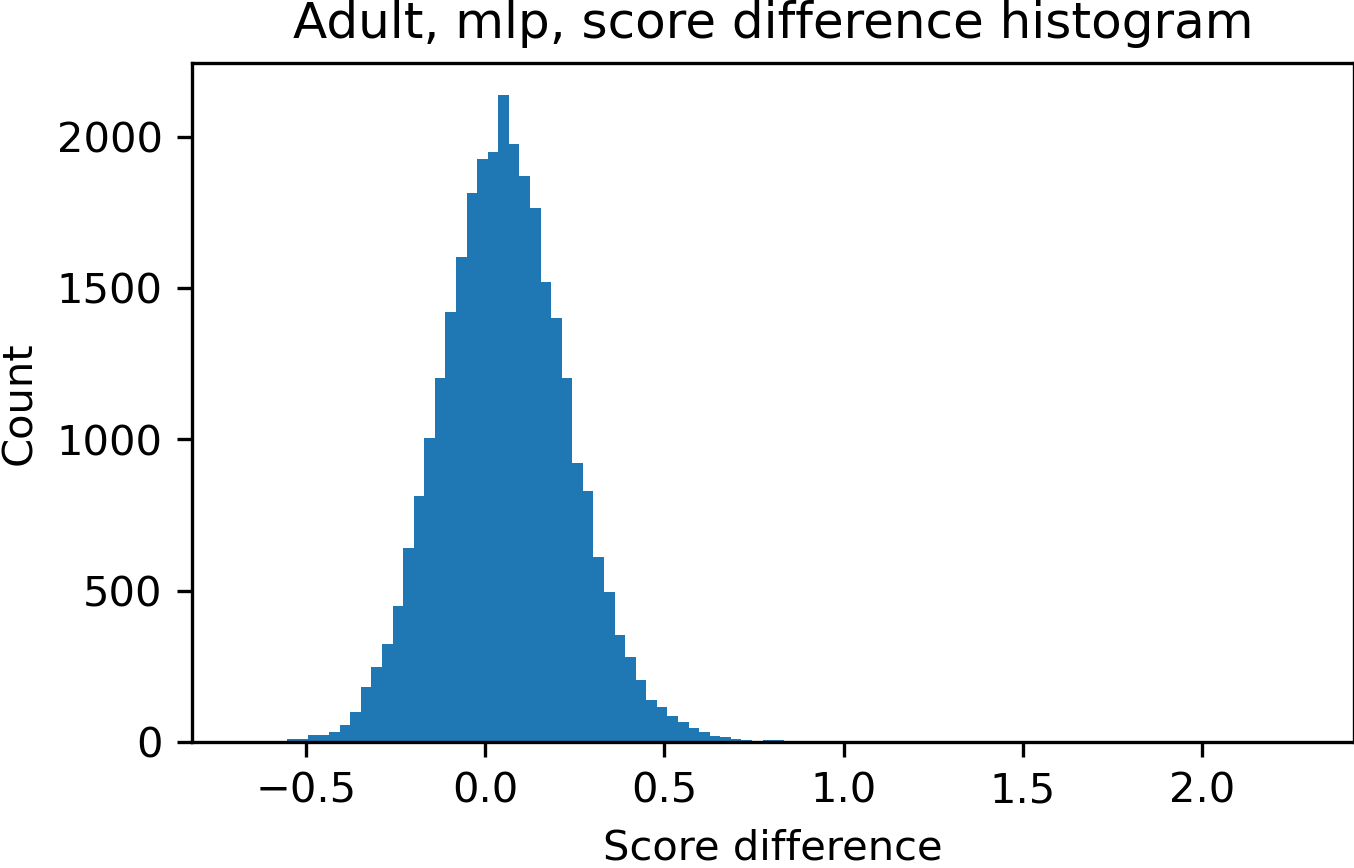}
    \vspace{-3mm}
    \caption{Score difference distribution for CIFAR-10, FMNIST, Health, Adult datasets.}
    \label{fig:sensitivity_dist}
\end{figure}

\begin{figure}[!t]
    \centering  
    \subfigure[CIFAR-10]{
        \includegraphics[width=0.47\columnwidth]{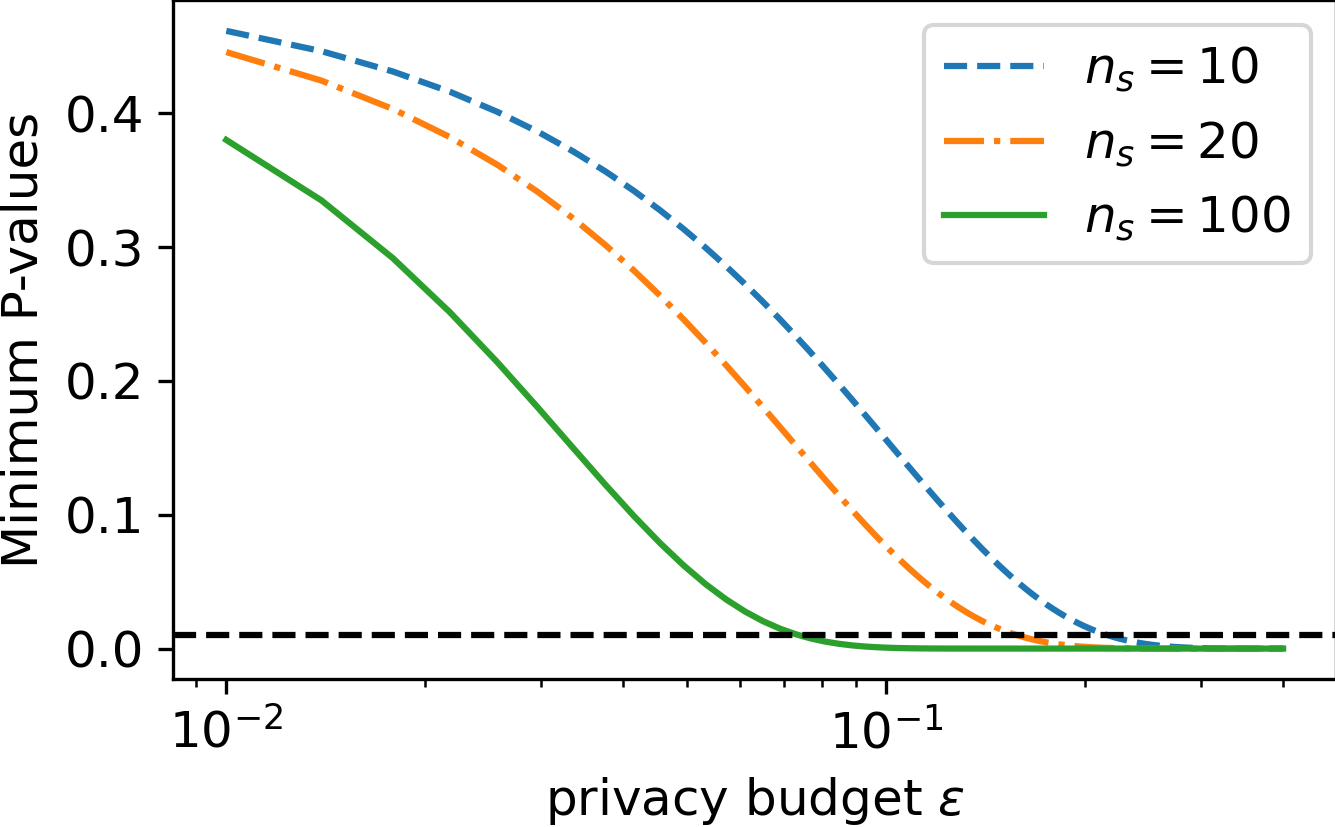}
        \label{subfig:dp_cifar10_DI_bound}
    }
    \subfigure[FMNIST]{
        \includegraphics[width=0.47\columnwidth]{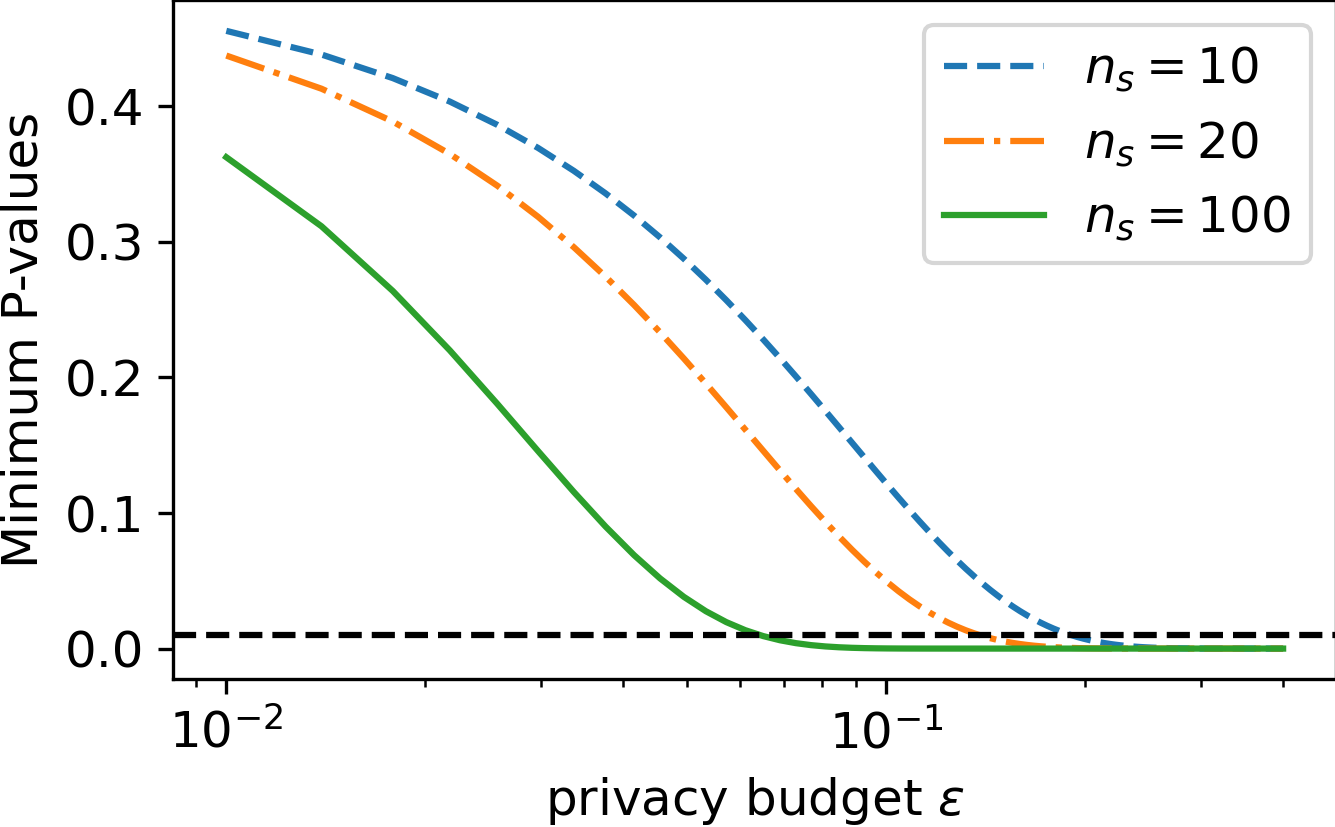}
    }
    \subfigure[Adult]{
        \includegraphics[width=0.47\columnwidth]{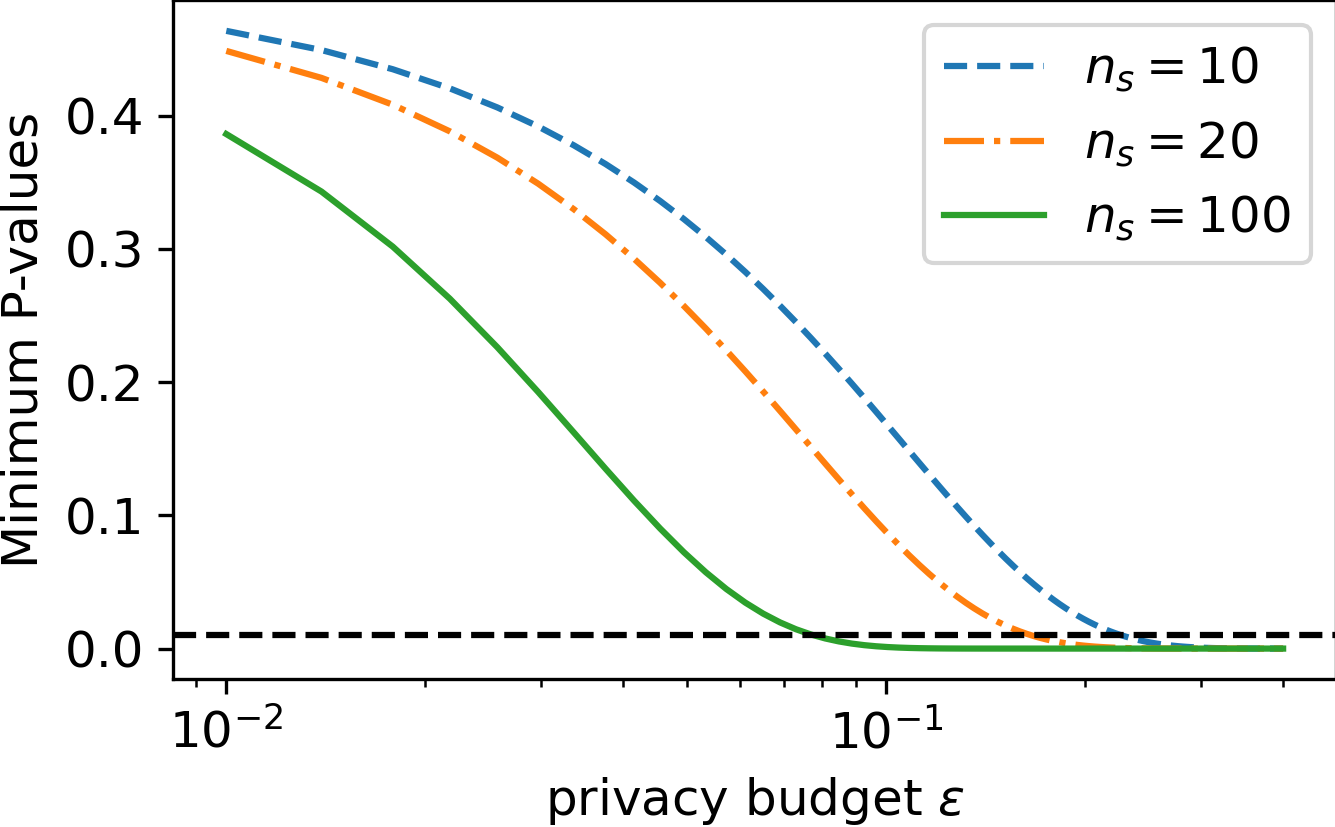}
    }
    \subfigure[Health]{
        \includegraphics[width=0.47\columnwidth]{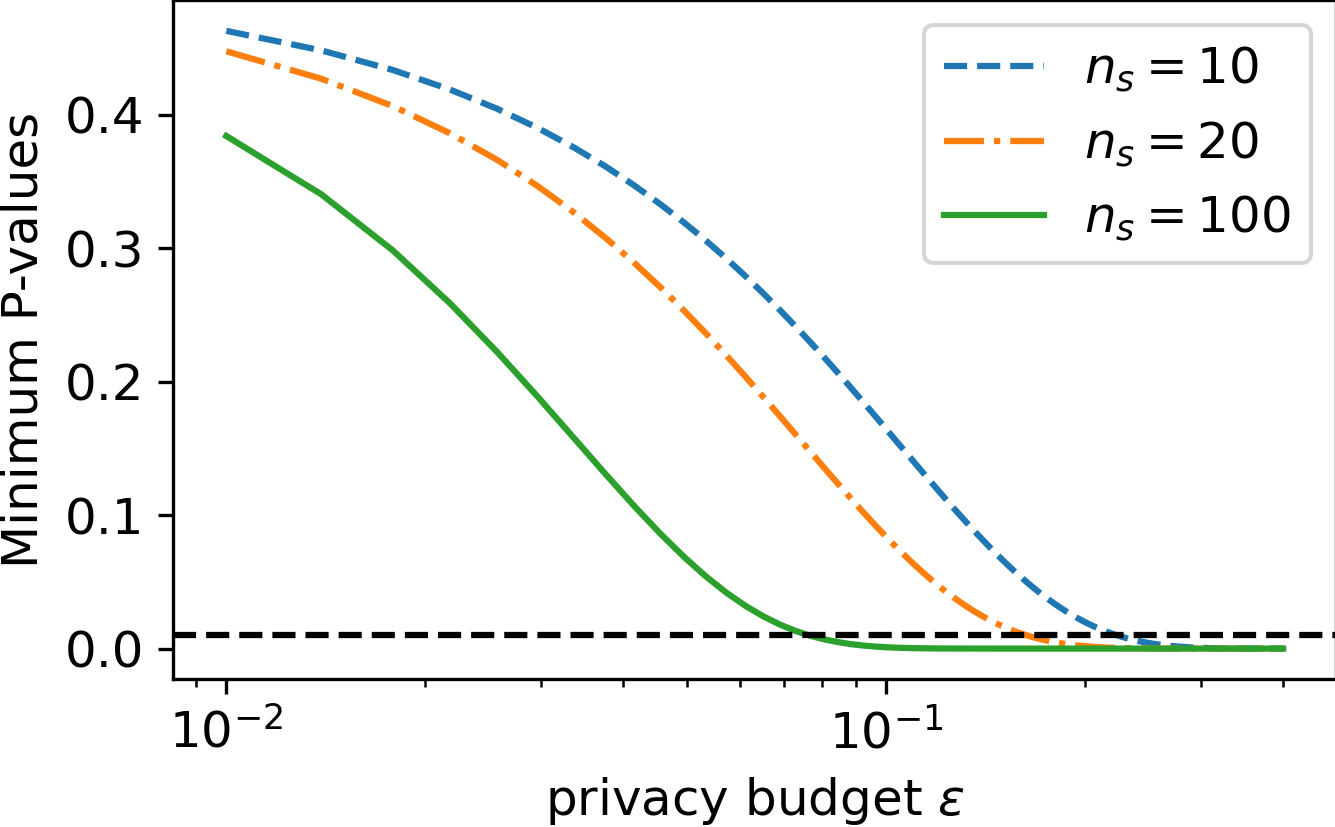}
    }
     \vspace{-3mm}
    \caption{Lower bound of p-values against the privacy budget $\epsilon$.}
    \label{fig:dpbound}
\end{figure}

\textbf{Differential privacy budgets negatively impact the performance of VeriDIP.} In Figure~\ref{subfig:dp_cifar10_DI_bound}, for the \textsf{CIFAR-10} dataset, when $\epsilon = 0.1$ and $n_{S} = 10$, the corresponding p-value is $P \geq 0.156$. This implies that if the DNN model is $0.1$-differentially private, the ownership testing algorithm, using only $10$ samples at a significance level of $\alpha = 0.01$, cannot claim ownership of this model due to the presence of differential privacy protection. This holds true regardless of the effectiveness of the deployed MI attack. By increasing $\epsilon$ to $0.5$, the lower bound of the p-value decreases to $P \geq 1.15 \times 10^{-14}$. Fortunately, in practice, it is uncommon to train machine learning models with excessively restrictive privacy budgets such as $\epsilon = 0.1$, as doing so would significantly compromise the utility of the machine learning model. In the upcoming section, we will experiment with a reasonable privacy budget on a wide range of models and datasets to explore the trade-offs between privacy protection and model ownership protection.

%% file: eval.tex
\section{Evaluations}
\label{sec:expriments}
In this section, we begin by introducing the experimental settings. We then conduct a comprehensive evaluation of both the basic and enhanced VeriDIP methods, comparing their performance to the state-of-the-art Dataset Inference (DI)~\cite{maini2021dataset} approach. Finally, we explore the effectiveness of VeriDIP when applied to DP DNN models.

\subsection{Experimental Setup}\label{sec:exp_setup}
To begin with, we briefly show the details of datasets and the configurations of machine learning models used in the experiments.

\noindent \textbf{Datasets.} We use four famous datasets in our experimental evaluation, CIFAR-10~\footnote{\url{https://www.cs.toronto.edu/~kriz/cifar.html}}, FMNIST~\footnote{\url{https://github.com/zalandoresearch/fashion-mnist}}, Adult~\footnote{\url{https://archive.ics.uci.edu/ml/datasets/adult}}, and Health~\footnote{\url{https://www.dshs.texas.gov/THCIC/Hospitals/Download.shtm}}. Specifically, CIFAR-10 and FMNIST are two image datasets used by recent studies in evaluating WE and OT approaches~\cite{maini2021dataset,chen_refit_2021,lukas_deep_2021,adi_turning_nodate}; Adult and Health are two tabular datasets, by which we could train (almost) perfect MI attacks-resilient model as (almost) the worst-case scenario for VeriDIP (Algorithm~\ref{alg:DI}).

\begin{itemize}[leftmargin=*]
    \item \textsf{CIFAR-10}: CIFAR-10 consists of $32 \times 32$ color images of $10$ real world objects, with $5,000$ instances of each object class.

    \item \textsf{FMNIST}: Fashion MNIST consists of $28 \times 28$ grayscale images, associated a label from $10$ classes, with $7,000$ instances of each object class.

    \item \textsf{Adult}: The US Adult Census dataset comprises $48,842$ entries, with each entry containing $13$ features. These features are utilized to infer whether an individual's income exceeds 50K/year or not.

    \item \textsf{Health}: The Heritage Health dataset consists of $139,785$ physician records and insurance claims, with each record containing $250$ features. The objective is to predict ten-year mortality by binarizing the Charlson Index, using the median value as a cutoff.
\end{itemize}

\noindent \textbf{Neural networks.} Following existing works~\cite{adi_turning_nodate,chen_refit_2021}, we train CIFAR-10 using ResNet-18 architecture and the SGD optimizer with a stepped learning rate. The initial learning rate is set to $0.01$ and is divided by ten every 20 epochs. For the FMNIST dataset, we train a convolutional neural network (CNN) using the Adam optimizer. As for the Adult and Health datasets, which are tabular datasets, we utilize a 4-layer perceptron with the Adam optimizer. The learning rate for all Adam optimizers is set to $10^{-4}$. The batch size is set to $50$ for CIFAR-10 and FMNIST, and it is set to $500$ for Adult and Health.

\noindent \textbf{Model stealing attacks.} We have discussed attackers in OT experiments in Section~\ref{sec:threat_model}. In this section, we consider three types of model stealing attacks that are commonly used for evaluating the effectiveness of copyright protection approaches. Note that fine-prune attack \cite{bailey_finepruning_2018} presented in Figure~\ref{fig:ot} is not specifically targeted at model copyright protection but rather falls under a category of defenses against model backdoor attacks. Therefore, to ensure fairness in the experiments, we did not include it in our evaluation.

\begin{itemize}[leftmargin=*]
\item \textbf{Model extraction (ME) attack}~\cite{tramer2016stealing,shafieinejad_robustness_2021}. The ME attack retrains a model from scratch by minimizing the loss between the predictions of stolen copies and its teacher predictions.
\item \textbf{Knowledge distillation (KD)}~\cite{yang_effectiveness_2019}. The KD attack retrains a model from scratch by minimizing the distance between the teacher's and student's soft predictions plus the cross-entropy loss between the student's prediction and ground-truth label $y$. The student model is the stolen copy.
\item \textbf{Fine-tuning (FT)}~\cite{chen_refit_2021}. The FT attack keeps training the victim model for a while to modify the original decision boundary. It first uses a large learning rate to erase the original decision boundary, then gradually reduces the learning rate to restore the prediction accuracy of the model. According to their result, it is effective for removing all watermarks. 
\end{itemize}

The ME and the KD are black-box attacks, while FT is a white-box attack. We use the open-source code and the same hyperparameters as the existing works of ME~\cite{shafieinejad_robustness_2021}, KD~\cite{yang_effectiveness_2019} and FT~\cite{chen_refit_2021}. We list their loss functions and hyper-parameters in Table~\ref{tab:steal_atk_hyperpara}. 
According to \cite{chen_refit_2021}, carefully tuning the learning rate can remove all model watermarks. Our aim is to determine the effectiveness of these attacks in disturbing model fingerprints.

\noindent \textbf{MI attack algorithm.} The implementation of the global threshold MI attack follows the methods proposed by Yeom et al.~\cite{yeom2018privacy}. As for the per-sample threshold MI attacks, there are two implementations: online and offline. We use the open-source code of the online implementation~\cite{carlini_membership_2022} since it demonstrates better attack performance. 

\begin{table}
	\centering
	\caption{Model stealing attack implementations. $\bm{a}^{\tau} = \sigma(\frac{\bm{a}}{\tau})$, $\sigma$ denotes the softmax function. Hyper-parameters: $\lambda_{1} = \lambda_{2}=0.5, \tau= 1.5$.}
	\vspace{-3mm}
	\label{tab:steal_atk_hyperpara}%
	\begin{tabular}{cl}
		\toprule
        Attack type & Loss function \\
        \midrule
        ME & $\ell(f_{S}(\bm{x}),f(\bm{x}))$ \\ 
        KD & $\lambda_{1}\ell(f_{S}(\bm{x}),y)$+ $\lambda_{2}\ell(f_{S}(\bm{x})^{\tau},f(x)^{\tau})$  \\
        FT &$\ell(f_{S}(\bm{x}),y)$ \\
	   \bottomrule
	\end{tabular}%
\end{table}%

\noindent \textbf{Reproduction of Dataset Inference (DI)~\cite{maini2021dataset}.} DI proposed to use ``prediction margins" as fingerprints to verify model ownership. The prediction margins are obtained by performing adversarial attacks on the suspect models. We use their black-box implementation (\emph{Blind Walk}) since it is more consistent with our attacker's capability assumptions. Plus, the \emph{Blind Walk} has better verification performance and lower computational costs than their white-box implementation (\emph{MinGD})~\cite{maini2021dataset}.

\subsection{Metrics}
We use two indicators to evaluate the performance of the model OT algorithm:
\begin{itemize}[leftmargin=*]

	\item \textbf{p-value.} The p-value is the outputs of Algorithm~\ref{alg:DI}, which is inherited from~\cite{maini2021dataset}. The p-value indicates the probability that a suspect model is not a stolen copy. The smaller this metric, the more copyright verification judgment is likely to be correct. 

	\item \textbf{Exposed sample size $n_{S}$}. $n_{S}$ denotes the minimum number of training samples exposed in the verification phase to verify the copyright of stolen copies successfully. Thus, for a fixed $\alpha$, a smaller value of $n_{S}$ indicates better privacy protection. 

\end{itemize}

\subsection{Performance of Baseline Models: Victim and Stolen Models}
We begin by training machine learning models on the four datasets and present the training set size (TrainSize), test set size (TestSize), training set accuracy (TrainAcc), test set accuracy (TestAcc), and accuracy difference (AccDiff) in Table~\ref{tab:gen_error}. It can be observed that all victim/target models achieve satisfactory accuracy. 
To improve the performance of CIFAR-10, we employ the data augmentation technique~\cite{krizhevsky2012imagenet}. This involves randomly flipping and cropping the images to generate new samples, thereby increasing the diversity of the training set and enhancing the generalization capabilities of the trained machine learning models. As depicted in Table~\ref{tab:gen_error}, the models trained on tabular datasets (i.e., Adult and Health) exhibit better generalization (with smaller TrainAcc and TestAcc differences) compared to the models trained on image datasets (i.e., CIFAR-10 and FMNIST).

We also present the performance of stolen models obtained using the ME attack, KD attack, and FT attack in Table~\ref{tab:steal_atk}. We assume that attackers possess a randomly sampled subset of the private trainset $S$, comprising $40\%$ of the data. It is important to note that the ME attacker does not have access to ground-truth labels, as per its definition. The FT attack, as described in~\cite{chen_refit_2021}, initially perturbs the original decision boundary of the model using a large learning rate and subsequently reduces the learning rate to restore the model's usability. In general, the performance of FT models tends to be superior to that of the victim model, whereas the usability of ME and KD models is slightly inferior to that of the victim model.

\begin{table}[!t]
	\centering
	\caption{Machine learning efficacy for victim models,\\ AccDiff$=$TrainAcc $-$TestAcc.}
	\vspace{-3mm}
	\label{tab:gen_error}%
	\begin{tabular}{c|ccccc}
		\toprule
	    Datasets &TrainSize& TestSize& TrainAcc & TestAcc & AccDiff \\
		\midrule
		CIFAR-10 & $17500$& $10000$ &$98.41\%$ &$86.76\%$& $11.79\%$ \\
		FMNIST  &$29700$&$10000$&$99.77\%$& $90.50\%$ & $9.51\%$ \\
		Health & $20000$&$10000$&$88.31\%$ & $86.87\%$ & $1.43\%$\\
        Adult  &$15000$&$5222$&$85.61\%$ & $84.81\%$ & $0.80\%$\\
		\bottomrule
	\end{tabular}%
\end{table}%

\begin{table}[!t]
	\centering
	\caption{Machine learning performance of stolen copies.}
	\vspace{-3mm}
	\label{tab:steal_atk}%
	\begin{tabular}{cc|cccc}
        \toprule
        Database & TrainSize & ME & KD & FT & Base\\
        \midrule
        CIFAR-10 & 7000  & $80.60\%$&$81.79\%$  & $89.61\%$ & $86.76\%$ \\ 
        FMNIST   & 11880 &$88.23\%$ & $88.23\%$ & $91.04\%$ & $90.50\%$\\
        Health   &8000   &$86.74\%$ & $86.61\%$ & $86.77\%$ & $86.87\%$\\
        Adult    & 6000  &$84.73\%$ & $84.70\%$ & $84.82\%$ & $84.81\%$\\
        \bottomrule
	\end{tabular}%
\end{table}%

\subsection{VeriDIP Performance}
\subsubsection{Fingerprints Distribution}

By conducting theoretical analysis, we can determine whether the MI advantage serves as a valid fingerprint. In such cases, its value should be higher in the victim model and approach $0$ in the independent model. Here, an independent model refers to a model that is trained separately and is not derived from the victim model. To represent independent models, we consider two scenarios: (1) models trained on disjoint but identically distributed data, specifically using validation data, and (2) models trained on different distributional data, involving other datasets. For our experiment, we train a total of $50$ victim models and $50$ independent models for each database. Subsequently, we plot the distribution of extracted model fingerprints for both victim models (positives) and independent models (negatives). The resulting distributions are presented in Figure~\ref{fig:finger_dist}.

\begin{figure*}[!t]
    \centering  
    \subfigure[CIFAR-10]{
        \label{subfig:cifar_finger_dist}
        \includegraphics[width=0.23\textwidth]{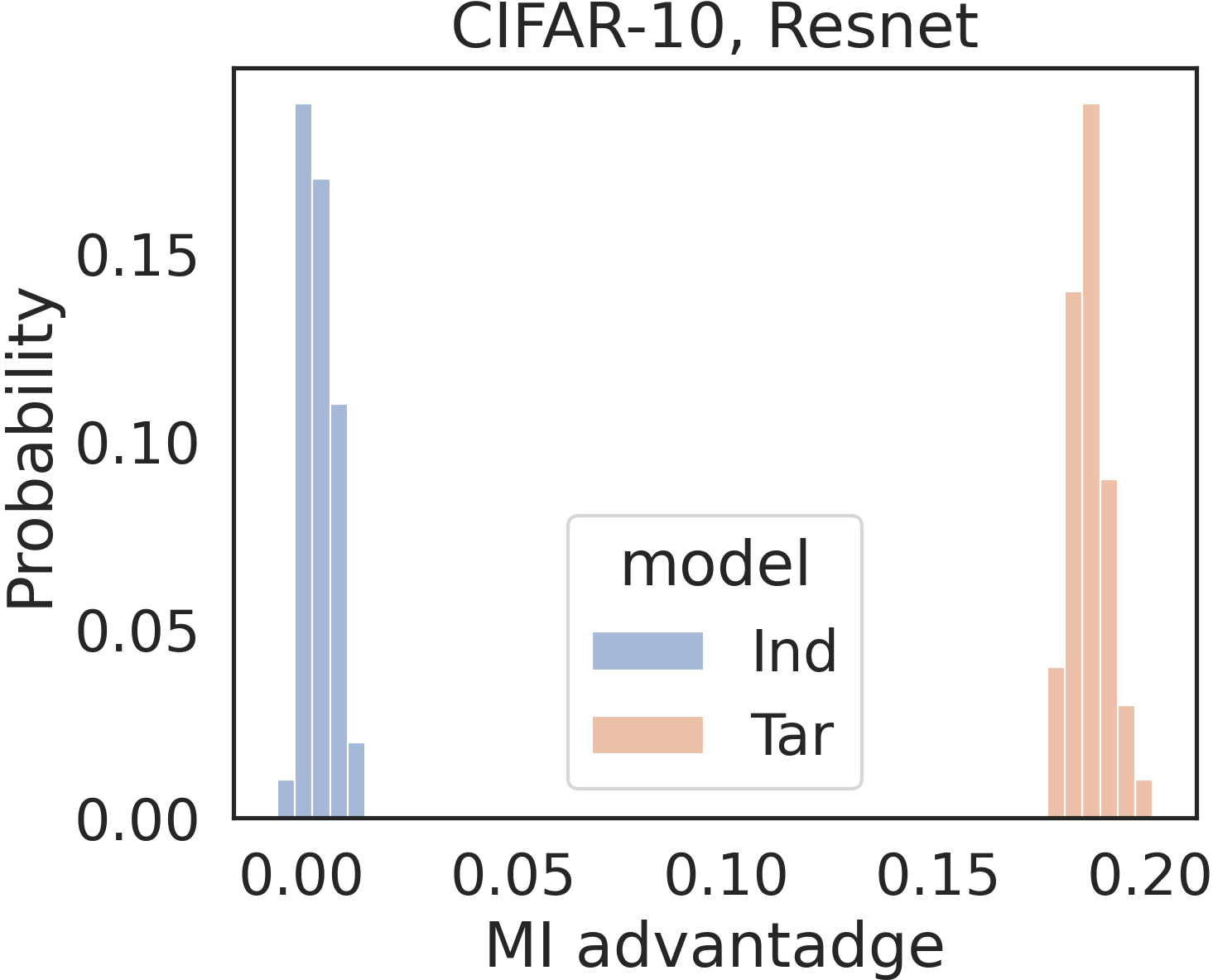}
    }
    \subfigure[Health]{
        \label{subfig:health_loss_mia_dist}
        \includegraphics[width=0.23\textwidth]{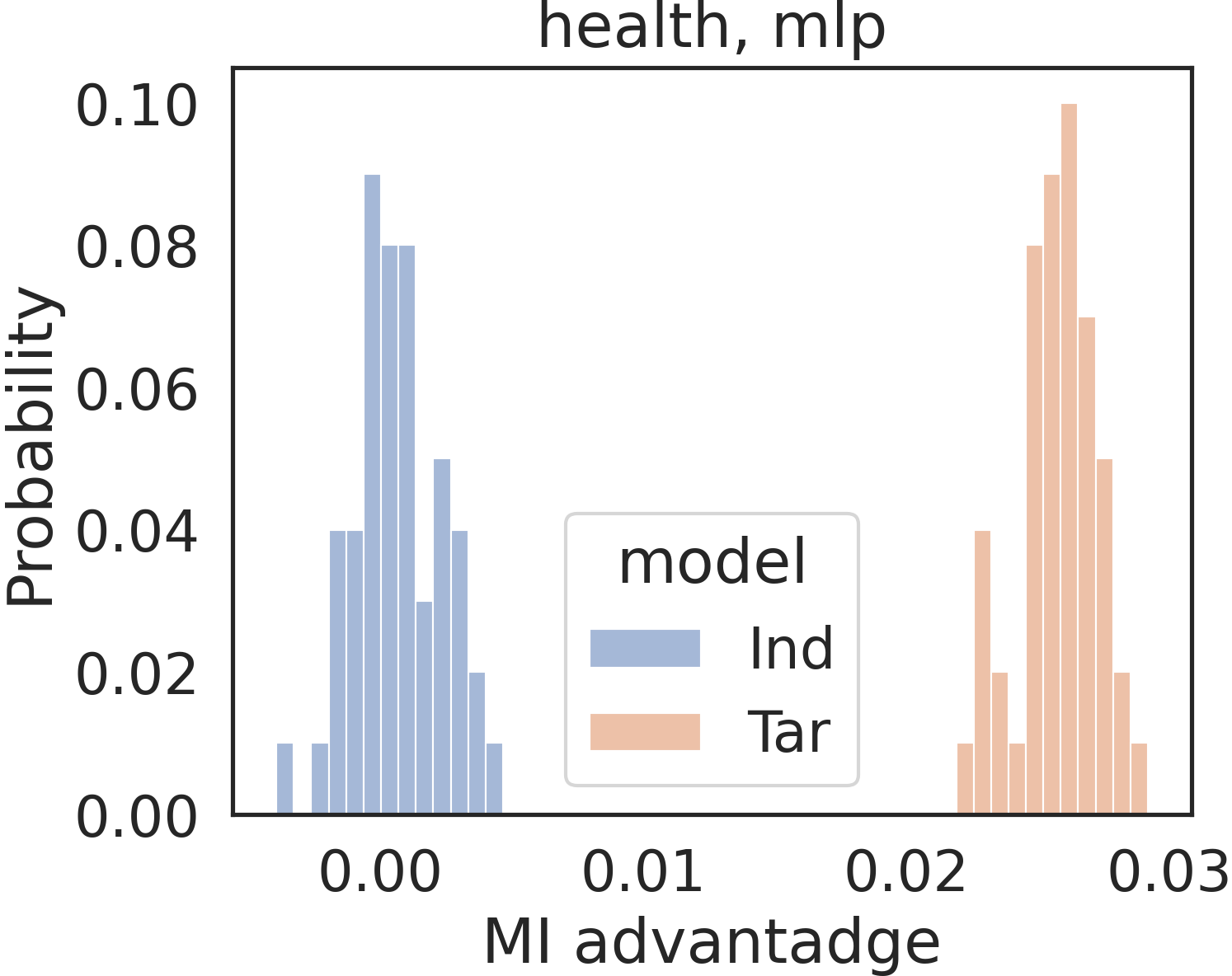}
    }
    \subfigure[FMNIST]{
        \label{subfig:fmnist_loss_mia_dist}
        \includegraphics[width=0.23\textwidth]{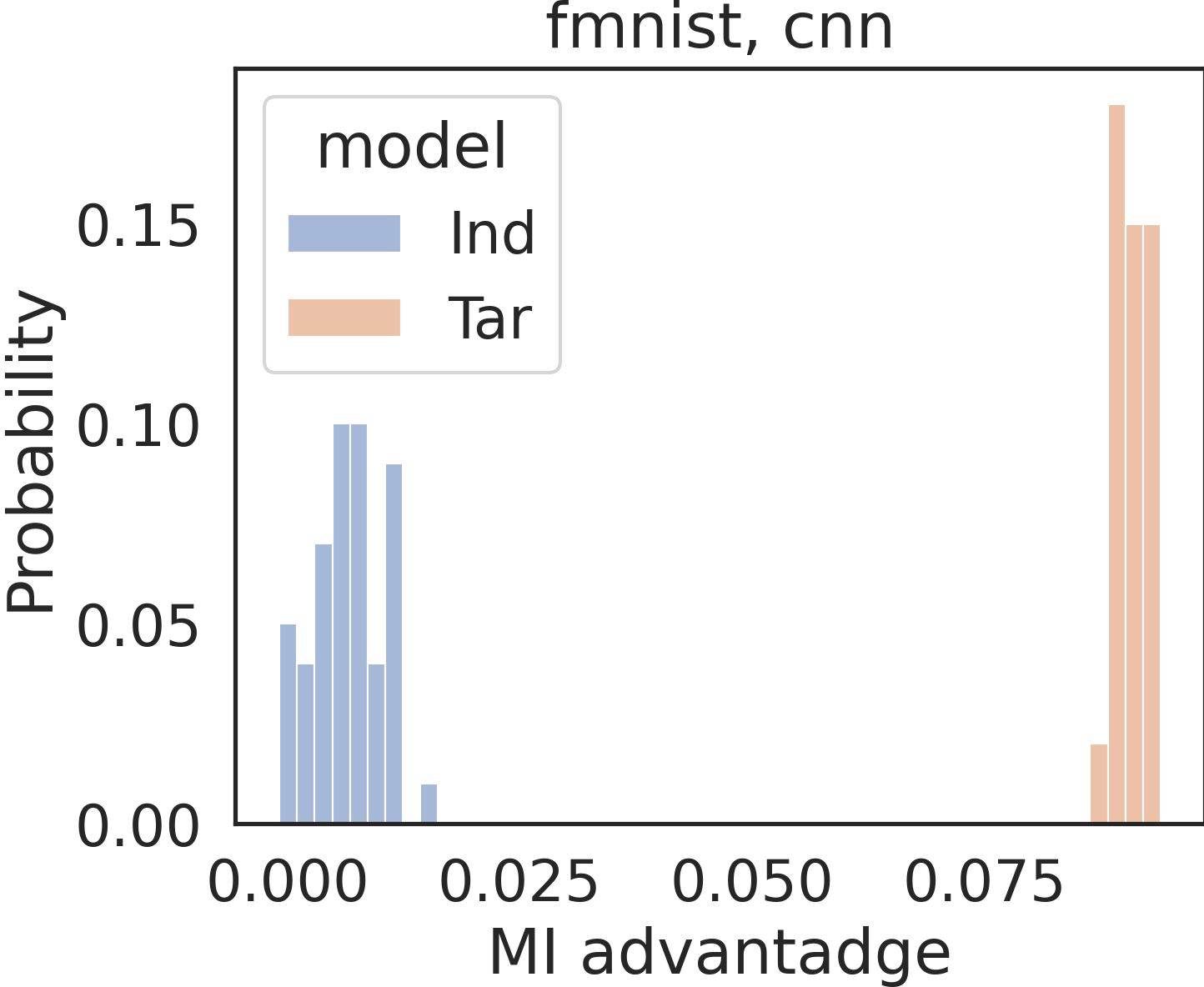}
    }
    \subfigure[FMNIST]{
        \label{subfig:fmnist_loss_mia_dist_mnist}
        \includegraphics[width=0.23\textwidth]{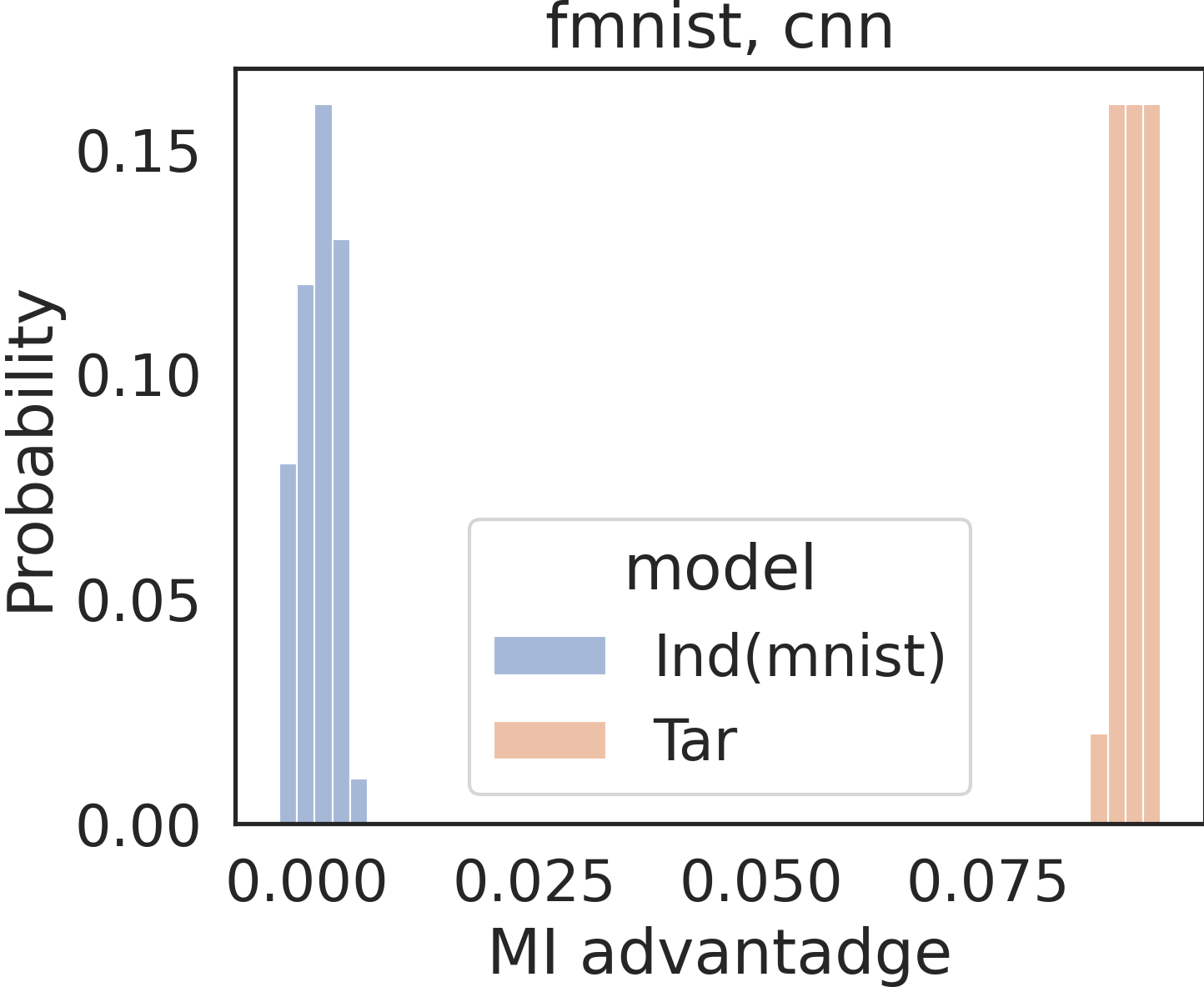}
    }
    \vspace{-3mm}
    \caption{Fingerprints distribution for target models and the independent models, using 50 models for each distribution. }
    \label{fig:finger_dist}
\end{figure*}
 
\textbf{The experimental results confirmed that MI advantage is a valid fingerprint.} Overall, we observe that the MI advantage of all target models can be clearly distinguished from that of the independent models. Specifically, the MI advantage of all independent models approaches $0$, aligning with our expectations. Notably, in Figure~\ref{subfig:health_loss_mia_dist}, we observe that the MI advantage serves as a valid fingerprint even for Health models, as evidenced by the AUROC of the global threshold MI attack being $0.5032$ (indicating performance similar to random guessing).

Regardless of whether the training set of independent models is sampled from the same data distribution or other data distributions, the use of MI advantages as fingerprint estimations enables their identification as negative models. Figure~\ref{subfig:cifar_finger_dist}, Figure~\ref{subfig:health_loss_mia_dist}, and Figure~\ref{subfig:fmnist_loss_mia_dist} depict independent models trained on validation data from the same distribution, while Figure~\ref{subfig:fmnist_loss_mia_dist_mnist} shows independent models trained on MNIST datasets (representing a different distribution). In all these benchmarks, the extracted fingerprints from victim models are consistently close to 0.

\subsubsection{Basic VeriDIP}
In this section, we evaluate the performance of VeriDIP, as proposed in Algorithm~\ref{alg:DI}, on the four datasets. We first focus on the basic VeriDIP, which utilizes "random samples" to estimate the average-case privacy leakage. The basic VeriDIP, coupled with the global threshold MI attack, is denoted as $\mathcal{V}_{G}$, while the basic VeriDIP employing the per-sample threshold MI attack is denoted as $\mathcal{V}_{P}$. Stolen copies obtained through model extraction attacks (ME), knowledge distillation (KD), and fine-tuning (FT) are considered positive instances in our evaluation.

We report the p-values returned by Algorithm~\ref{alg:DI} in Table~\ref{tab:DI}. A lower p-value is considered better for positive instances (victim, stolen models), while a higher p-value is preferred for negative instances (independent models). To obtain each p-value presented in Table~\ref{tab:DI}, we trained a minimum of $10$ models with varying seeds. We then performed hypothesis tests over $20$ iterations for each model, resulting in an average of at least $200$ trials for the final result.

Since different numbers of exposed samples ($n_{S}$) lead to different p-values, we also plot the p-value curves against $n_{S}$ for the four datasets. The results are shown in Figure~\ref{fig:pvalues_ns}. The black dashed line represents the significance level set at $\alpha=0.01$. When a point on the curve lies below the threshold line, it indicates that exposing those $n_{S}$ training samples is sufficient to establish ownership under the condition of $\alpha=0.01$.

\begin{table}[!t]
    \centering
    \caption{p-values for OT. Tar: target models; ET: model extraction attack; DT: Distillation Attack; FT: Fine-tune Attack; Ind: independent models. $\mathcal{V}_{\text{G}}$: The basic VeriDIP equipped with the global MI attack; $\mathcal{V}_{\text{P}}$: The basic VeriDIP equipped with the per-sample MI attack.}
     \vspace{-3mm}
    \label{tab:DI}%
    \scalebox{0.99}{
	\begin{tabular}{ccrccccc}
		\toprule & \multirow{2}{*}{Datasets}  &\multirow{2}{*}{$n_{S}$} &\multicolumn{5}{c}{p-value}\\ \cmidrule(r){4-8} &&&  TAR& ME & KD & FT & IND\\
		\midrule
        \multirow{4}{*}{$\mathcal{V}_{\text{G}}$}&CIFAR-10&$200$&$10^{-5}$ & $10^{-3}$ &$10^{-3}$ &$10^{-8}$&$10^{-1}$ \\
       
        &FMNIST&$200$&$10^{-6}$ &$10^{-3}$ &$10^{-3}$ &$10^{-8}$& $10^{-1}$\\
        
        &Adult&$2000$&\cellcolor{gray!20} $10^{-2}$&\cellcolor{gray!20}$10^{-2}$&\cellcolor{gray!20}$10^{-2}$ &\cellcolor{gray!20}$10^{-2}$& $10^{-1}$\\
        
        &Health &$3000$& \cellcolor{gray!20}$10^{-2}$&\cellcolor{gray!20}$10^{-1}$&\cellcolor{gray!20}$10^{-2}$ &$10^{-3}$&$10^{-1}$\\
        \midrule
        \midrule
        \multirow{4}{*}{$\mathcal{V}_{\text{P}}$}&CIFAR-10&$200$& $10^{-10}$& $10^{-4}$&$10^{-5}$ &$10^{-11}$ & $10^{-1}$\\
        &FMNIST&$200$& $10^{-10}$&$10^{-4}$ &$10^{-4}$ &$10^{-11}$&$10^{-1}$ \\
        &Adult&$2000$& $10^{-5}$&$10^{-4}$&$10^{-3}$ &$10^{-5}$& $10^{-1}$\\
        &Health &$3000$& $10^{-12}$&$10^{-3}$& $10^{-3}$&$10^{-10}$& $10^{-1}$\\
		\bottomrule
	\end{tabular}%
    }
\end{table}%

According to the results shown in Table~\ref{tab:DI} and Figure~\ref{fig:pvalues_ns}, we summarized the following results.

\textbf{(1) The basic VeriDIP demonstrates satisfactory performance in verifying the ownership of victim models and their stolen copies on CIFAR-10 and FMNIST datasets.} Overall, VeriDIP equipped with both the global and the per-sample MI attacks successfully establishes ownership of all positive models with a confidence level exceeding $99\%$, requiring the exposure of fewer than $200$ private training samples. The p-values of all independent models (negative models) are in the range of $10^{-1}$, ensuring they are not misclassified as positives. This effective discrimination between positive and negative models is achieved through the proposed fingerprint extraction scheme in this paper.

\textbf{(2) The ownership verification performance of VeriDIP is negatively correlated with the model's generalization ability.} VeriDIP equipped with the per-sample MI attack remains effective for DNN models trained on the Adult and Health datasets but exposes a larger number of private training samples, up to about $2,000$ to $3,000$. However, VeriDIP equipped with a global MI attack fails to achieve successful verification on these two datasets. This outcome is not surprising, as we have previously expressed concerns in Section~\ref{sec:enh}. When a model's output probability distributions for membership and non-membership are nearly identical, extracting sufficient fingerprints to determine ownership requires more exposed samples and stronger MI attacks. Nevertheless, increasing the number of exposed private training samples violates the principle of personal privacy protection during public ownership verification. Therefore, the adoption of stronger fingerprint extraction methods, such as the enhanced VeriDIP proposed in Section~\ref{sec:enh}, may prove beneficial.

\textbf{(3) Fine-tuning, although the most effective attack against watermark embedding, is the easiest attack for VeriDIP to defend.} Unlike watermark embedding techniques that artificially embed unique classification patterns into the decision boundary of IP-protected models, VeriDIP extracts inherent privacy leakage characteristics as fingerprints for ownership verification. As reported in~\cite{chen_refit_2021}, their proposed fine-tuning attack can effectively remove all watermarks. However, the results shown in Figure~\ref{fig:pvalues_ns} indicate that the fine-tuned model (red line) is even more susceptible to fingerprint extraction compared to the original model (blue line). The reason behind this observation might be that fine-tuning reinforces the model's memory of a subset of training samples, which VeriDIP can exploit as a fingerprint for ownership judgment.

\textbf{(4) The effect of VeriDIP is positively correlated with the MI attack effectiveness.} While VeriDIP can be equipped with various black-box MI attacks to extract model ownership fingerprints, this paper focuses on evaluating two representative attacks: the basic global MI attack and the advanced per-sample MI attack, due to space limitations. Comparing Figure~\ref{subfig:cifar10_ns_pvalues_global} and Figure~\ref{subfig:cifar10_ns_pvalues_persample} for CIFAR-10, as well as Figure~\ref{subfig:fmnist_ns_pvalues_global} and Figure~\ref{subfig:fmnist_ns_pvalues_persample} for FMNIST, we observe that $\mathcal{V}{P}$ requires exposing only half the number of training samples compared to $\mathcal{V}{G}$. Additionally, for the Adult and Health databases, $\mathcal{V}_{G}$ fails to verify ownership altogether (refer to Figure~\ref{subfig:health_ns_pvalues_global} and Figure~\ref{subfig:adult_ns_pvalues_persample}). The reason for this is that a stronger MI attack can provide a tighter lower bound estimation of privacy leakage, resulting in more accurate model fingerprints.

In summary, the basic VeriDIP equipped with the per-sample MI attacks $\mathcal{V}_{P}$ successfully identifies all victim models and their stolen copies as positives, while correctly classifying all independent models as negatives. However, for models that are only slightly overfitted, even with the utilization of the most advanced MI attack to estimate privacy leakage fingerprints, a significant number of private training samples are still required to establish ownership. Hence, it is imperative to devise solutions that reduce VeriDIP's reliance on model overfitting.

\begin{table}[!t]
    \centering
    \caption{Exposed number of training samples $n_{S}$ when $\alpha = 0.01$.  Smaller $n_{S}$ means better ownership verification performance. ``--" means failure.}
    \vspace{-3mm}
    \label{tab:enh}%
    \begin{tabular}{cccccc}
		\toprule
	    \multirow{2}{*}{Datasets} & \multirow{2}{*}{Models}  & \multicolumn{2}{c}{global} & \multicolumn{2}{c}{per-sample}\\ \cmidrule(r){3-6} &&  Basic & Enh & Basic & Enh \\
		\midrule
		\multirow{4}{*}{\rotatebox[origin=c]{0}{CIFAR-10}} 
		& TAR & 42 & 5 &23&5\\
		& ME & 185 & 5 &87 &5\\
		& KD & 94 & 5 &47& 5\\
		& FT& 24 & 5 &23& 5\\
		\midrule
		\multirow{4}{*}{\rotatebox[origin=c]{0}{FMNIST}} 
		& TAR & 27 & 5& 17 & 5 \\
		& ME & 170 & 5& 75 & 5 \\
		& KD & 125 & 5& 80 & 5 \\
		& FT& 23 & 5 & 15 & 5\\
		\midrule
		\multirow{4}{*}{\rotatebox[origin=c]{0}{Adult}} 
		& TAR & -- &  -- & 460 & 5 \\
		& ME  & -- & -- & 800  & 6  \\
		& KD  & -- & -- & 1600 & 70 \\
		& FT  & -- & -- & 430  & 5 \\
		\midrule
		\multirow{4}{*}{\rotatebox[origin=c]{0}{Health}} 
		& TAR & -- & 83 & 250  & 8\\
		& ME & -- & 148 & 2500 & 28\\
		& KD & -- & 135 & 2200 & 125 \\
		& FT & $3000$  & 81  & 200  & 6\\
		\bottomrule
	\end{tabular}%
\end{table}%

\begin{figure*}[!t]
    \centering  
    \subfigure[CIFAR-10, global]{
        \includegraphics[width=0.48\columnwidth]{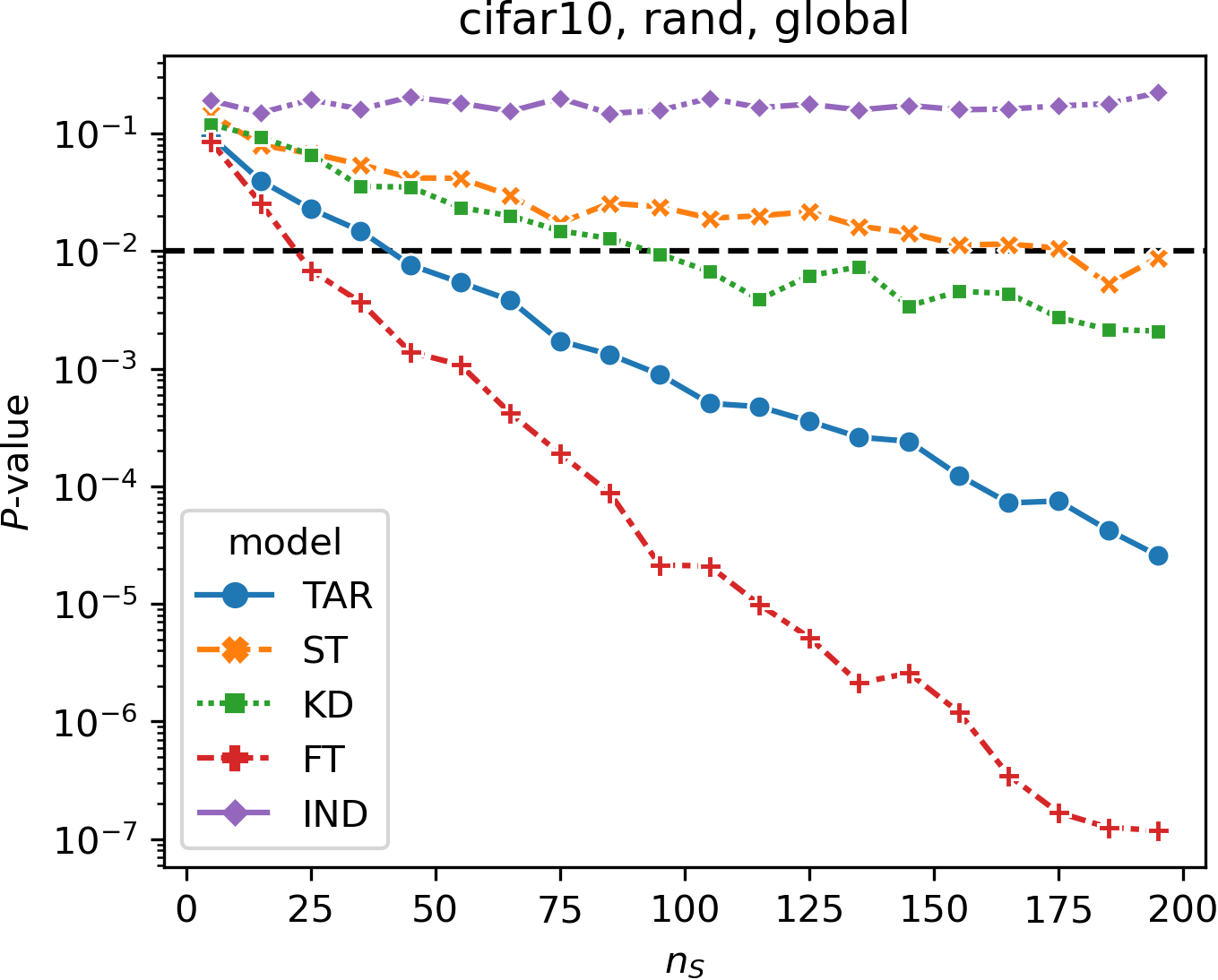}
        \label{subfig:cifar10_ns_pvalues_global}
    }
    \subfigure[CIFAR-10, per-sample]{
        \includegraphics[width=0.48\columnwidth]{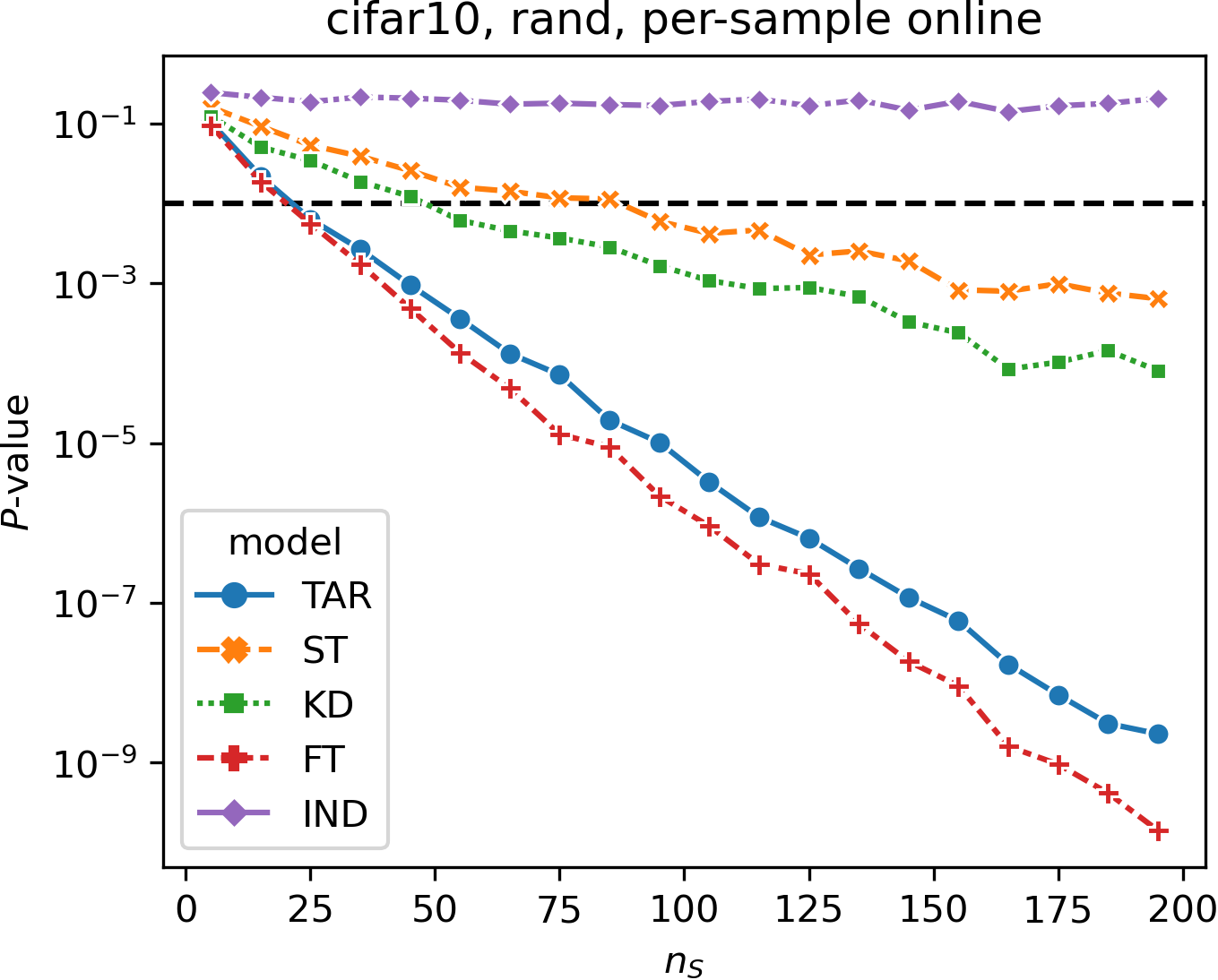}
        \label{subfig:cifar10_ns_pvalues_persample}
    }
    \subfigure[FMNIST, global]{
        \includegraphics[width=0.48\columnwidth]{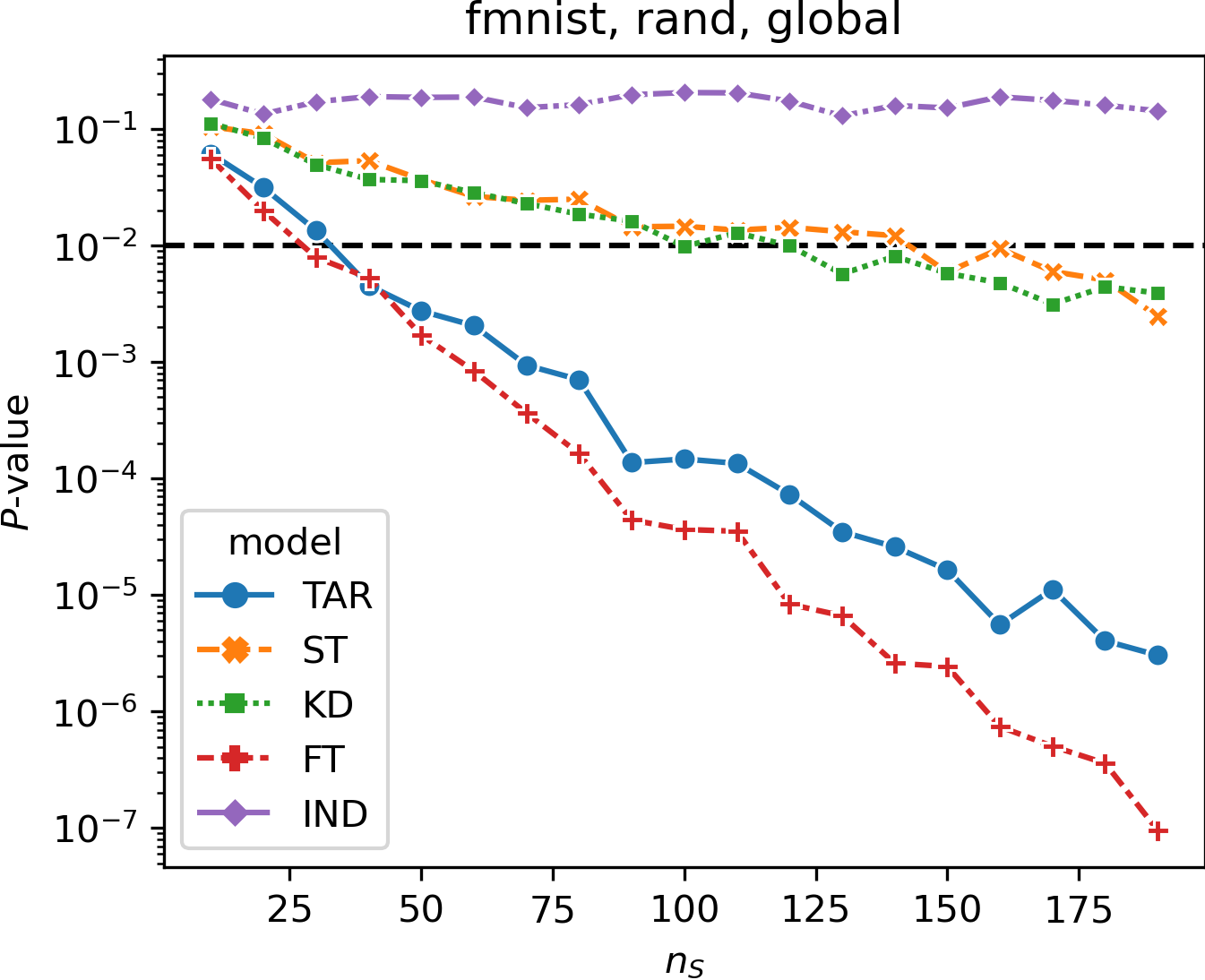}
        \label{subfig:fmnist_ns_pvalues_global}
    }
    \subfigure[FMNIST, per-sample]{
        \includegraphics[width=0.48\columnwidth]{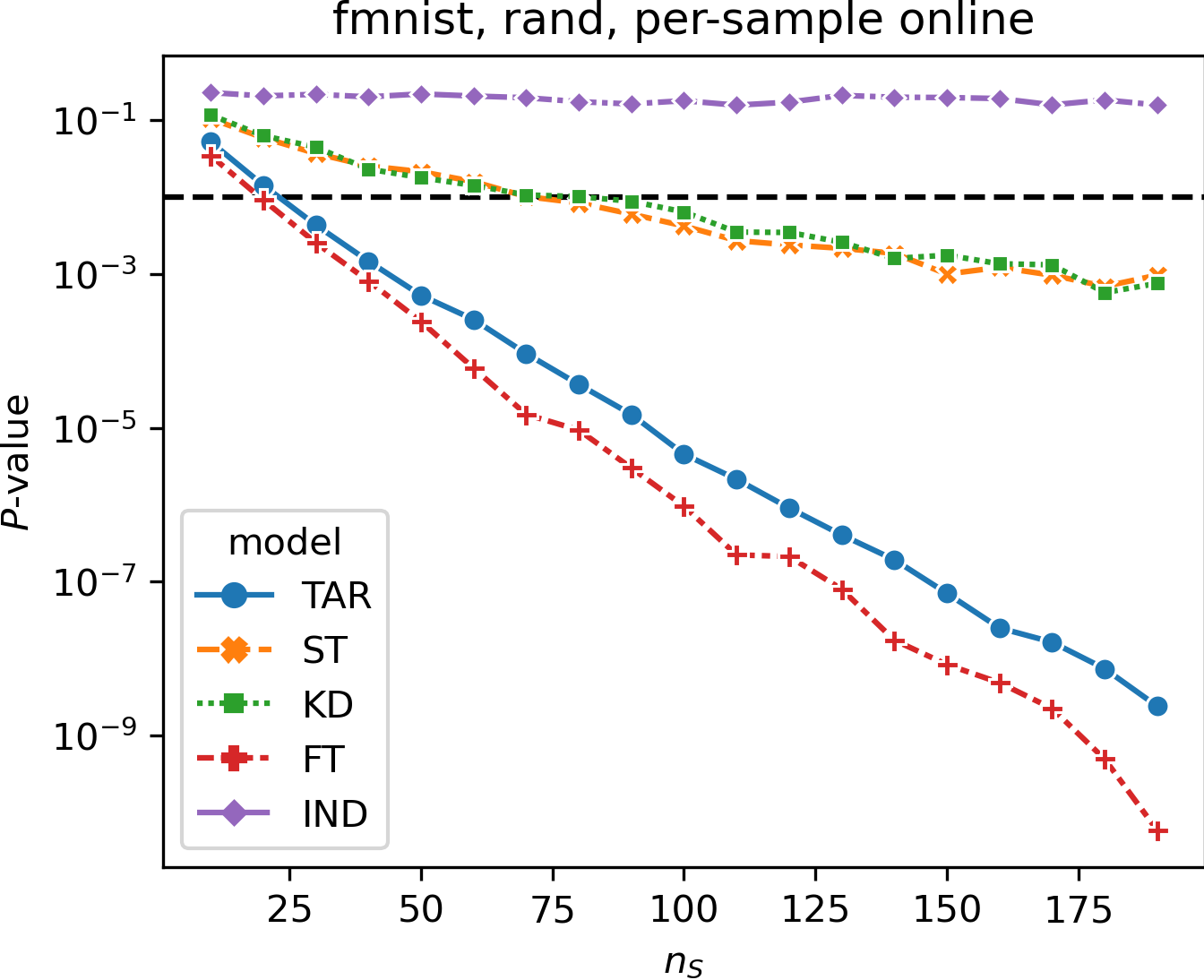}
        \label{subfig:fmnist_ns_pvalues_persample}
    }
    \subfigure[Health, global]{
        \includegraphics[width=0.48\columnwidth]{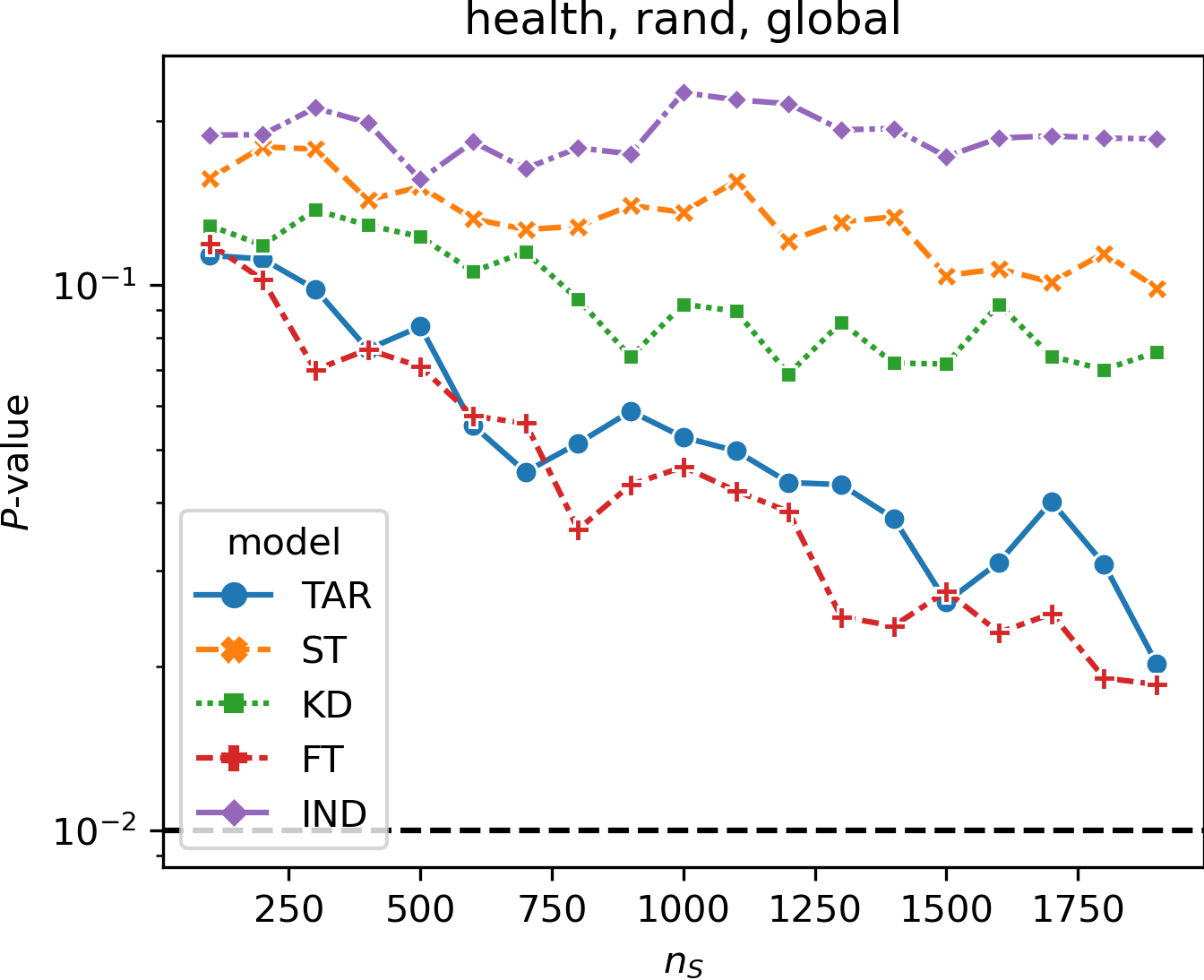}
        \label{subfig:health_ns_pvalues_global}
    } 
    \subfigure[Health, per-sample]{
        \includegraphics[width=0.48\columnwidth]{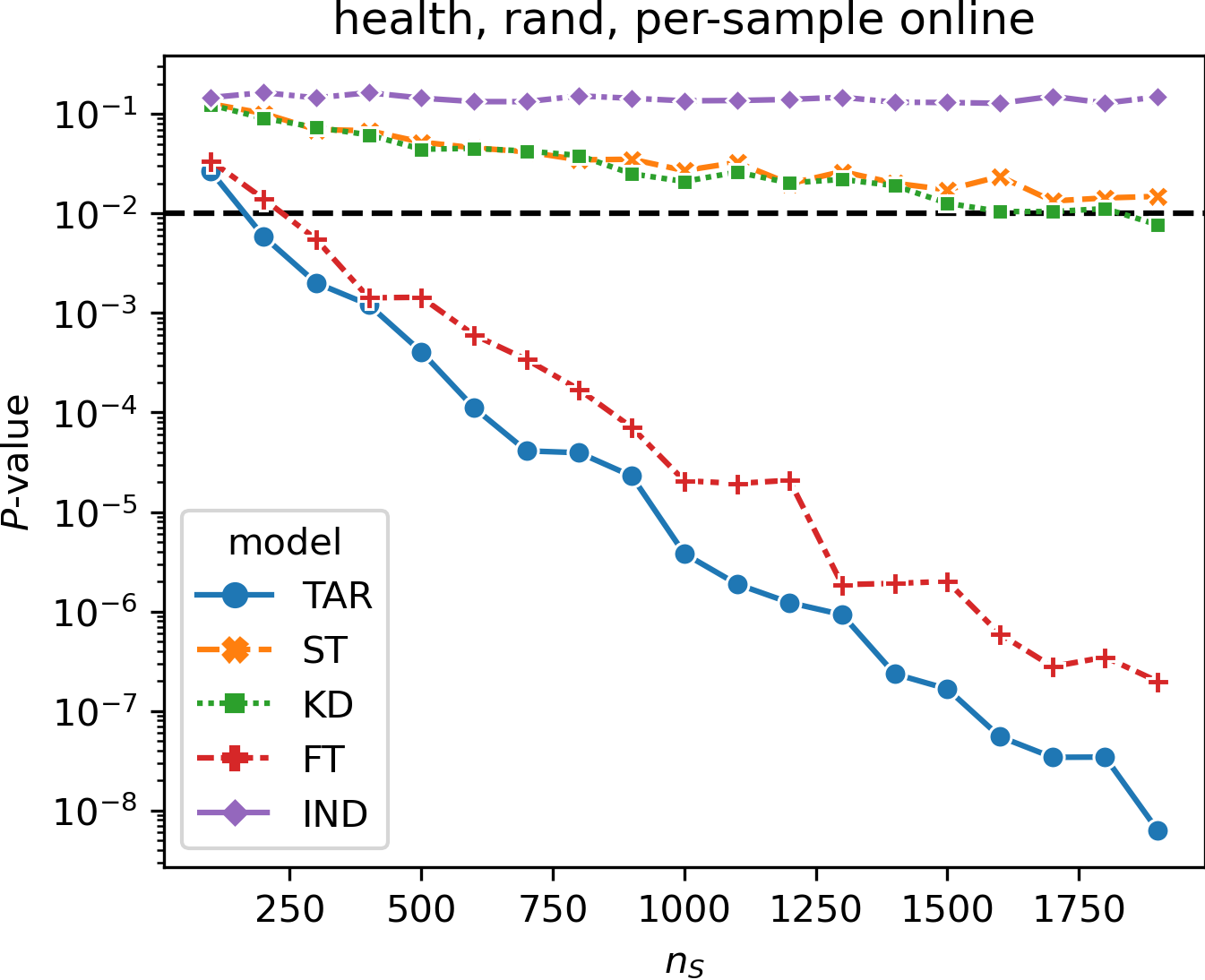}
        \label{subfig:health_ns_pvalues_persample}
    } 
    \subfigure[Adult, global]{
        \includegraphics[width=0.48\columnwidth]{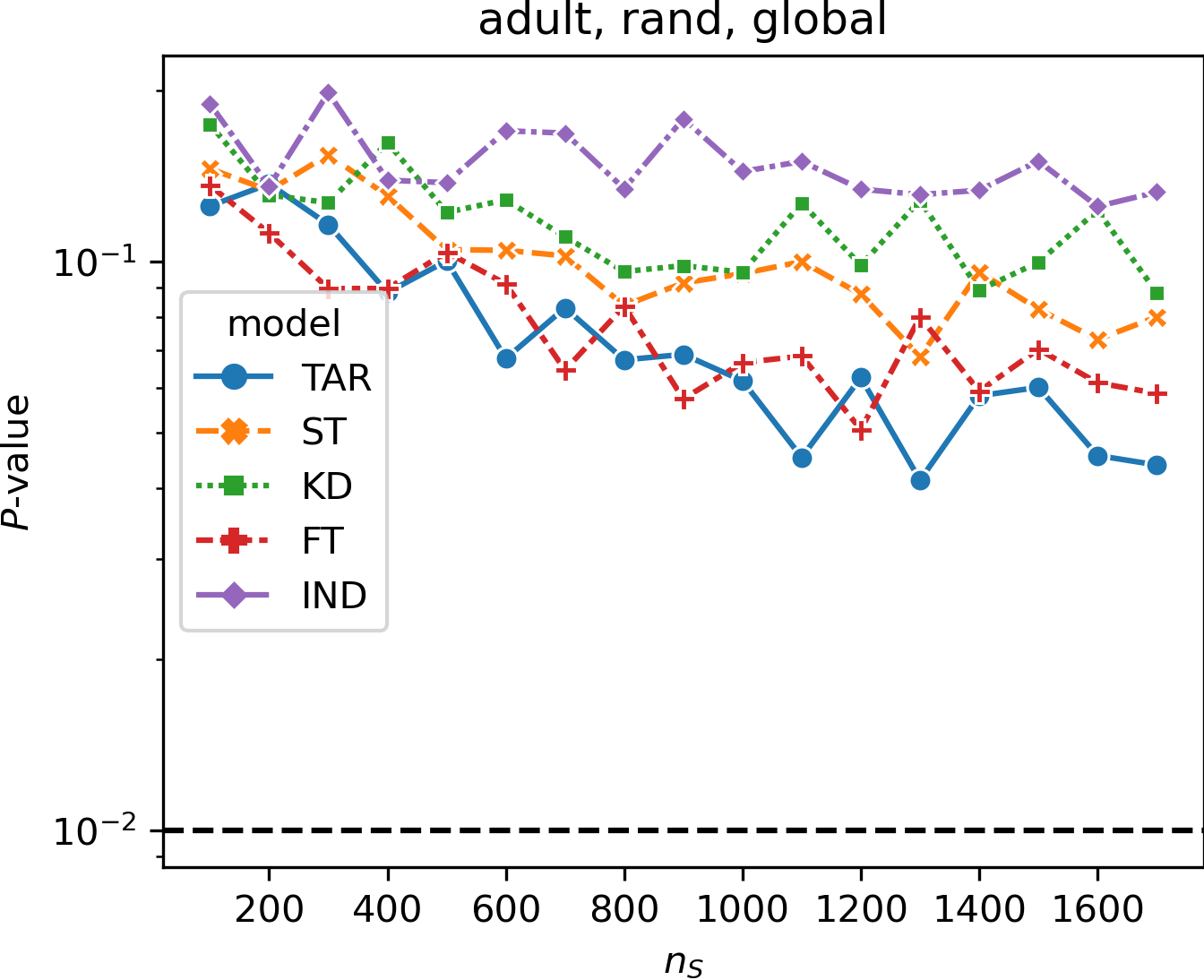}
        \label{subfig:adult_ns_pvalues_global}
    } 
    \subfigure[Adult, per-sample]{
        \includegraphics[width=0.48\columnwidth]{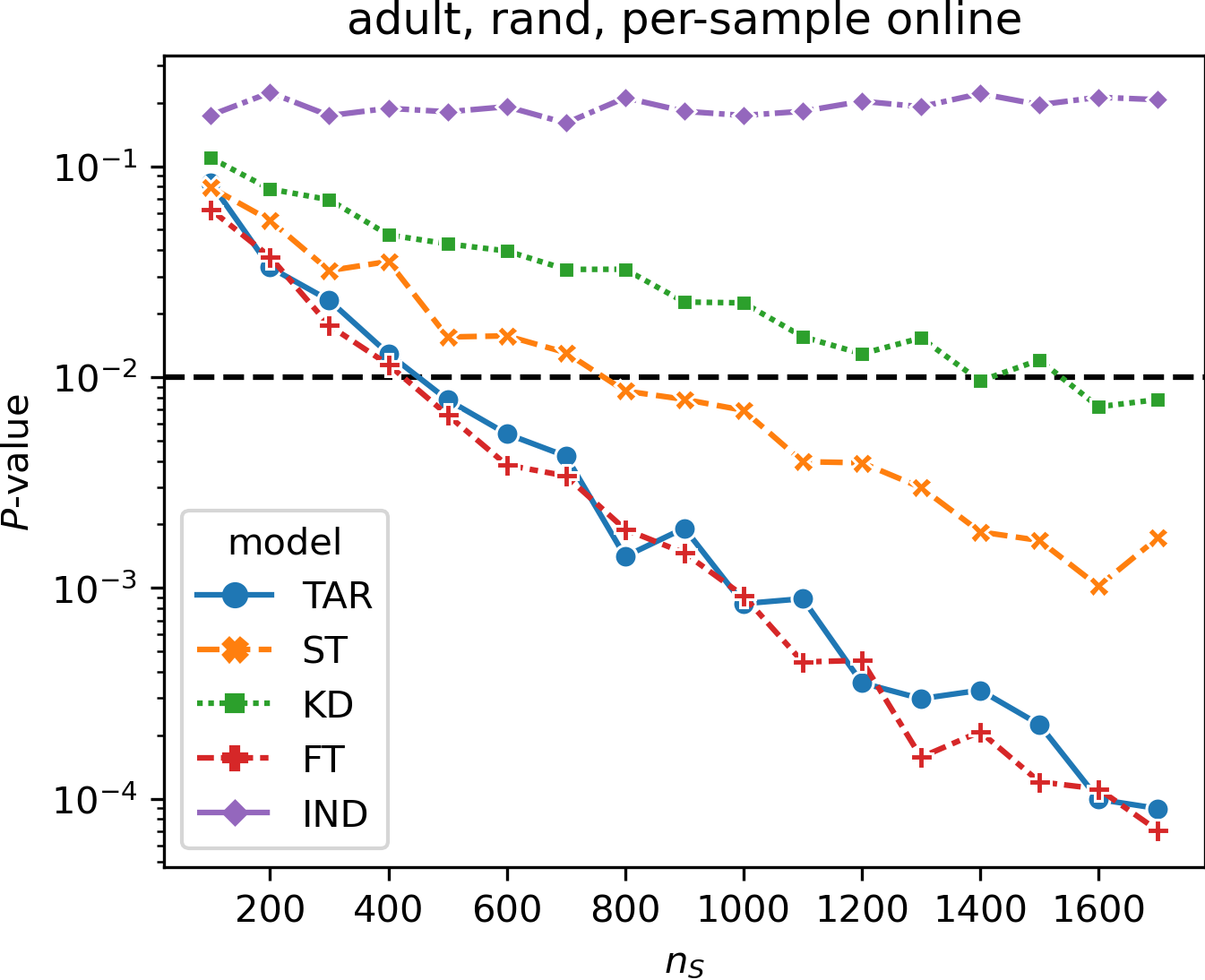}
        \label{subfig:adult_ns_pvalues_persample}
    } 
     \vspace{-3mm}
    \caption{p-value against the number of exposed training samples $n_{S}$. Black dotted line implies $\alpha =0.01$.}
    \label{fig:pvalues_ns}
\end{figure*}

\begin{figure*}[!t]
    \centering  
    \subfigure[CIFAR-10]{
        \includegraphics[width=0.47\columnwidth]{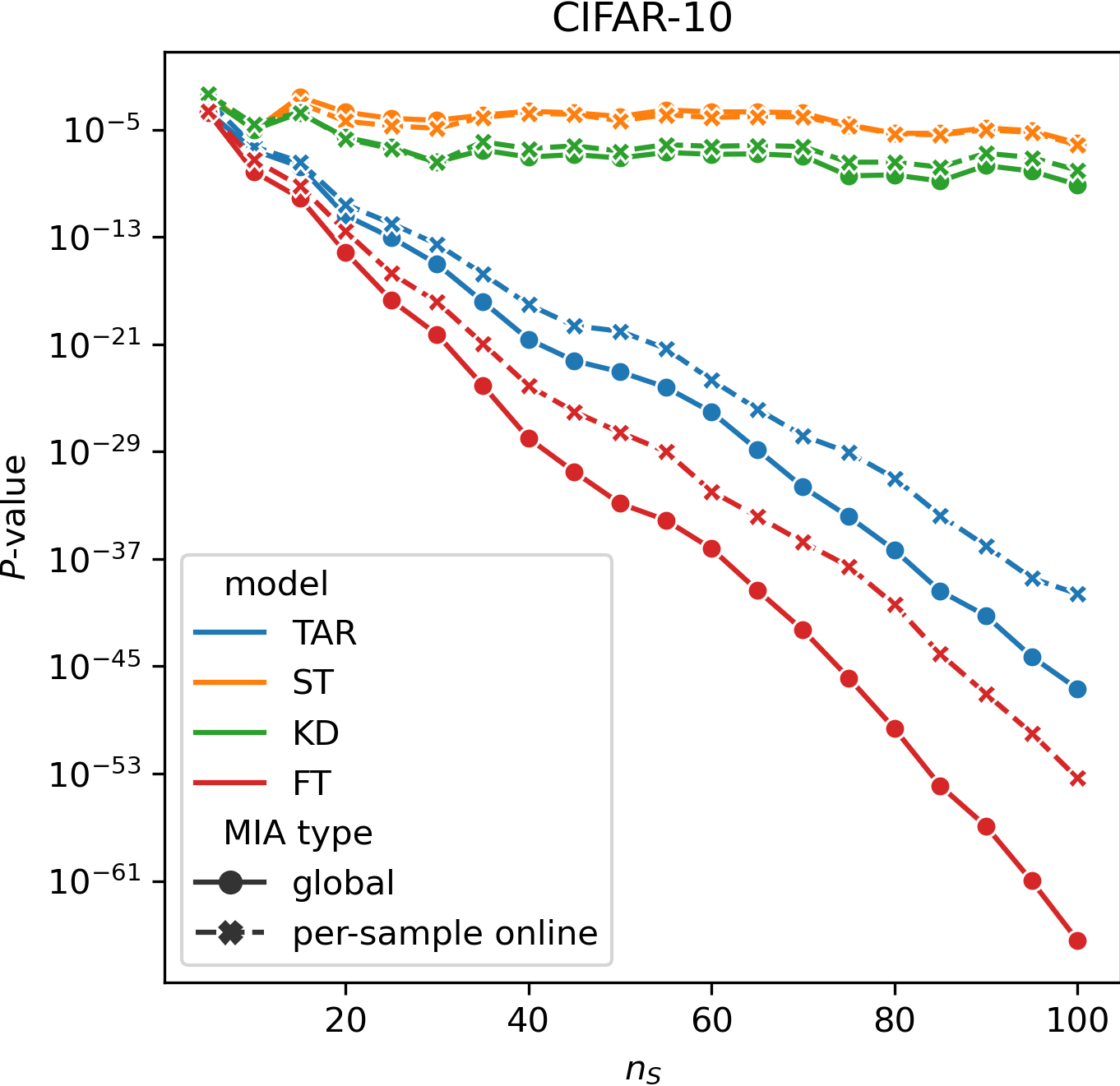}
        \label{subfig:cifar_lesspriv}
    }
    \subfigure[FMNIST]{
        \includegraphics[width=0.47\columnwidth]{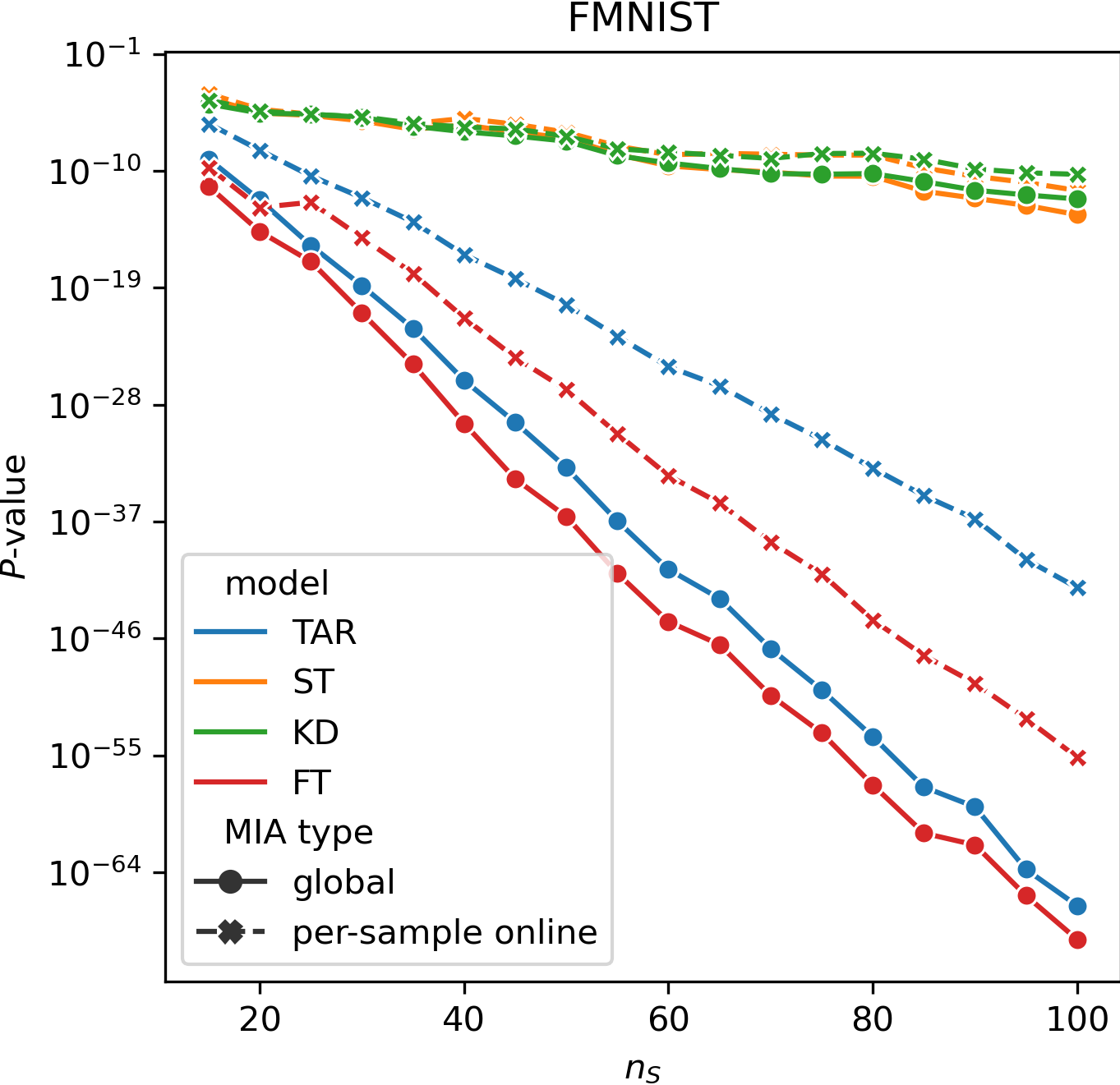}
        \label{subfig:fmnist_lesspriv}
    }
    \subfigure[Adult]{
        \includegraphics[width=0.47\columnwidth]{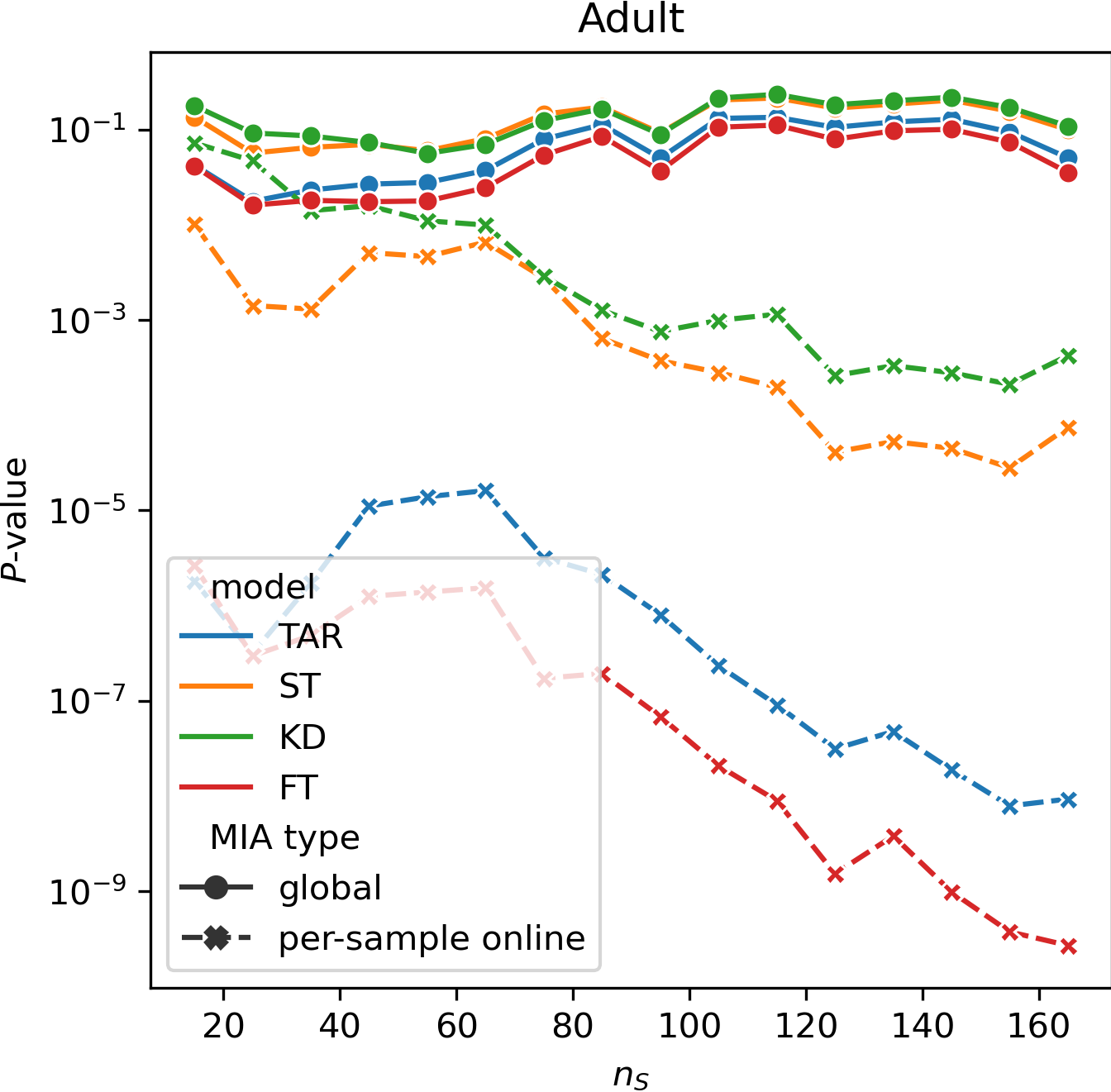}
        \label{subfig:adult_lesspriv}
    }
    \subfigure[Health]{
        \includegraphics[width=0.47\columnwidth]{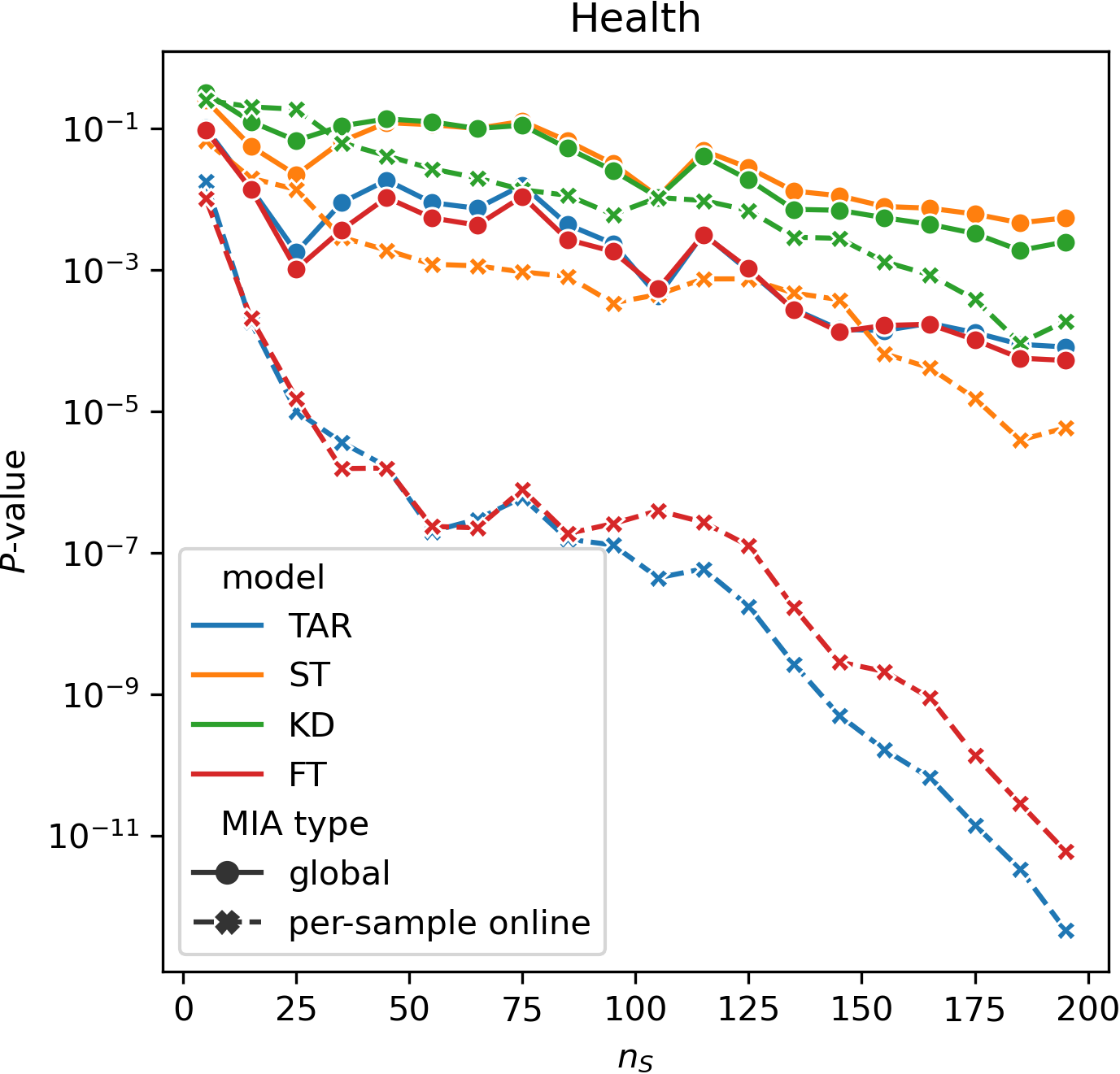}
        \label{subfig:health_lesspriv}
    }
 \vspace{-3mm}
    \caption{Comparison between the enhanced VeriDIP equipped with the global MI attack $\mathcal{V}_{\text{E-G}}$ (Dotted line with marker ``$\times$") and the enhanced VeriDIP equipped with the per-sample MI attack $\mathcal{V}_{\text{E-P}}$ (Solid line with marker ``$\cdot$'). }
    \label{fig:compare_MIA}
\end{figure*}

\subsubsection{Enhanced VeriDIP}

In this section, we evaluate the enhanced VeriDIP on four datasets and compare the results with those of the basic VeriDIP. Table~\ref{tab:enh} reports the minimum number of exposed training samples required to verify ownership at a significance level of $\alpha = 0.01$ (with $99\%$ confidence). Note that the p-values of all independent models remain at $10^{-1}$, and therefore, we have omitted the corresponding $n_{S}$ values for them.

To identify the less private data in advance, we train $N$ shadow models ($N=100$), where each model is trained by sampling half of the database. Consequently, for each data point, we have approximately $N/2$ models that include the data and $N/2$ models that exclude the data. We compute the loss difference $\eta(z)$ for each data point using Equation~\eqref{eq:pred_gap} and select the $k$ samples with the highest $\eta(z)$ values as the less private data.

\textbf{The enhanced VeriDIP offers superior performance compared to the basic VeriDIP.} For CIFAR-10 and FMNIST datasets shown in Table~\ref{tab:enh}, the enhanced VeriDIP equipped with both the global MI attacks and the per-sample MI attacks successfully verify the ownership of all target (``Tar") and stolen models (``ME", ``KD", and ``FT") by exposing only $5$ samples. In the case of more generalized models, such as Adult and Health, the number of exposed training samples is reduced to $\frac{1}{100}$-$\frac{1}{10}$ of the basic VeriDIP. It is worth noting that the enhanced VeriDIP equipped with the global MI attack fails to prove ownership for the Adult database. We believe this is because the global MI attack is not powerful enough to extract useful privacy leakage fingerprints in such generalized models. The main reasons for the success of the enhanced solution are:
\begin{itemize}[leftmargin=*]
	\item Leveraging the worst-case privacy leakage as the model fingerprint can significantly amplify the characteristics of the positive model that are different from the negative counterparts (see Figure~\ref{fig:lesspriv}); 
	\item The decision boundary for less private data is transferable (not easy to erase) in the process of model stealing. 
\end{itemize}

We then compare the performance of the enhanced VeriDIP equipped with the global MI attack (denoted as $\mathcal{V}_{\text{E-G}}$) with the enhanced VeriDIP equipped with the per-sample MI attack (denoted as $\mathcal{V}_{\text{E-P}}$) and plot the p-value against $n_{S}$ in Figure~\ref{fig:compare_MIA}.

\textbf{Compared with the basic VeriDIP where $\mathcal{V}_{\text{P}}$ is superior to $\mathcal{V}_{\text{G}}$ for all tasks, the behavior of $\mathcal{V}_{\text{E-P}}$ and $\mathcal{V}_{\text{E-G}}$ is more complex in enhanced VeriDIP.} For instance, in Figure~\ref{subfig:cifar_lesspriv} and Figure~\ref{subfig:fmnist_lesspriv}, $\mathcal{V}_{\text{E-G}}$ shows surprisingly better performance than $\mathcal{V}_{\text{E-P}}$, but the opposite is true for the Health and Adult databases. Particularly for the Adult database (see Figure~\ref{subfig:adult_lesspriv}), $\mathcal{V}_{\text{E-G}}$ fails to identify all positive models. Investigating the attack ability of MI attacks on different types of databases is beyond the scope of this work. However, we can conclude that the enhanced VeriDIP equipped with the global MI attack is more than sufficient to prove ownership of models trained on CIFAR-10 and FMNIST databases. For models that are barely overfitted, such as those trained on the Adult and Health databases, the enhanced VeriDIP equipped with the per-sample MI attack is a better choice.

\subsubsection{Comparisons with State-of-the-art} 
Dataset Inference (DI)~\cite{maini2021dataset} is the most similar to our idea, but differs in terms of model fingerprint extraction methods. Therefore, we compare our verification performance and costs with DI both functionally and experimentally. The result are show in Table~\ref{tab:func_cmpr} and Table~\ref{tab:exp_cmpr}. We summarize the results in the following aspects:

First, VeriDIP is applicable to tabular trained DNN models, while DI is not. DI uses adversarial noise as fingerprints, but finding the adversarial noise is not trivial for models trained on tabular data. Tabular data may contain a combination of continuous, discrete, and categorical features, making it difficult to calculate adversarial noise through gradient descent. VeriDIP, on the other hand, only requires querying the DNN model's prediction probability, making it applicable to all classifiers.

\begin{table}[!t]
	\centering
	\caption{Functional comparison with Dataset Inference~\cite{maini2021dataset}.}
	\vspace{-3mm}
	\label{tab:func_cmpr}%
	\begin{tabular}{cccc}
		\toprule
		& \makecell{Immune to \\ detector attack} & \makecell{Support table-\\trained models} & \makecell {Directly \\link to DP}\\
        \midrule
		DI & no & no & no \\
		Ours & yes & yes & yes \\
		\bottomrule
	\end{tabular}%
\end{table}%

\begin{table}[!t]
	\centering
	\caption{Experimental comparison with Dataset Inference~\cite{maini2021dataset}. $\mathcal{V}_{\text{E-G}}$: the enhanced VeriDIP equipped with the global MI attack; $\mathcal{V}_{\text{E-P}}$: the enhanced VeriDIP equipped with the per-sample MI attack.}
  \vspace{-3mm}
	\label{tab:exp_cmpr}%
	\begin{tabular}{ccccc}
		\toprule
		Database & OT algorithm & $n_{S}$ & No. of queries   & p-value \\
		\midrule
		\multirow{3}{*}{\rotatebox[origin=c]{0}{CIFAR-10}} & DI & 10 & 10*20*50 & $10^{-6}$ \\
		&$\mathcal{V}_{\text{E-G}}$ & 10 & 10& $10^{-7}$ \\
		&$\mathcal{V}_{\text{E-P}}$ & 10 & 10& $10^{-6}$ \\
		\midrule
		\multirow{3}{*}{\rotatebox[origin=c]{0}{FMNIST}} & DI & 10 & 10*30*50 & $10^{-5}$ \\
		&$\mathcal{V}_{\text{E-G}}$ & 10 & 10& $10^{-9}$ \\
		&$\mathcal{V}_{\text{E-P}}$ & 10 & 10& $10^{-6}$ \\
		\bottomrule
	\end{tabular}%
\end{table}%
Second, compared to DI, VeriDIP significantly reduces the number of required queries during ownership verification, making it immune to the detector attack~\cite{hitaj_have_2018}. DI requires querying the suspect model $n_{S} \times n_{\text{adv}} \times T$ times to obtain a model fingerprint. However, this can raise suspicion from pirated APIs, leading to refusals to answer or adding noise to the responses. Here, $n_{S}$ denotes the number of exposed training samples, $n_{\text{adv}}$ is the number of repeated adversarial attacks per sample, and $T$ is the number of queries for one adversarial attack. In the original setting of~\cite{hitaj_have_2018}, $n_{\text{adv}}=30$ and $T=50$. Table~\ref{tab:exp_cmpr} lists the experimental results for identifying target models in CIFAR-10 and FMNIST. We do not provide the results for Adult and Health datasets because DI does not support them. Consequently, VeriDIP achieves similar or better performance with significantly fewer exposed training samples (two orders of magnitude less than DI). 

Third, VeriDIP can be directly linked to the definition of DP, as the privacy leakage estimated by MI attacks serves as a lower bound for the privacy budget $\epsilon$ in DP (see analysis in Section~\ref{sec:bound_DP}). In contrast, DI leaves the connection to DP as an open question.

\subsubsection{Differential Privacy Relationship}
\label{sec:dp_DI}
In this section, we experimentally discuss the effectiveness of VeriDIP on DP machine learning models, which is also a remaining problem addressed in~\cite{maini2021dataset}. For this evaluation, we select the enhanced VeriDIP models $\mathcal{V}{\text{E-P}}$ and $\mathcal{V}{\text{E-G}}$ due to their improved performance.

\begin{table}[!t]
    \centering
	\caption{Hyper-parameters and test accuracy for DP models.  $\mathrm{z}$: noise multiplier, $C$:clipping threshold.}
  \vspace{-3mm}
	\label{tab:dp-DI}%
	\begin{tabular}{ccccc}
		\toprule
		Database & ($\epsilon$, $\delta$) & epoch  & $(C,\mathrm{z})$ & TestAcc  \\
		\midrule
		\multirow{2}{*}{\rotatebox[origin=c]{0}{CIFAR-10}}& $(1.0,10^{-5})$ & 60 & (5e-4,2.1) & $84.79\%$\\  
		&$(0.5,10^{-5})$ & 60 & (5e-4,4.1) & $84.49\%$\\
		\midrule 
		
		\multirow{2}{*}{\rotatebox[origin=c]{0}{FMNIST}}& $(1.0,10^{-5})$ & 19  & (5e-3,1.2) & $90.54\%$\\  
		& $(0.5,10^{-5})$ & 20 & (5e-3,1.9) & $90.00\%$\\
		
		\midrule 
		\multirow{2}{*}{\rotatebox[origin=c]{0}{Health}} & $(1.0,10^{-5})$ & 50 & (1e-3,4.9) & $86.97\%$ \\
		& $(0.5,10^{-5})$ & 50 & (1e-3,9.7) &  $86.92\%$ \\
		\midrule
		\multirow{2}{*}{\rotatebox[origin=c]{0}{Adult}} & $(1.0,10^{-5})$ & 70  & (1e-3,7.9) & $84.69\%$ \\
		& $(0.5,10^{-5})$ & 60  & (1e-3,14.9) & $84.73\%$\\
		\bottomrule
	\end{tabular}%
\end{table}%

\begin{table}[!t]
	\centering
	\caption{Verification performance of the enhanced VeriDIP on DP models. }
	 \vspace{-3mm}
	\label{tab:dp_per}%
	\begin{tabular}{cccccc}
		\toprule
	    \multirow{2}{*}{Datasets} & \multirow{2}{*}{Models}  & \multicolumn{2}{c}{$\epsilon=0.5$} & \multicolumn{2}{c}{$\epsilon=1.0$}\\ \cmidrule(r){3-6}
	&&  $n_{S}$ & p-value & $n_{S}$ & p-value  \\
		\midrule
		\multirow{4}{*}{\rotatebox[origin=c]{0}{CIFAR-10}} 
		& TAR & 5 & $10^{-4}$ &5 &$10^{-4}$\\
		& ME  & 5 & $10^{-3}$ &5 &$10^{-4}$\\
		& KD  & 5 & $10^{-3}$ &5 &$10^{-4}$ \\
		& FT  & -- & -- &-- &-- \\
		
		\midrule
		\multirow{4}{*}{\rotatebox[origin=c]{0}{FMNIST}} 
		& TAR & 5 & $10^{-6}$  &5 & $10^{-6}$ \\
		& ME  & 5 & $10^{-3}$ &5 &$10^{-3}$ \\
		& KD  & 5 & $10^{-3}$ &5 &$10^{-3}$ \\
		& FT  & 5 & $10^{-4}$ &5 & $10^{-6}$ \\
		
		\midrule
		\multirow{4}{*}{\rotatebox[origin=c]{0}{Adult}} 
		& TAR & 5  & $10^{-3}$ &5  & $10^{-4}$\\
		& ME  & 35 & $10^{-3}$ &25 & $10^{-3}$\\
		& KD  & 75 & $10^{-3}$ &55 & $10^{-3}$ \\
		& FT  & 15 & $10^{-3}$ &5  & $10^{-4}$ \\
		
		\midrule
		\multirow{4}{*}{\rotatebox[origin=c]{0}{Health}} 
		& TAR  &15  & $10^{-4}$& 15 & $10^{-4}$\\
		& ME   &175 &$10^{-3}$ & 55 & $10^{-3}$ \\
		& KD   &135 &$10^{-3}$ & 75 & $10^{-3}$  \\
		& FT   & 15 & $10^{-3}$& 5  & $10^{-3}$\\
		
		\bottomrule
	\end{tabular}%
\end{table}%

\begin{figure*}[!t]
    \centering  
    \subfigure[CIFAR-10, Global]{
        \label{subfig:cifar_dp_global}
        \includegraphics[width=0.23\textwidth]{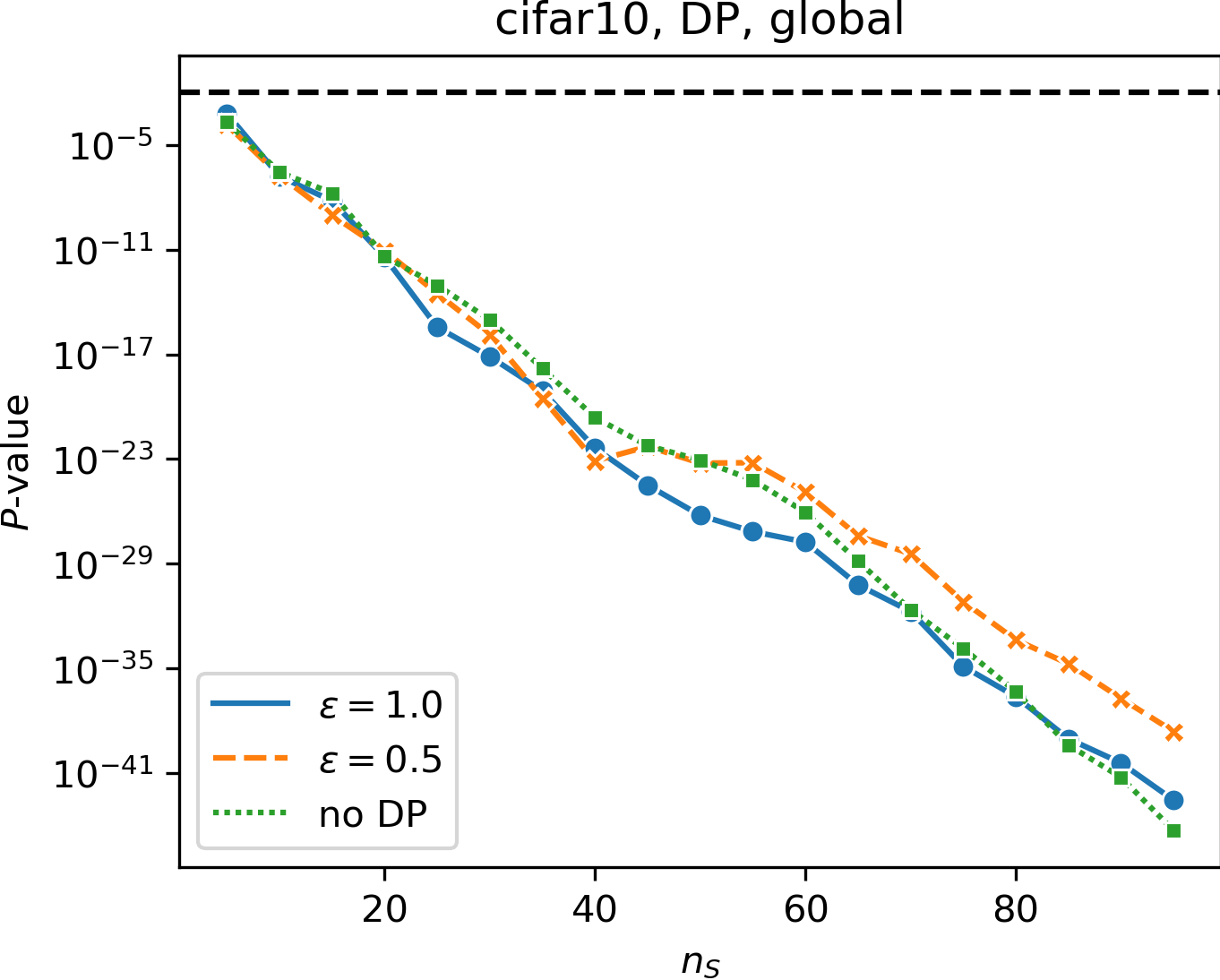}
    }
    \subfigure[CIFAR-10, Per-sample]{
        \label{subfig:cifar_dp_persample}
        \includegraphics[width=0.23\textwidth]{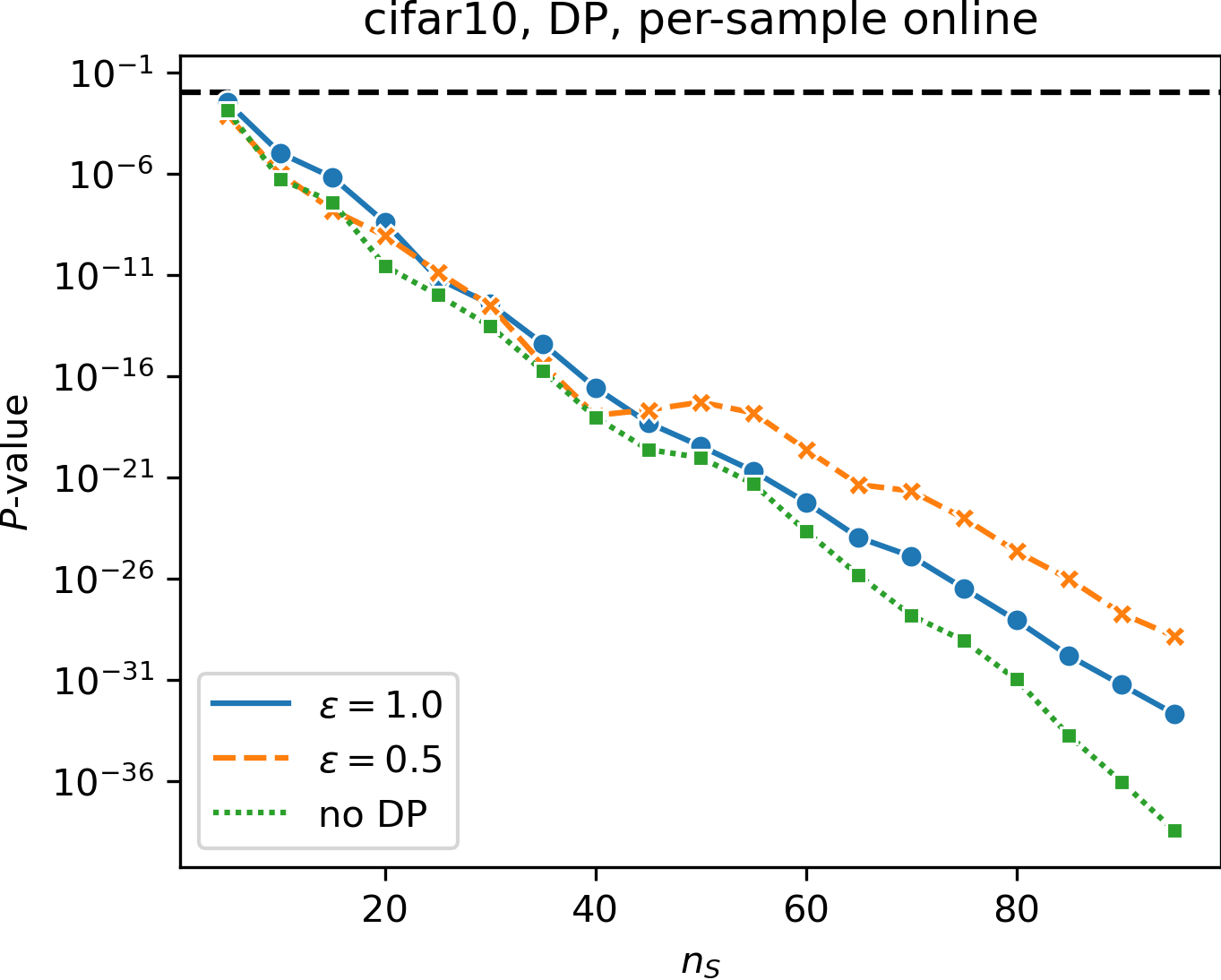}
    }
    \subfigure[FMNIST, Global]{
        \label{subfig:fmnist_dp_global}
        \includegraphics[width=0.23\textwidth]{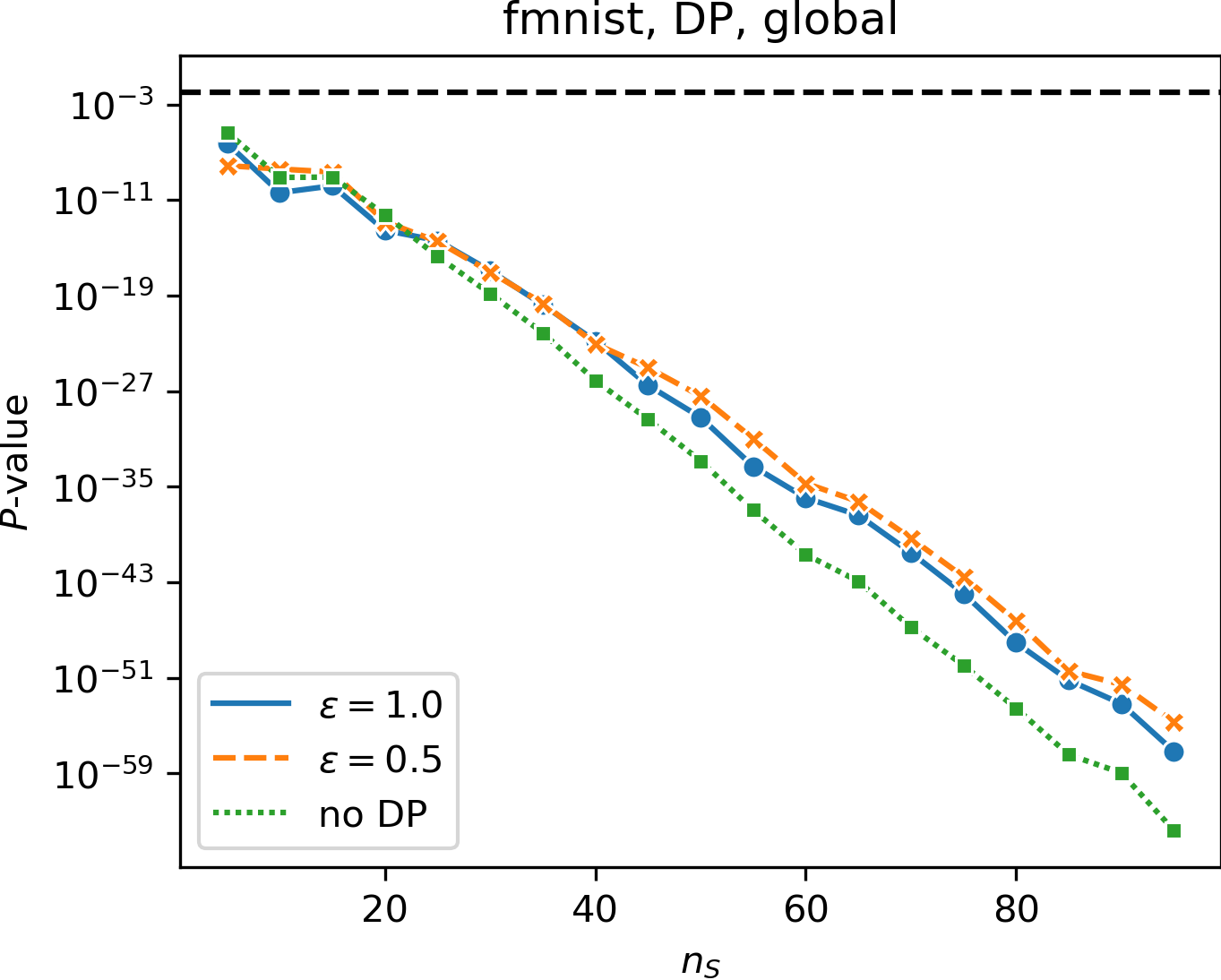}
    }
    \subfigure[FMNIST, Per-sample]{
        \label{subfig:fmnist_dp_persample}
        \includegraphics[width=0.23\textwidth]{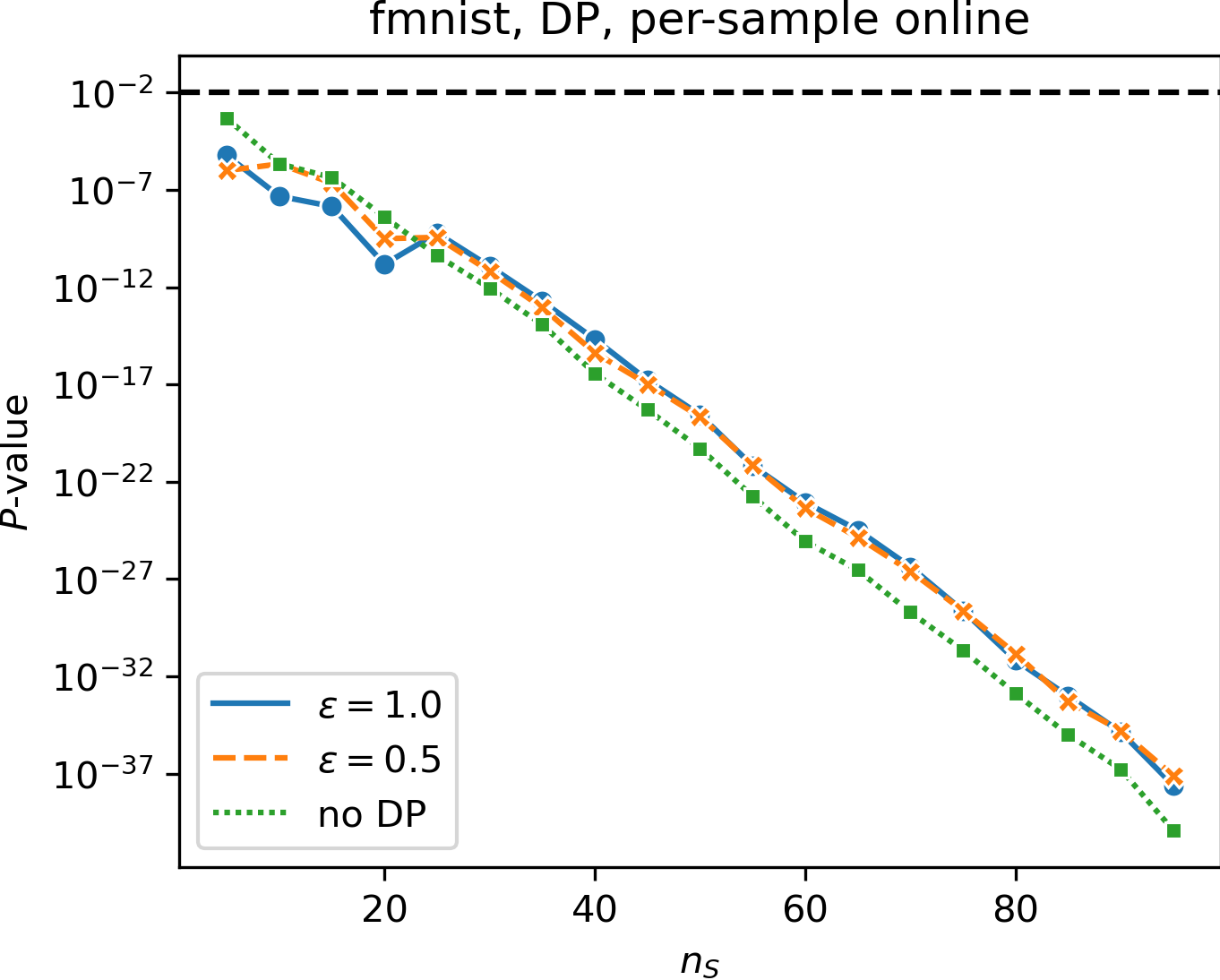}
    }
    \subfigure[Health, Global]{
        \label{subfig:health_dp_global}
        \includegraphics[width=0.23\textwidth]{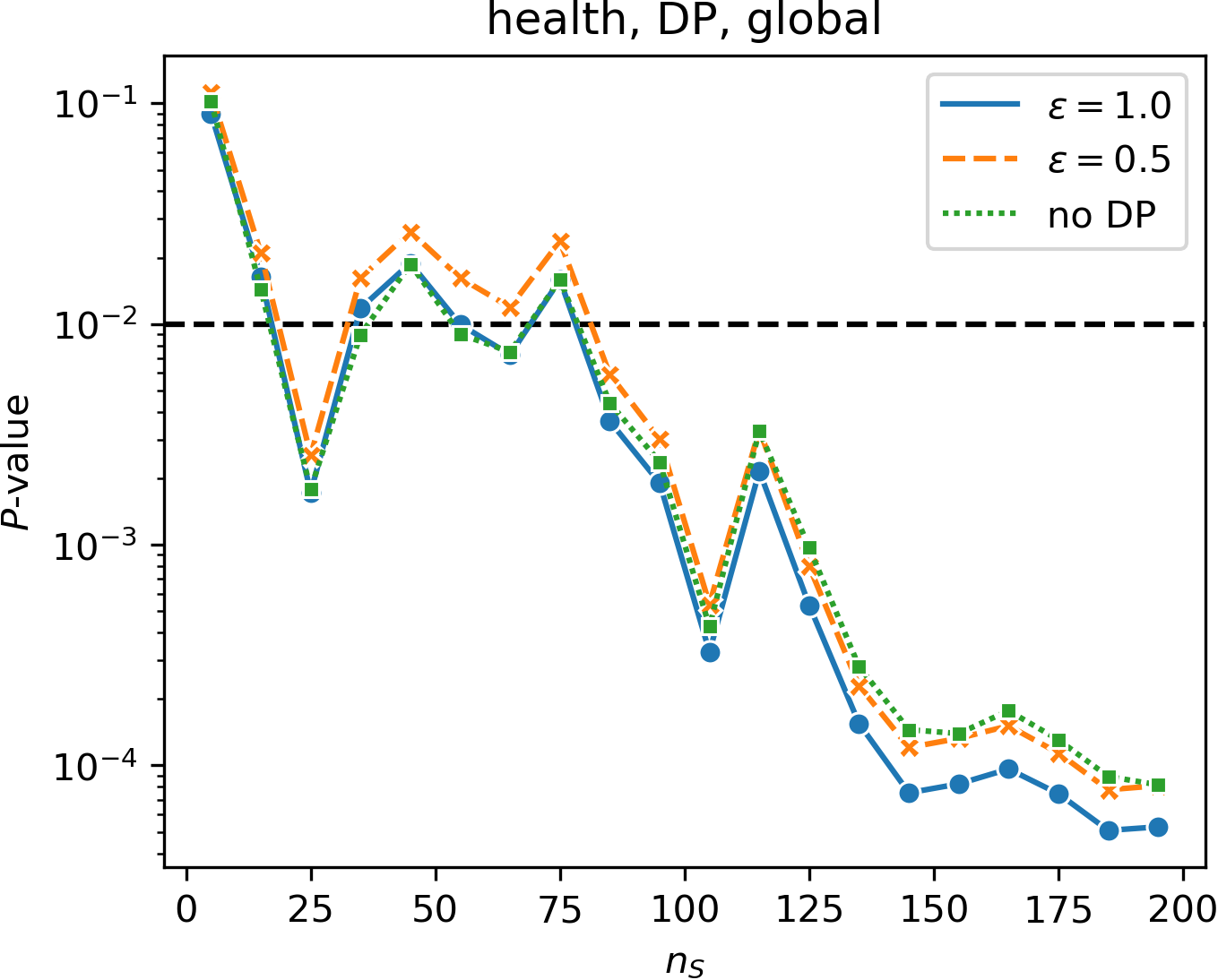}
    }
    \subfigure[Health, Per-sample]{
        \label{subfig:health_dp_persample}
        \includegraphics[width=0.23\textwidth]{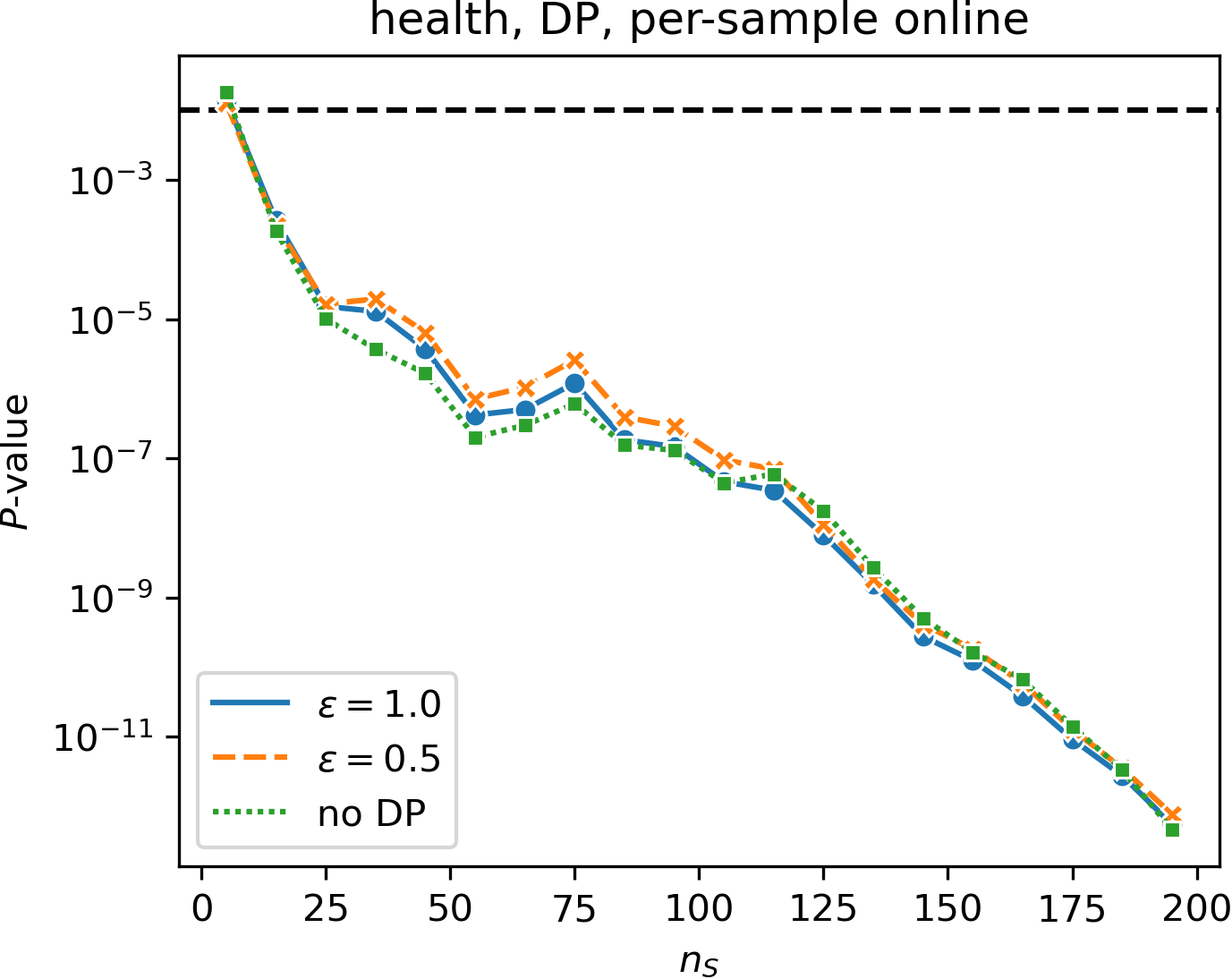}
    }
    \subfigure[Adult, Per-sample]{
        \label{img:adult_dp_persample}
        \includegraphics[width=0.23\textwidth]{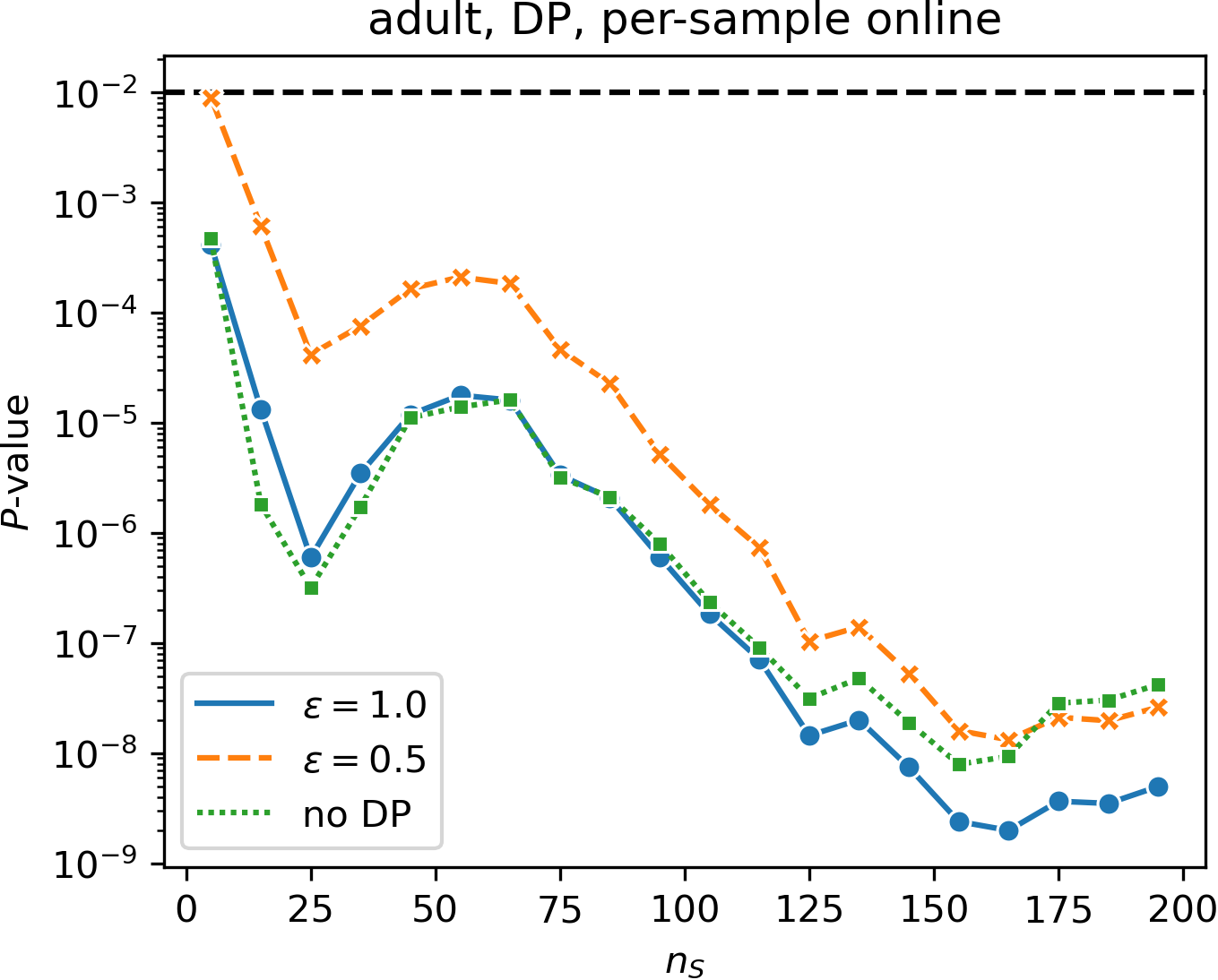}
    }
\vspace{-3mm}
    \caption{Performance of The Enhanced VeriDIP $\mathcal{V}_{\text{E-G}}$ and $\mathcal{V}_{\text{E-P}}$ on DP IP-protected models.}
    \label{fig:dp_performance}
\end{figure*}

\textbf{Experiment setup.} We use the DP Adam optimizer~\cite{abadi2016deep} to train DP machine learning models and compose the privacy budget using RDP techniques~\cite{mironov2017renyi}. In each iteration, we first clip gradient norm with the threshold $C$, then add Gaussian noise with scale $\sigma = \mathrm{z} * C$ (see Table~\ref{tab:dp-DI}) where $\mathrm{z}$ stands for the noise multiplier. We adjust different pairs of hyper-parameters $(C,\mathrm{z})$ to trade off privacy vs. utility. For each dataset, we choose two privacy budget options for $(\epsilon,\delta)$, such that $(0.5,10^{-5}) $ and $(1.0, 10^{-5})$, where $\delta$ is usually set to be the inverse of the number of training sets, as shown in~\cite{abadi2016deep}. These options are commonly used in training DP machine learning models. A smaller privacy budget $\epsilon$ indicates a higher privacy protection level (yet lower model utility). The hyper-parameters that are related to training DP models and testing the accuracy of DP models are listed in Table~\ref{tab:dp-DI}. Note that, the configuration of model stealing attacks are identical to the former's (see Section~\ref{sec:exp_setup}).

Recall the theoretical analysis in Section~\ref{sec:bound_DP}, we bound the privacy budgets align with the VeriDIP's performance, for instance, $\epsilon =0.1$ result in $P>0.156$. Thus, we first experiment with $\epsilon =0.1$ and find all DP models experienced a substantial loss in functionality. Particularly for CIFAR-10, the $(0.1,10^{-5})$-DP model achieved only $76.71\%$ test accuracy, compared with the non-DP benchmark, it loses approximately $10\%$ of the accuracy. In accordance with the theoretical analysis, none of these models can be verified for ownership using VeriDIP. However, protecting the copyright of DP models becomes less meaningful without preserving utility, which motivated us to focus on evaluating the effectiveness of VeriDIP on more useful DP models. Based on our analysis, when $\epsilon = 0.5$, the limitation on the p-value is already negligible. We then experiment with $\epsilon =0.5$ and $\epsilon =1.0$ and Table~\ref{tab:dp_per} presents the main result for VeriDIP on $(0.5,10^{-5})$-DP and $(1.0,10^{-5})$-DP models. Additionally, Figure~\ref{fig:dp_performance} illustrates the comparisons of p-values against $n_{S}$ curves for these DP models and non-DP models. Note that the fine-tuning attack~\cite{chen_refit_2021} fails to steal a functionally-preserving DNN model trained with Adam optimizer, which is why the fourth row of CIFAR-10 is empty. 

\textbf{VeriDIP is as effective on utility-preserving DP models as it is on non-DP models}. Comparing the model utility presented in Table~\ref{tab:dp-DI} and Table~\ref{tab:gen_error}, we found that, by carefully choosing DP hyper-parameters, all DP models show comparable utility with non-DP baselines. From Table~\ref{tab:dp_per} and Figure~\ref{fig:dp_performance}, we can see that 
the effectiveness of $\mathcal{V}_{\text{E-G}}$ and $\mathcal{V}_{\text{E-P}}$ on CIFAR-10 and FMNIST are hardly affected by the noise injected by DP. While on Adult and Health datasets, more strict privacy protection may increase the number of exposed training samples. In Table~\ref{tab:dp_per}, the number of exposed samples $n_{S}$ of $(0.5,10^{-5})$-DP models is higher than that of $(1.0,10^{-5})$-DP models. This indicates that there is a trade-off between privacy protection and copyright protection, especially for those barely overfitted models. 

Since there is a subtle balance between privacy protection and copyright protection in generalized models, we study the behavior of VeriDIP varying different DP hyper-parameters for Adult and Health datasets. In particular, We study two types of DP hyper-parameters: DP clipping threshold $C$ and the number of training epochs, and analyze their influence on VeriDIP.

\textbf{(1) DP clipping threshold $C$.} $C$ represents the clipping threshold for batch gradients in each training iteration. We conducted experiments with different values of $C$ as it does not affect the value of $\epsilon$ but impacts the training performance. We kept the noise multiplier $\mathrm{z}$ and the number of training epochs fixed for $\epsilon = 0.5$. The p-value against $n_{S}$ curve comparisons are depicted in Figure~\ref{fig:compare_cg_factor}. From the figures, we observe that certain choices of $C$ lead to the failure of VeriDIP, such as $C=10^{-1}$, $C=10^{-2}$, and $C=10^{-5}$ in Figure~\ref{subfig:adult_dp_lesspriv_cg_factor}, and $C=10^{-1}$ in Figure~\ref{subfig:health_dp_lesspriv_cg_factor}. Excessively large or small values of $C$ have a detrimental effect on the effectiveness of VeriDIP. A large $C$ introduces excessive noise due to the noise scale $\sigma = \mathrm{z} * C$. Conversely, a small $C$ restricts the gradient magnitude in each iteration, thereby affecting the model's learning process. Hence, we encourage model owners to explore various choices of $C$ to determine the optimal value when training a DNN model with both privacy protection and copyright protection.

\textbf{(2) Number of training epochs.} In addition to $C$, the model trainer has two options to achieve the same privacy protection: (a) more training epochs but less noise for each iteration. (b) less training epochs but more noise for each iteration. Thus, we compare these options and the results are shown in Figure~\ref{fig:compare_epoch_factor}. As a result, we find that option (a) has better VeriDIP performance for the DNN models than option (b). 

\begin{figure}[!t]
    \centering  
    \subfigure[Adult, $(0.5,10^{-5})$-DP]{
        \includegraphics[width=0.47\columnwidth]{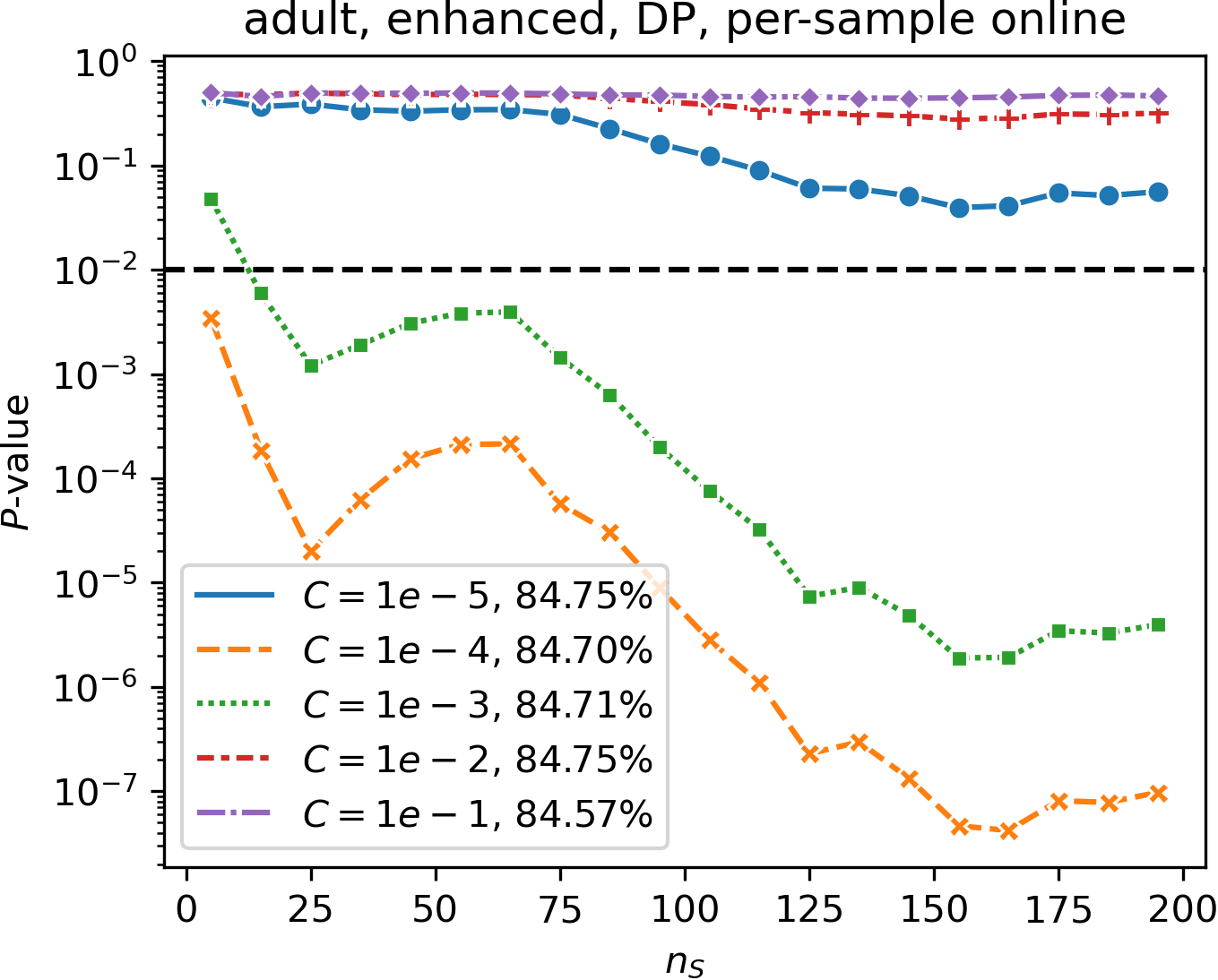}
        \label{subfig:adult_dp_lesspriv_cg_factor}
    }
    \subfigure[Health, $(0.5,10^{-5})$-DP]{
        \includegraphics[width=0.47\columnwidth]{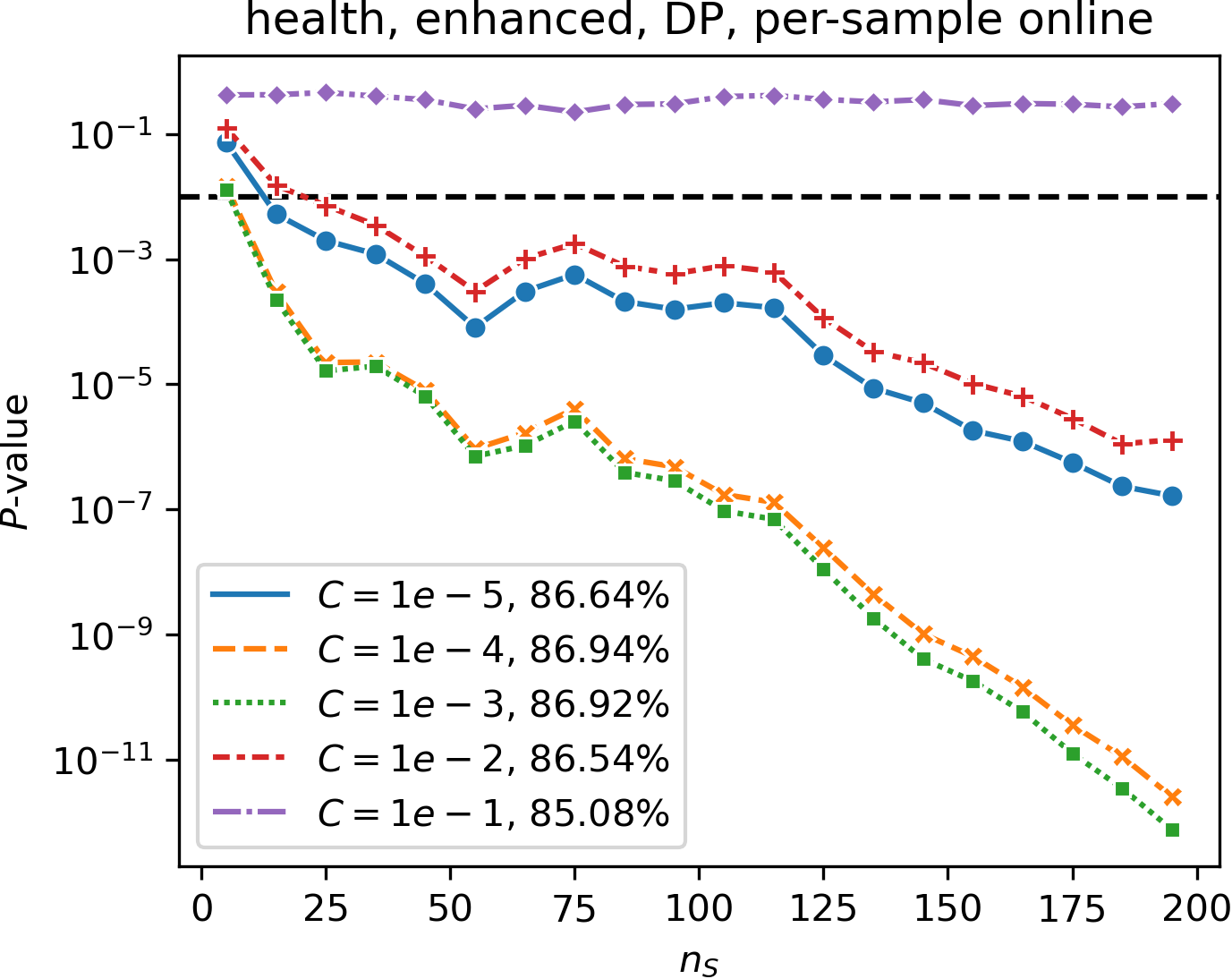}
        \label{subfig:health_dp_lesspriv_cg_factor}
    }
 \vspace{-4mm}
    \caption{The DP hyper-parameters $C$ impact VeriDIP's performance.}
    \label{fig:compare_cg_factor}
\end{figure}

To summarize, the enhanced VeriDIP is effective on DP-protected DNN models. Some privacy-preserving models may double or triple the number of exposed training samples in VeriDIP as a trade-off. Besides, carefully selecting the DP hyperparameters is crucial for model owners to simultaneously benefit from privacy protection and copyright protection.

%% file: conclusion.tex
\section{Conclusion and Future Work Directions}
\textbf{Conclusion of This Paper.} The increasing prevalence of model-stealing attacks poses a significant threat to the protection of neural network models' copyrights. In this work, we propose a novel ownership testing framework for DNN models, VeriDIP, along with its enhanced version, to combat model plagiarism. VeriDIP leverages privacy leakage as a natural fingerprint for verifying DNN model ownership. The enhanced VeriDIP utilizes a reduced amount of private data to estimate the worst-case privacy leakage of models, serving as enhanced model fingerprints. Our comprehensive experiments demonstrate that the enhanced VeriDIP achieves a true positive rate of $100\%$ and a false positive rate of $0$ in accurately identifying positive models (victim models and their stolen copies) as opposed to negative models (independent models), requiring a minimum of 5 data samples during the verification process. Furthermore, the enhanced VeriDIP effectively addresses an open problem concerning the protection of the copyright of any utility-preserved differentially private models.

\textbf{Future Work Directions.} We list the following potential future work directions for this paper.
\begin{enumerate}[leftmargin=*]
    \item Quantitative standard for the Number of Shadow Models Required. In this paper, in order to identify less private data for the enhanced VeriDIP, we trained $100$ shadow models for each mentioned dataset. It is important to note that this empirical number of shadow models may vary depending on the specific datasets. Therefore, it would be valuable to propose a quantitative standard for determining the appropriate number of shadow models based on the characteristics of the given datasets.

    \item Extending to other data domains. While our study primarily focuses on image and tabular data, future research can explore the applicability of VeriDIP to other data types and domains. This could include natural language processing, audio data, or even more specialized domains such as genomics or finance.

    \item Efficiency improvement. Future work can focus on enhancing the efficiency of the VeriDIP framework by reducing the computation costs associated with finding less private data. These efforts will contribute to minimizing the computational overhead and making the framework more practical for real-world deployment.
\end{enumerate} 

\begin{figure}[!t]
    \centering 
    \subfigure[Adult, $(0.5,10^{-5})$-DP]{
        \includegraphics[width=0.47\columnwidth]{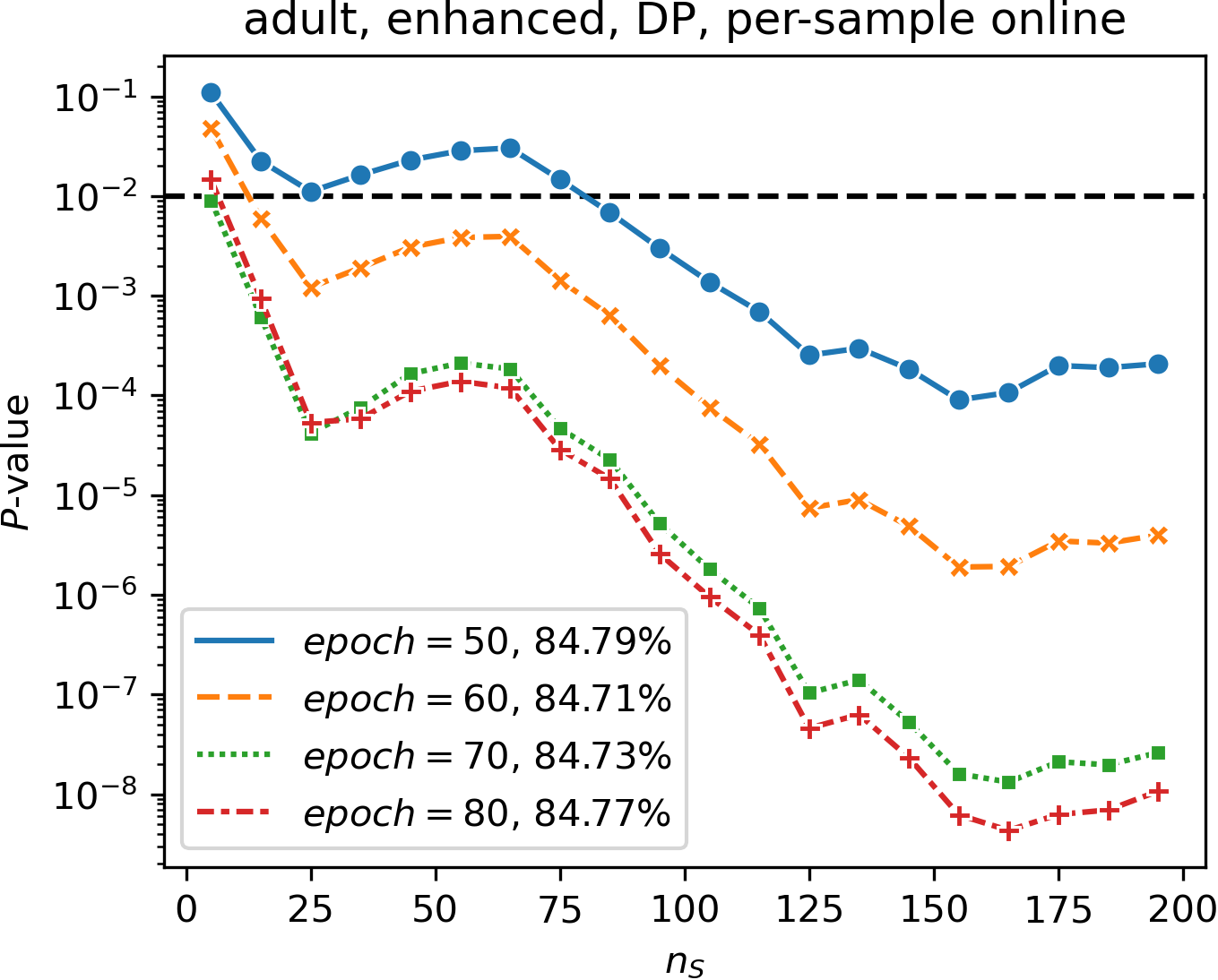}
        \label{subfig:adult_dp_lesspriv_epoch_factor}
    }
    \subfigure[Health, $(0.5,10^{-5})$-DP]{
        \includegraphics[width=0.47\columnwidth]{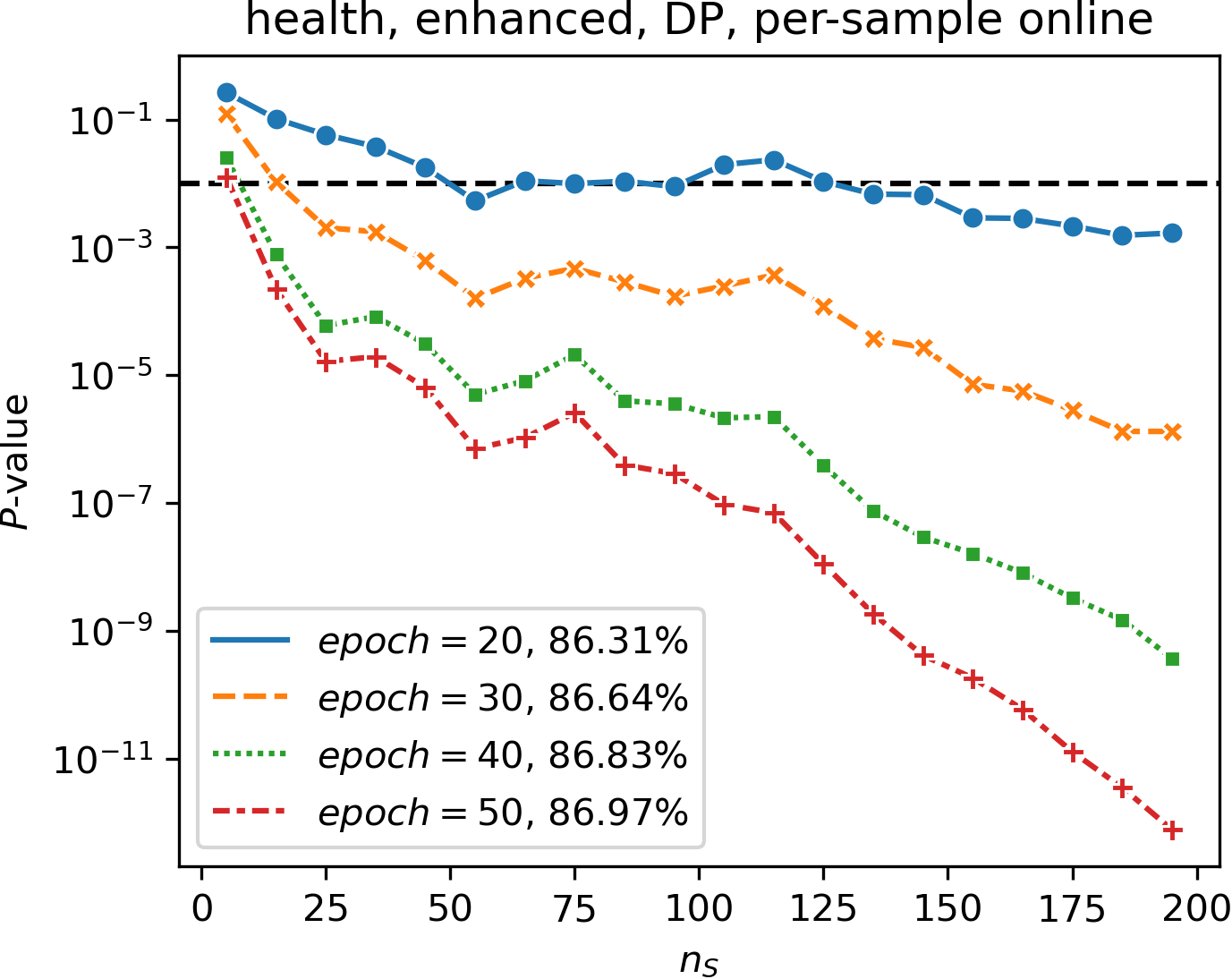}
        \label{subfig:health_dp_lesspriv_epoch_factor}
    }
 \vspace{-4mm}
    \caption{Training epochs of DP models impact VeriDIP's performance.}
    \label{fig:compare_epoch_factor}
\end{figure}